\documentclass[aps,prl,twocolumn,superscriptaddress,floatfix,nobibnotes,nofootinbib,longbibliography]{revtex4-2}

\usepackage{amsmath,amssymb,amsfonts,amsthm}
\usepackage{graphicx}
\usepackage{physics}
\usepackage{hyperref}
\hypersetup{colorlinks=true, linkcolor=blue, citecolor=blue, urlcolor=blue}
\usepackage{xcolor}
\usepackage{bm}
\usepackage{bbm}

\newtheorem{theorem}{Theorem}
\newtheorem{corollary}[theorem]{Corollary}

\newcommand{\Dg}{\hat{\mathcal{D}}_\gamma}
\newcommand{\Kz}{\hat{K}_z}

\newcommand{\Kp}{\hat{K}_+}
\newcommand{\Km}{\hat{K}_-}
\newcommand{\Kphi}{\hat{K}_\phi}
\newcommand{\Veps}{\hat{V}_\varepsilon}
\newcommand{\Hstab}{\hat{H}_{\mathrm{stab}}}
\newcommand{\Ntrunc}{N_{\mathrm{trunc}}}

\newcommand{\HH}{\mathcal{H}}
\newcommand{\DD}{\mathbb{D}}
\newcommand{\PP}{\mathbb{P}}
\newcommand{\id}{\mathbbm{1}}
\newcommand{\SU}{\mathrm{SU}}
\newcommand{\PSU}{\mathrm{PSU}}

\newcommand{\GL}{\mathrm{GL}}
\newcommand{\Aut}{\mathrm{Aut}}

\begin{document}
\raggedbottom

 \title{Bosonic codes from compact phase spaces}

\author{David Roberts$^{\ast,\dagger}$}
\affiliation{Extropic Corporation, Cambridge, Massachusetts 02458, USA}
\affiliation{Joint Quantum Institute, University of Maryland, College Park, MD 20742, USA}
\affiliation{Joint Center for Quantum Information and Computer Science, NIST/University of Maryland, College Park, MD 20742, USA}

\author{Aaron Slipper$^{\ast}$}
\affiliation{Department of Mathematics, Duke University, Durham, NC 27708, USA}

\author{Alireza Parhizkar}
\affiliation{Joint Quantum Institute, University of Maryland, College Park, MD 20742, USA}

\author{Victor V.\ Albert}
\affiliation{Joint Center for Quantum Information and Computer Science, NIST/University of Maryland, College Park, MD 20742, USA}

\author{Mohammad Hafezi}
\affiliation{Joint Quantum Institute, University of Maryland, College Park, MD 20742, USA}
\affiliation{Joint Center for Quantum Information and Computer Science, NIST/University of Maryland, College Park, MD 20742, USA}

\date{\today}

\begin{abstract}

We present the algebraic structure of bosonic quantum error-correcting codes on genus-two Riemann surfaces. We explicitly construct the code words as automorphic forms and analytically generate the full tower of code spaces at all weights. We prove a fundamental no-go theorem: for any genus greater than one, the stabilizer group is non-amenable, forcing a strictly positive spectral gap in the stabilizer Hamiltonian. Consequently, no normalizable quantum state can satisfy all stabilizer conditions. This sharply contrasts with standard Gottesman--Kitaev--Preskill (GKP) codes, where the amenability of the stabilizer group $\mathbb{Z}^2$ permits approximate code words with arbitrary precision.
\end{abstract}

\maketitle
{\renewcommand{\thefootnote}{\fnsymbol{footnote}}%
\footnotetext[1]{These authors contributed equally to this work.}%
\footnotetext[2]{Contact author: david.b.roberts@outlook.com}}

{\it Introduction}. The Gottesman--Kitaev--Preskill (GKP) code
protects quantum information in a single bosonic mode by imposing a
lattice of displacement symmetries on phase
space~\cite{GKP2001}.  The quotient of the phase plane by this
lattice is a torus (an elliptic curve), and the code words are theta
functions on it, holomorphic sections of a line bundle arising from
geometric quantization, a perspective developed
in~\cite{Roberts2024,Conrad2024,MayrandRoyer2026}. Crucially, the torus is the \emph{phase
space} of the oscillator, not a configuration space: the elliptic
curve emerges physically as the image of the coherent-state family
projected into the code space~\cite{SM}.  Because the phase
space is compact, the code space is finite-dimensional, with its
dimension governed by the Riemann--Roch theorem.  This raises a sharp
question: what bosonic codes arise when the phase space is a compact
surface of genus $g \geq 2$?

Here we address this question for compact hyperbolic Riemann surfaces,
with genus two as the first nontrivial case.  Such a surface is the
hyperbolic disk $\DD$ folded up by a discrete group $\Gamma$.  The
symmetries of the disk form the
group $\PSU(1,1)$, whose unitary action on two bosonic modes is
by the quadratic Gaussian transformations of $\SU(1,1)$. Taking
$\Gamma$ as the stabilizer group therefore makes the stabilizers
\emph{squeezing} transformations, the hyperbolic counterpart of the
displacements of GKP.  The code words are automorphic forms on~$\DD$,
their number fixed by the Riemann--Roch theorem.

We report two complementary results.  Both the stabilizers and the
logical gates of these two-mode codes
are Gaussian operations (squeezing and rotations of two bosonic
modes), yet their logical gate group is the underlying surface's automorphism
group, and by Greenberg's theorem~\cite{Jones2019} \emph{every}
finite group arises this way.  Any finite group
can therefore serve as the logical gate group, in just two modes.  For
the Bolza surface~\cite{Bolza1887}---the most symmetric genus-two curve---it is the 48-element single-qubit Clifford group $\GL(2,\mathbb{F}_3)$,
which we compute explicitly.

\begin{figure}[t]
  \includegraphics[width=\columnwidth]{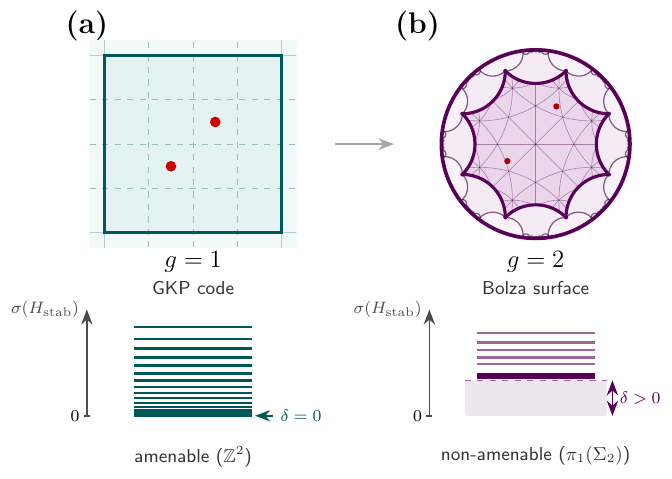}
  \caption{\label{fig:phase_spaces}%
    \textbf{Bosonic codes from compact phase spaces.}
    Top: torus ($g=1$, GKP code) and hyperbolic octagon
    ($g=2$, Bolza surface) with red dots marking stellar zeros of
    representative code states.
    Bottom: spectrum of the stabilizer Hamiltonian~$\Hstab$.
    For $g=1$ the surface group $\mathbb{Z}^2$ is amenable and the
    spectrum touches zero ($\delta=0$), permitting approximate code
    words.  For $g=2$ the surface group $\pi_1(\Sigma_2)$ is
    non-amenable, opening a spectral gap ($\delta>0$) that
    obstructs construction of normalizable exact stabilizer
    eigenstates.}
\end{figure}
The second result is a sharp obstruction.  A physical code needs a
normalizable state fixed by every stabilizer; we prove that for
$g \geq 2$ none exists, not even approximately.  The cause is
geometric: the surface group $\pi_1(\Sigma_g)$ is \emph{non-amenable},
forcing a strictly positive gap in the spectrum of the stabilizer
Hamiltonian.  The same geometry that makes the gate groups
so rich thus obstructs the code states, in sharp contrast to the GKP
code, whose amenable stabilizer group $\Gamma\simeq \mathbb{Z}^2$ admits
approximate code words of arbitrary precision.

\emph{Framework.}---We take as our starting point a general quantum phase
space: a coherent-state family \emph{holomorphically} parametrized by
a complex manifold~$C$, i.e., a holomorphic map
$C \to \PP(\HH)$ into the space of rays
in~$\HH$~\cite{Berezin1975,Baez2018}.  On the complex
plane, the Bargmann representation $z \mapsto |z\rangle$ provides
such a family globally: each point $z \in \mathbb{C}$ indexes a
coherent state, and the dependence on~$z$ is holomorphic.  Given a state~$\psi$, we define its \emph{stellar function}
$f^\psi(z) = \langle\psi|z\rangle$ (antilinear in~$\psi$, with no
normalization factor), adapting conventions of Ref.~\cite{Chabaud2020} for reasons that will be clear
shortly.  On the plane, $f^\psi$ is entire, so its zeros, which
encode the non-Gaussian structure of~$\psi$, can be infinite in
number and carry no topological constraint.

Now suppose the phase space is a compact Riemann surface~$C$ of
genus~$g$; that is, suppose we have a coherent-state family
parametrized by~$C$ \cite{Baez2018}.  The map $C \to \PP(\HH)$ is well defined globally, but every
holomorphic function on~$C$ is constant, so it admits no global
holomorphic lift to~$\HH$~\cite{SM}.  Holomorphic lifts
exist only locally: one works with frames
$z \mapsto |z\rangle_\alpha$ on coordinate patches~$U_\alpha$,
related on overlaps by nowhere-vanishing holomorphic functions
$|z\rangle_\alpha = g_{\alpha\beta}(z)\,|z\rangle_\beta$.  These
transition functions define a holomorphic line
bundle~$\mathbb{L} \to C$; the \emph{stellar functions}
$f^\psi_\alpha(z) = \langle\psi|z\rangle_\alpha$ on each patch
together constitute a global holomorphic section of the dual
bundle~$\mathcal{L} := \mathbb{L}^*$.  Unlike the planar case, where
$f^\psi$ is a single entire function, the compact geometry forces
$f^\psi$ to be defined only locally, with chart-to-chart
transformations dictated by~$\mathcal{L}$.

Compactness makes the stellar representation \emph{rigid}: the
stellar rank of every state is the same topological invariant,
$\deg\mathcal{L}$.  Two states with the same zero locus are
proportional~\cite{Gunning1966}: every state is pinned, up to a
scalar, by a finite constellation of points on~$C$.  This generalizes the
Majorana constellation on the sphere
($g = 0$)~\cite{Majorana1932,Bacry2004}
and the stellar hierarchy on the plane~\cite{Chabaud2020}.  The stellar transform $\psi \mapsto f^\psi$ identifies the
Hilbert space with the space of global sections
$\Gamma(C, \mathcal{L})$ of the line bundle~$\mathcal{L}$;
for stellar rank $r = \deg\mathcal{L} > 2g{-}2$, the
Riemann--Roch theorem gives
\begin{equation}
\dim\HH = r - g + 1.
\end{equation}
Furthermore, $r:=\deg \mathcal L  = A/2\pi$ \cite{BatesWeinstein1997}, where $A$ is the area of~$C$; hence the formula reads as a generalization of the
Bohr-Sommerfeld rule to phase spaces with nontrivial topology.

\emph{Hyperbolic stellar representations.}---Quantum systems occurring in nature typically realize only simply connected phase spaces, for example the plane of a bosonic mode, or the sphere of a spin \cite{Perelomov1986}. Topologically richer phase spaces arise by quotienting a universal cover by a discrete group of symmetries. The uniformization theorem guarantees that any compact Riemann surface of genus $g \geq 2$ arises this way, specifically as a quotient of the Poincar\'e disk $\DD = \{z \in \mathbb{C} : |z| < 1\}$. Crucially, $\DD$ itself corresponds to a physically realizable phase space: two bosonic modes $\hat{a}$,
$\hat{b}$ carry a representation of $\SU(1,1)$, generated by
$\Kp = \hat{a}^\dagger\hat{b}^\dagger$,
$\Km = \hat{a}\hat{b}$, and
$\Kz = \tfrac{1}{2}(\hat{n}_a {+} \hat{n}_b {+} 1)$.  This
representation splits into sectors $\HH_k$ of fixed occupation-number difference
$\hat{n}_a - \hat{n}_b = 2k{-}1$, each containing a unique vacuum state
$|k, 0\rangle$ of lowest total boson number.  Fixing $k$, we identify points of the
disk with the squeezed translates of this
vacuum,
$|z\rangle = e^{z\Kp}|k, 0\rangle$. An element
$g = \bigl(\begin{smallmatrix}\alpha & \beta \\
\beta^* & \alpha^*\end{smallmatrix}\bigr) \in \SU(1,1)$ acts on
$\DD$ by the M\"obius transformation
$g \cdot z = (\alpha z + \beta)/(\beta^* z + \alpha^*)$, and the corresponding unitary $\hat{\mathcal D}_g$ acts on stellar functions via~\cite{BrifVourdasMann1996}
\begin{equation}\label{eq:transformation}
  f^{\hat{\mathcal{D}}_g \psi}(g\cdot z)[g'(z)]^{k}\, = f^\psi(z),
\end{equation}
where the transformation factor $j(g;z) = [g'(z)]^{k}$ \cite{footnote_double_cover} is called a \emph{factor of
automorphy}.  Therefore, stellar
functions transform as holomorphic $k$-differentials
$f^\psi(z)\,(dz)^{k}$.

\emph{Compact phase spaces via stabilizer codes}. Let
$\Gamma \subset \PSU(1,1)$ be a cocompact Fuchsian group
uniformizing a compact genus-$g$ surface $C = \Gamma \backslash \DD$.
Imposing its squeezing operators as
stabilizers~\cite{footnote_spin_structure} defines the code space
\begin{equation}\label{eq:code}
  \HH_\Gamma := \bigl\{\psi \in \mathcal H_k :
  \hat{\mathcal D}_\gamma\,\psi = \psi
  \;\;\forall\, \gamma \in \Gamma\bigr\}.
\end{equation}
For a code state, $f^\psi(z)\,(dz)^k$ is $\Gamma$-invariant and
therefore descends to a holomorphic $k$-differential on
$C = \Gamma\backslash\DD$, a section of the $k$th tensor power
$\Omega_C^{\otimes k}$ of the holomorphic cotangent bundle.  This identifies
$\HH_\Gamma$ with the space of global sections
$\Gamma(C, \Omega_C^{\otimes k})$ of this line bundle, which has
degree $\deg\Omega_C = 2g{-}2$. The stellar rank is thus
$r = k(2g{-}2)$ and the Riemann--Roch formula gives
\begin{equation}\label{eq:dim}
  \dim \HH_\Gamma = k(2g{-}2) - g + 1 = (2k - 1)(g - 1)
\end{equation}
for $k \geq 2$; for $k = 1$ the dimension is
$g$~\cite{ACGH1985}. 

\begin{figure*}[t]
  \includegraphics[width=0.99\textwidth]{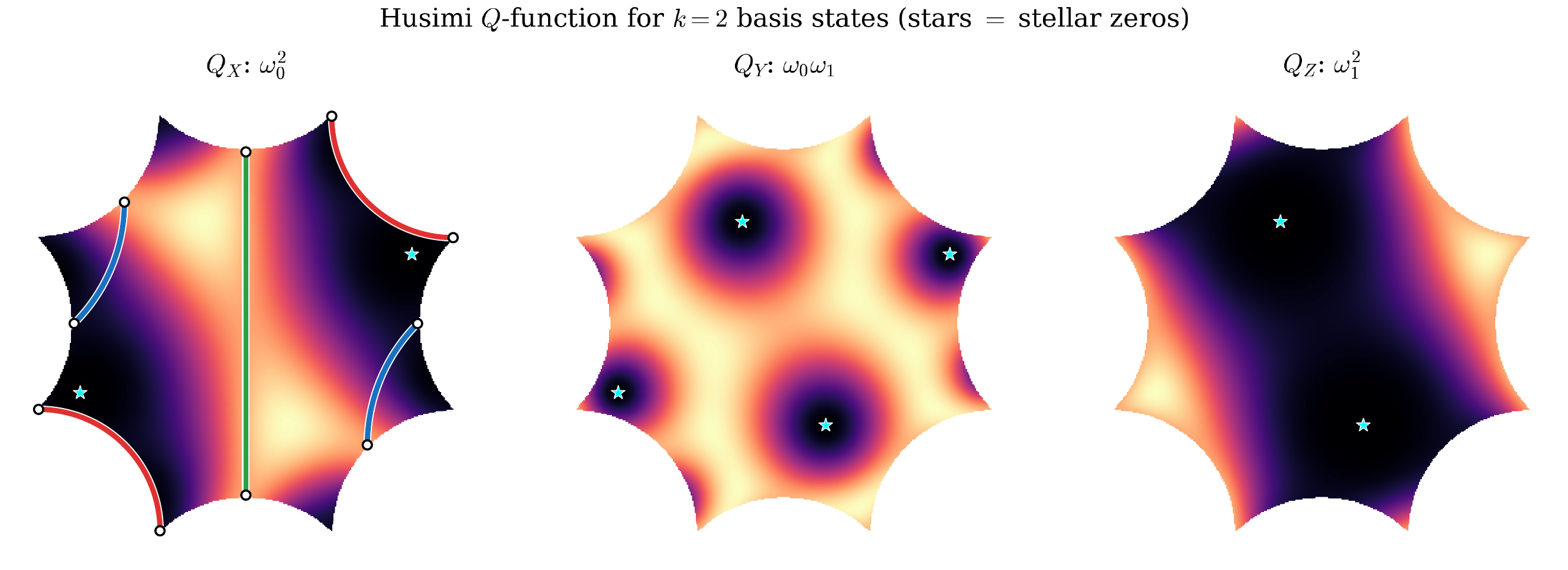}
  \caption{\label{fig:code_words}%
    Husimi $Q$-function for the $k = 2$ code space
    basis on the Bolza octagon.
    Every basis state has stellar rank
    $\deg \Omega_C^{\otimes 2} = 4$.
    Left: $Q_{\omega_0^2}$ (2~double zeros).
    Center: $Q_{\omega_0\omega_1}$ (4~simple zeros).
    Right: $Q_{\omega_1^2}$ (2~double zeros).
    Stars mark zero locations; the octagonal fundamental domain tiles the Poincar\'e disk
    under~$\Gamma$.  The three pants geodesics
    $\gamma_R, \gamma_G, \gamma_B$ (red, green, blue) are overlaid
    on the left panel; $\gamma_R$ runs along two identified edges
    of the octagon.}
\end{figure*}

\emph{Code words.}---Averaging over the stabilizer group produces
invariant states---this is how GKP code words are built at genus
one---and here we average the coherent states of the disk.  A point
$x$ of the surface is an orbit $\Gamma z \subset \DD$, and on a
small patch $U_\alpha \subset C$ one can choose the representative
$z_\alpha(x) \in \DD$ holomorphically.  Averaging then attaches a
generalized code word---a formal stabilizer average---to each
point of the surface~\cite{footnote_convergence},
\begin{equation}\label{eq:frames}
  |x\rangle_\alpha := \sum_{\gamma \in \Gamma}
  \hat{\mathcal D}_\gamma\,|z_\alpha(x)\rangle.
\end{equation}
Choosing a different coordinate patch changes $|x\rangle_\alpha$ only by a non-vanishing
scalar, so the ray it spans is chart-independent. Therefore,
the family of states Eq. \eqref{eq:frames} descends to a holomorphic map $C \to \PP(\HH_\Gamma)$. This is by definition a quantum phase space: the surface itself is the physical phase space of the code.

At genus two the $k = 1$ code space is two-dimensional, and any
basis $\omega_0\,dz$, $\omega_1\,dz$ generates the entire
tower of code spaces analytically. The $k$-differentials
\begin{equation}\label{eq:canonical}
  \begin{cases}
    \omega_0^{k-j}\,\omega_1^j\,(dz)^k, & j = 0, \ldots, k, \\[3pt]
    W \cdot \omega_0^{k-j-3}\,\omega_1^j\,(dz)^k, & j = 0, \ldots, k{-}3
    \;\;(k \geq 3),
  \end{cases}
\end{equation}
with $W = \omega_0\omega_1' - \omega_0'\omega_1$ the Wronskian,
form a basis of the code space at every integer
weight~$k$~\cite{SM} (cf.~Ref.~\cite{KockTait2014}).  The series~\eqref{eq:frames}
diverges at $k = 1$~\cite{Iwaniec2002}, so $\omega_0, \omega_1$
cannot be constructed directly by averaging over the stabilizer group $\Gamma$; we extract them from $k = 2$
instead.  It suffices to find a stellar-function basis
$X, Y, Z$ of the $k = 2$ code space obeying
the conic relation
\begin{equation}\label{eq:conic}
  Y(z)^2 = X(z)\,Z(z).
\end{equation}
The holomorphic square roots $\omega_0 = \sqrt{X}$ and
$\omega_1 = \sqrt{Z}$ are then single-valued on~$\DD$, the
differentials $\omega_0\,dz$, $\omega_1\,dz$ are a basis of the $k = 1$
code space~\cite{SM}, and Eq.~\eqref{eq:canonical} generates every
other code word.

\begin{figure*}[t]
  \scalebox{1}[0.9]{\includegraphics[width=\textwidth]{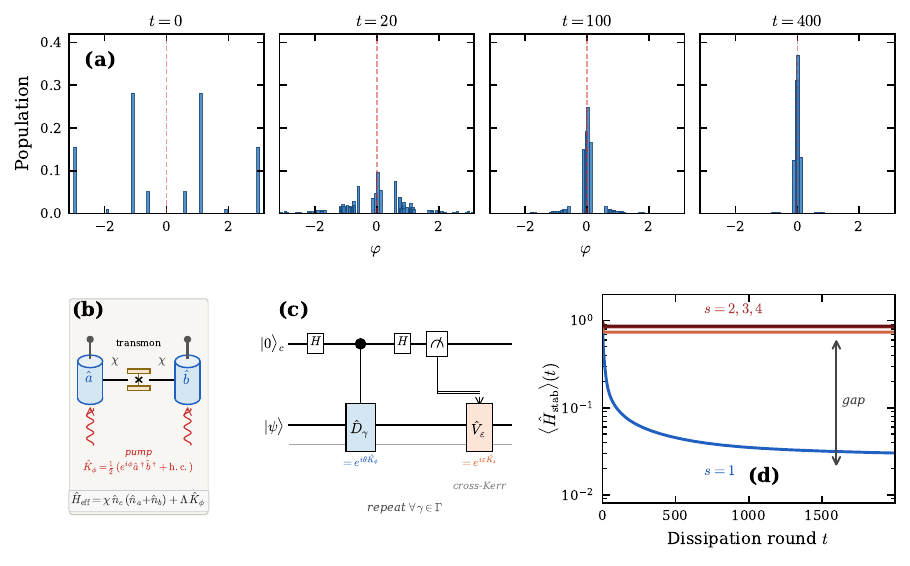}}
  \caption{\label{fig:protocol}\label{fig:scaling}%
    (a)~Single-generator conditional-kick protocol ($k = 2$, Bolza
    surface, $\Ntrunc = 60$, $\varepsilon = 0.5$): population in
    the $\Dg$-eigenbasis at rounds $t = 0, 20, 100, 400$.
    The dashed line marks $\varphi = 0$ (target stabilizer phase);
    population
    concentrates there over successive rounds.
    (b)~Circuit-QED module: two microwave cavities ($\hat a$,
    $\hat b$) dispersively coupled ($\chi$) to a transmon ancilla,
    with a single parametric pump ($\Lambda$) driving correlated
    photon-pair creation (two-mode squeezing).
    The effective Hamiltonian
    $\hat{H}_{\mathrm{eff}} = \chi\,\hat n_c(\hat n_a + \hat n_b)
    + \Lambda\,\hat{K}_\phi$ provides the native gate set.
    (c)~One round of the conditional-kick channel: controlled
    $\hat{\mathcal D}_\gamma = e^{i\theta\hat{K}_\phi}$ (parametric squeeze)
    followed by a classically conditioned correction
    $\hat{V}_\varepsilon = e^{i\varepsilon\hat{K}_z}$ (cross-Kerr),
    repeated for each generator $\gamma \in \Gamma$.
    (d)~Stabilizer energy
    $\langle\Hstab\rangle(t)$ [Eq.~\eqref{eq:Hstab}] for
    $s = 1, 2, 3, 4$ generators ($\Ntrunc = 200$).  The
    single-generator (amenable) protocol drives
    $\langle\Hstab\rangle$ toward zero, while all multi-generator
    (non-amenable) protocols collapse onto a common
    $\mathcal{O}(1)$ plateau set by the spectral gap
    (Theorem~\ref{thm:gap}).  The floor is independent
    of~$\Ntrunc$ (not shown).}
\end{figure*}

\emph{$k = 2$ theta series.}---An explicit basis at $k = 2$ comes
from closed geodesics.  Each simple closed geodesic of~$C$ defines
a conjugacy class $[\gamma] \subset \Gamma$ of hyperbolic elements
and a theta series
\begin{equation}\label{eq:theta}
  \theta_\gamma(z) = \sum_{\nu \in [\gamma]}
  \bigl(\beta^*\,z^2 + (\alpha^* - \alpha)z - \beta\bigr)^{-k},
  \quad
  \nu = \bigl(\begin{smallmatrix}
    \alpha & \beta \\ \beta^* & \alpha^*
  \end{smallmatrix}\bigr).
\end{equation}
Cutting $C$ along disjoint simple closed geodesics decomposes it
into \emph{pairs of pants}, and every such pants
decomposition uses exactly $3g - 3$ geodesics.  These theta series are linearly
independent \cite{Wolpert1986}, and therefore always provide a basis for the space of $k=2$ differentials for any complex curve. Therefore, for a curve of genus two, the problem of producing the entire tower of code words reduces to that of solving for the linear transformation that brings the three theta series in \eqref{eq:theta} into the conic form~\eqref{eq:conic}.

\emph{The Bolza curve.}---We now specialize to the simplest
genus-two example: the Bolza
curve, a regular octagonal fundamental domain in~$\DD$ with
opposite sides identified (Fig.~\ref{fig:code_words}). The four side identifications are the stabilizer
generators: one per side pair, they are the two-mode squeezing operators
$\hat{\mathcal D}_j = \exp\!\bigl(\xi_j\,\Kp - \xi_j^{*}\,\Km\bigr)$
with $\xi_j = \tfrac{\varphi}{2}\,e^{ij\pi/4}$,
$j = 0, \ldots, 3$, sharing the common squeezing strength
$\varphi = 2|\xi_j| = 2\,\mathrm{arccosh}(1 + \sqrt{2})$.

The Bolza curve is cut into two pairs of pants by three of its
shortest geodesics, $\gamma_R, \gamma_G, \gamma_B$ (each of hyperbolic
length~$\varphi$), displayed on the left panel of
Fig.~\ref{fig:code_words}.  Their conjugacy classes furnish the
theta-series basis $\theta_R, \theta_G, \theta_B$ of
Eq.~\eqref{eq:theta}, and the change of basis to the conic
form~\eqref{eq:conic} is fixed by a single complex parameter
$\lambda = -1.014(1) - 0.239(1)\,i$~\cite{SM}:
\begin{align}\label{eq:bolza_XYZ}
  X &= \tfrac{c}{6}\,(\theta_R {+} \theta_G {-} 2\theta_B)
       + \tfrac{b}{2}\,(\theta_R {-} \theta_G),
      \notag \\[3pt]
  Y &= \tfrac{a}{3}\,(\theta_R + \theta_G + \theta_B),
       \\[3pt]
  Z &= \tfrac{c}{6}\,(\theta_R {+} \theta_G {-} 2\theta_B)
       - \tfrac{b}{2}\,(\theta_R {-} \theta_G),
       \notag
\end{align}
with $a = \sqrt{6\lambda{+}3}$,
$b = \sqrt{2{-}2\lambda}$,
$c = \sqrt{6\lambda{-}6}$. For a code state the Husimi function
$Q_\psi(z) = |f^\psi(z)|^2(1-|z|^2)^{2k}$ is $\Gamma$-invariant and therefore descends to a well-defined density on the surface;
Fig.~\ref{fig:code_words} plots it on the fundamental domain for
the conic basis~\eqref{eq:bolza_XYZ}.  The double zeros of~$X$
and~$Z$ confirm that each is a perfect square, as the conic
relation~\eqref{eq:conic} demands.

{\it Logical gates}.  The order-8 rotation
$R: z \mapsto e^{i\pi/4}z$ together with an order-3 M\"obius
transformation $U$ (explicit $\SU(1,1)$ matrix in~\cite{SM})
generate the curve's 48-element symmetry group,
$\Aut(C) \cong \GL(2,\mathbb{F}_3)$.  Every symmetry is a logical
gate: it lifts to a M\"obius transformation $g$ normalizing the
Fuchsian group ($g\Gamma g^{-1} = \Gamma$), and at integer~$k$
the squeezing operators represent $\PSU(1,1)$ honestly
(non-projectively), so $\hat{\mathcal D}_g$ permutes the
stabilizers and preserves the code space.  The same symmetry
group furnishes the encoded gates of the tessellation codes of
Ref.~\cite{WXL2024}, where the surface is a tiled configuration
space; in our codes it is the phase space itself.

Like the stabilizers, the logical gates are two-mode Gaussian
circuits:
$\hat{\mathcal D}_R = e^{i(\pi/4)\Kz}$ is a phase rotation, while
$\hat{\mathcal D}_U = \hat{\mathcal D}_0^{1/2}\hat{\mathcal D}_R^3$
is implemented by three octagon rotations followed by a stabilizer squeeze with half the squeezing parameter
(${\sim}6.6$~dB of two-mode squeezing)~\cite{SM}.  At $k = 1$ the
code space is a logical qubit: writing $|\omega\rangle$ for the
code word whose stellar function is~$\omega$, the states
$|\omega_0\rangle, |\omega_1\rangle$ form a basis, and the
logical gates act as
\begin{align}\label{eq:reps}
  V_R &= \frac{1}{\sqrt{6}}
  \begin{pmatrix}1{-}i\sqrt{3} & i\sqrt{2} \\
    -i\sqrt{2} & -(1{+}i\sqrt{3})\end{pmatrix},
 \end{align}
 \begin{align}
  V_U &= \frac{1}{2\sqrt{3}}
  \begin{pmatrix}-(\sqrt{3}{-}i) & 2\sqrt{2} \\
    -2\sqrt{2} & -(\sqrt{3}{+}i)\end{pmatrix}.
\end{align}
These two unitary matrices generate a faithful copy of
$\GL(2, \mathbb{F}_3)$: all 48 symmetries give distinct Gaussian
gates.  The Bolza curve is only the simplest example: by
Greenberg's theorem~\cite{Jones2019}, \emph{any} finite group
can be realized as the classical logical gate group, at integer
weight, on \textit{some} hyperbolic surface.  Away from integer weight, however, the symmetries act
only projectively, and only a subgroup of them lifts to valid
quantum operations at
all~\cite{footnote_projective_logicals,SM}.

{\it Physical platform}.\label{sec:platform} The two-mode
$\mathrm{SU}(1,1)$ representation requires two microwave resonators
(modes $\hat{a}$, $\hat{b}$) linearly coupled to a transmon
ancilla; the two-mode squeezing drive is supplied by a DC
SQUID~\cite{Grosso2025}.  In the dispersive
regime, a Schrieffer--Wolff transformation yields
\begin{equation}\label{eq:Heff}
  \hat{H}_\mathrm{eff} \simeq \Delta_c\,\hat{n}_c
  + \chi\,\hat{n}_c\,(\hat{n}_a {+} \hat{n}_b)
  + \Lambda\,\Kphi,
\end{equation}
where $\Delta_c$ is the transmon detuning in the rotating frame,
contributing only a passive transmon phase~\cite{SM}, $\chi$ is
the dispersive shift (assumed equal for both modes),
$\Lambda$ the parametric pump amplitude, and
$\Kphi = \tfrac{1}{2}(e^{i\phi}\hat{a}^{\dagger}\hat{b}^{\dagger}
+ e^{-i\phi}\hat{a}\hat{b})$ is the two-mode squeezing generator.

Three gates follow.  (i)~\emph{Conditional two-mode squeeze}: when
$\chi \gg \Lambda$, the parametric drive is resonant only for the
transmon in $|g\rangle$, giving, up to an unconditional
squeeze~\cite{SM},
$\mathrm{C}\Kphi(\theta) = e^{i\hat{n}_c\theta\Kphi}$.
(ii)~\emph{Conditional $\Kz$}: setting $\Lambda = 0$, the cross-Kerr
interaction yields
$\mathrm{C}\Kz(\theta) = e^{i\hat{n}_c\theta\Kz}$
with $\Kz = \tfrac{1}{2}(\hat{n}_a {+} \hat{n}_b {+} 1)$.
(iii)~\emph{Unconditional $\SU(1,1)$}: preparing the transmon in
$|e\rangle$ converts any conditional gate to unconditional.  Generic
$\SU(1,1)$ elements are synthesized by Trotterization.

\emph{Dissipative stabilization.}---These gates assemble into one
stabilization round per generator (Fig.~\ref{fig:protocol}c):
prepare the transmon in $|+\rangle$, apply the controlled squeeze
$\mathrm{C}\hat{\mathcal{D}}_\gamma$, and measure the transmon in
the $X$ basis; when the outcome is $-1$, apply the feedforward
correction $\Veps = e^{i\varepsilon\Kz}$.  One round realizes
the Kraus map
$\hat{\mathcal{K}}_0 = \tfrac{1}{2}(\id + \hat{\mathcal{D}}_\gamma)$,
$\hat{\mathcal{K}}_1 = \tfrac{1}{2}\Veps(\id - \hat{\mathcal{D}}_\gamma)$,
cycled over the generators~\cite{SM}.  Code states are dark: the
error outcome fires with probability
$\tfrac{1}{2}\bigl(1 - \mathrm{Re}\langle\Dg\rangle\bigr)$,
vanishing as the state approaches the code space~\cite{SM}. For a single Bolza generator,
$\mathrm{Re}\langle\Dg\rangle$ rises from $0.03$ to $0.97$ within
500 rounds (Fig.~\ref{fig:protocol}a), converging from any
initial state, including the maximally mixed state, to the same
attractor.

{\it The spectral gap obstruction}\label{sec:nogo}. While
single-generator stabilization succeeds
(Fig.~\ref{fig:protocol}a), adding a second squeezing stabilizer
along a different axis renders stabilization impossible, not
merely slower.  Define the stabilizer Hamiltonian
\begin{equation}\label{eq:Hstab}
  \Hstab = \frac{1}{n_{\mathrm{gen}}}\sum_{i=1}^{n_{\mathrm{gen}}}\Bigl(\id -
  \tfrac{1}{2}\bigl(\hat{\mathcal{D}}_i + \hat{\mathcal{D}}_i^\dagger\bigr)\Bigr) = \id - M,
\end{equation}
where the sum runs over the stabilizer generators, $n_{\mathrm{gen}}$ is
their number, and $M$ is the Markov operator of a symmetric random walk
on~$\Gamma$, here realized by two-mode squeezing operators.  The spectral gap
$\delta = 1 - \sup\sigma(M)$ measures how far \emph{any}
normalizable state must be from satisfying all stabilizer
constraints simultaneously.

The same random walk can also be run on the group itself, in the
regular representation ($\Gamma$ translating functions in
$\ell^2(\Gamma)$), and there its Markov operator $M_\lambda$
\emph{is} gapped: every cocompact Fuchsian group is
\emph{non-amenable}~\cite{BHV2008}, so $\|M_\lambda\| < 1$
(Kesten's theorem).  Since $\sup\sigma(M) \leq \|M\|$, the gap transfers
to~$\HH_k$ provided $\|M\| \leq \|M_\lambda\|$.  It does: the
representation of~$\Gamma$ on~$\HH_k$ is \emph{weakly
contained}~\cite{BHV2008} in the regular representation, and weak
containment never increases operator norms~\cite{SM}.

\begin{theorem}[Spectral gap]\label{thm:gap}
  For $g \geq 2$, the stabilizer Hamiltonian~\eqref{eq:Hstab} has
  a spectral gap $\delta > 0$: $\langle v|\Hstab|v\rangle \geq
  \delta$ for every unit vector $v \in \HH_k$.
\end{theorem}

\noindent {\it Proof sketch}. Concatenate the three bounds above:
$\sup\sigma(M) \leq \|M\| \leq \|M_\lambda\| < 1$; the full
argument is in~\cite{SM}.~$\square$\\

The gap can be made quantitative.   At genus two, Nagnibeda's upper bound on $\|M_\lambda\|$ gives
$\delta \geq 1 - \|M_\lambda\| \geq
0.337$~\cite{Nagnibeda1997,SM}---one bound for every genus-two
code built on the standard four surface-group generators, since
$\|M_\lambda\|$ sees only the abstract group and its generating
set, not their realization---while for the Bolza surface at $k = 2$ the lowest eigenvalue
of~$\Hstab$, extrapolated from finite-dimensional truncation,
is $\delta \approx 0.45$~\cite{SM}: three quarters of the gap is
set by the group-theoretic properties of~$\Gamma$ alone.

The gap obstructs not only exact code words but
\emph{approximate} ones.  For GKP codes,
$\inf\sigma(\Hstab) = 0$ and one can construct approximate
code words with arbitrarily small stabilizer
energy~\cite{GKP2001}. The gap $\delta > 0$ makes this
impossible here: no sequence of normalizable states can
drive $\langle\Hstab\rangle$ below~$\delta$.
Combining with the known existence of approximate GKP
states:

\begin{corollary}[Amenability dichotomy]\label{cor:dichotomy}
  Among the bosonic stabilizer codes on compact Riemann surfaces
  of genus~$g \geq 1$ constructed here---plane and
  discrete-series representations---approximate code words of
  arbitrarily high precision exist if and only if\/ $g = 1$.
\end{corollary}

{\it Discussion}\label{sec:discussion}. This work completes a
trinity of continuous-variable quantum codes from compact, two-dimensional phase spaces, classified by
curvature and genus $g$.  On the sphere ($g = 0$), the stabilizer group is finite
and exact code words exist; examples include spin~\cite{Gross2020} and the spherical-Landau-level~\cite{Fan2023} codes.  On the flat
torus ($g = 1$), the stabilizer group $\mathbb{Z}^2$ is infinite but
amenable, and the GKP codes~\cite{GKP2001} admit approximate code
words of arbitrary precision at the cost of increasing energy.  

We
settle the hyperbolic case $g \geq 2$: the non-amenability of the
surface group opens a spectral gap (Theorem~\ref{thm:gap}) that
forbids even approximate stabilizer eigenstates.  Amenability
of~$\pi_1(\Sigma_g)$ is the sharp dividing line
(Corollary~\ref{cor:dichotomy}). The infinite-dimensional setting is essential: in finite
dimension every generalized eigenvector is normalizable, so a
nonzero generalized code space cannot coexist with a positive
bound on all normalizable states.  The closest finite-dimensional analogue is the no-low-energy-trivial-states (NLTS)
theorem~\cite{FreedmanHastings2014,ABN2023}: for certain
families of qubit stabilizer Hamiltonians whose interaction
graphs are expander graphs, every state below a fixed energy
density is nontrivial, i.e. unreachable by a constant-depth
circuit. So expansion is a kind of finite-family analogue of
non-amenability~\cite{SM}, but it removes only the
\emph{trivial} low-energy states; here non-amenability
removes them all.

The obstruction, however, is intrinsically a statement about the exact unitary stabilizers $\hat{\mathcal{D}}_\gamma$. Finite-energy analogues—in which these unitaries are deformed by a non-unitary envelope to confine the system's energy—need not inherit this gap. For standard GKP codes, precisely this strategy underlies the Trotterized stabilization protocols of Royer \emph{et al.}~\cite{Royer2020}, which engineer dissipation toward the $+1$ eigenspace of finite-energy stabilizers. Determining if these dissipative stabilization schemes can be adapted to tame the non-amenable geometry of higher-genus curves is an interesting open challenge for realizing these codes in hardware.

{\it Acknowledgements}. We would like to thank Hossein Dehgani and Jonathan Roberts for helpful discussions. D.R. acknowledges support from the Joint Quantum Institute at the University of Maryland, where this work was initiated, and from Extropic. V.V.A. acknowledges NSF grant OMA2120757 (QLCI). All code, notebooks, and data reproducing the numerical results and figures of this work are available in a Zenodo replication package~\cite{ZenodoRepo2026}. This manuscript was edited with the assistance of Claude, developed by Anthropic. Claude was used to help obtain the spectral-gap theorem and amenability dichotomy, in accordance with the authors' instructions. All content, claims, and conclusions have been reviewed and verified by the authors to ensure accuracy and originality. Certain products, commercial and otherwise, are mentioned in this publication. These mentions are for informational purposes only, and do not imply recommendation or endorsement by NIST.

\bibliography{prl-refs}

\begin{thebibliography}{38}
\bibitem [{foo({\natexlab{a}})}]{footnote_double_cover}%
  \BibitemOpen
  \bibinfo {note} {When $k$ is not an integer, evaluating $[g'(z)]^k$ requires
  a choice of holomorphic square root; the resulting $\pm 1$ ambiguity means
  the squeezing operators only double-cover the group $\PSU(1,1)$ of M\"obius
  transformations. In the planar case the analogous ambiguity spans the full
  $\mathrm{U}(1)$}\BibitemShut {NoStop}%
\bibitem [{foo({\natexlab{b}})}]{footnote_spin_structure}%
  \BibitemOpen
  \bibinfo {note} {Properly, a stabilizer group is a subgroup of the squeezing
  operators that projects isomorphically onto $\Gamma$. At half-integer~$k$
  each generator contains a sign ambiguity, and a consistent choice of signs
  corresponds precisely to a choice of spin structure on the Riemann surface.
  See \cite{SM} for more details}\BibitemShut {NoStop}%
\bibitem [{foo({\natexlab{c}})}]{footnote_convergence}%
  \BibitemOpen
  \bibinfo {note} {Convergence in Eq.~\eqref{eq:frames} means convergence of
  the stellar function as a holomorphic function on $\mathbb{D}$, i.e.\
  pointwise convergence. For $k \geq 2$ this convergence is uniform on compact
  subsets of the disk}\BibitemShut {NoStop}%
\end{thebibliography}

\end{document}


\title{Supplemental Material: Bosonic codes from compact phase spaces}

\author{David Roberts}
\thanks{These authors contributed equally to this work.}
\affiliation{Extropic Corporation, Cambridge, Massachusetts 02458, USA}
\affiliation{Joint Quantum Institute, University of Maryland, College Park, MD 20742, USA}
\affiliation{Joint Center for Quantum Information and Computer Science, NIST/University of Maryland, College Park, MD 20742, USA}

\author{Aaron Slipper}
\thanks{These authors contributed equally to this work.}
\affiliation{Department of Mathematics, Duke University, Durham, NC 27708, USA}

\author{Alireza Parhizkar}
\affiliation{Joint Quantum Institute, University of Maryland, College Park, MD 20742, USA}

\author{Victor V.\ Albert}
\affiliation{Joint Center for Quantum Information and Computer Science, NIST/University of Maryland, College Park, MD 20742, USA}

\author{Mohammad Hafezi}
\affiliation{Joint Quantum Institute, University of Maryland, College Park, MD 20742, USA}
\affiliation{Joint Center for Quantum Information and Computer Science, NIST/University of Maryland, College Park, MD 20742, USA}

\date{\today}

\maketitle

\renewcommand{\theequation}{S\arabic{equation}}
\renewcommand{\thefigure}{S\arabic{figure}}
\renewcommand{\thetable}{S\arabic{table}}
\renewcommand{\thesection}{S\Roman{section}}
\setcounter{section}{0}
\setcounter{secnumdepth}{3}

\vspace{\baselineskip}

\tableofcontents

\section{Framework: phase-space stabilizer codes}\label{fw:framework}

\emph{Classical and generalized displacement groups.}---Fix a quantized phase
space: a complex manifold $C$ together with a holomorphic map into $\PP(\HH)$, as in the main text\footnote{Two realizations occur in this work: the plane,
with the Bargmann family $|z\rangle=e^{z\hat a^{\dagger}}|0\rangle$,
and the disk, with the $\SU(1,1)$ family of the main text,
$|z\rangle=e^{z\Kp}|k,0\rangle$ at Bargmann index $k$
(Sec.~\ref{fw:hyperbolic}). We expect a more general
framework in which the physical system supporting the stabilizer code itself may have a noncontractible phase space $C$, but we do not
develop it here. Such a framework is necessary for considering e.g. concatenated codes.}. We call a
\emph{classical displacement group} any subgroup
$G\subseteq\Aut_{\mathrm{hol}}(C)$ of the holomorphic automorphism
group of $C$.  A \emph{generalized displacement group} is a central
extension of $G$ by a group of phases $Z$,
\begin{equation}\label{fw:eq:pauli-ext}
  1 \longrightarrow Z \longrightarrow \cD
    \xrightarrow{\ \pi\ } G \longrightarrow 1 ,
\end{equation}
realized on states through a projective action of $G$ that is
\emph{linearized} on $\cD$, in the sense that each element
$\tilde g\in\cD$ over $g=\pi(\tilde g)\in G$ is represented by a unitary operator on $\mathcal H$ that acts on the
coherent-state frame by
\begin{equation}\label{fw:eq:automorphy}
  \hat{\mathcal D}_{\tilde g}\,|z\rangle \;=\; \cJ_{\tilde g}(z)\,\bigl|\,g(z)\bigr\rangle ,
\end{equation}
where the nowhere-vanishing holomorphic multiplier $\cJ_{\tilde g}(z)$
is the \emph{factor of automorphy}.  Associativity of
\eqref{fw:eq:automorphy} is the cocycle identity
$\cJ_{\tilde g'\tilde g}(z)=\cJ_{\tilde g'}(g z)\,\cJ_{\tilde g}(z)$, so
$\tilde g\mapsto\hat{\mathcal D}_{\tilde g}$ is an honest (unitary)
representation of $\cD$ and a projective representation of $G$.

\begin{definition}[Phase-space stabilizer code]\label{fw:def:code}
A \emph{classical stabilizer group} is a discrete subgroup
$\Gamma\subset G$ acting freely and properly discontinuously on $C$.
A \emph{phase-space stabilizer code} is a lift of $\Gamma$ through
$\pi$: a subgroup $\cS\subset\cD$ with
$\pi|_{\cS}\colon\cS\to\Gamma$ an isomorphism, equivalently the
image of a splitting $\sigma\colon\Gamma\to\cD$
($\pi\circ\sigma=\mathrm{id}_{\Gamma}$).  We write
$\CS := C/\Gamma$ for the compact quotient phase space.
\end{definition}

\noindent Exact stabilizer states are non-normalizable---ideal GKP
code words already show this---so the fixed-point condition must be
posed in a space slightly larger than $\HH$.  We use a rigged Hilbert
space (Gelfand triple) $\Phi\subset\HH\subset\Phi'$: a dense test
space $\Phi$ containing every coherent frame vector $|z\rangle$ and
preserved by every displacement operator,
$\hat{\mathcal D}_{\tilde g}\,\Phi=\Phi$, together with its antidual
$\Phi'$, on which the $\hat{\mathcal D}_{\tilde g}$ then act by
duality:
\begin{equation}\label{fw:eq:dual-action}
  \bigl\langle\hat{\mathcal D}_{\tilde g}\,\psi\,\big|\,\varphi\bigr\rangle
  \;:=\;
  \bigl\langle\psi\,\big|\,\hat{\mathcal D}_{\tilde g}^{\dagger}\,\varphi\bigr\rangle,
  \qquad
  \psi\in\Phi',\quad \varphi\in\Phi .
\end{equation}
The \emph{code space} of $\cS$ consists of the states fixed by every
stabilizer,
\begin{equation}\label{fw:eq:code-space}
  \mathcal H_{\cS}
  \;:=\;\bigl\{\psi\in\Phi'\,:\,
  \hat{\mathcal D}_{s}\,\psi=\psi\ \ \forall s\in\cS\bigr\}.
\end{equation}
Elements of $\mathcal H_{\cS}$ are the \emph{code
states} (stabilizer states).\\

\subsection{Inner product on $\mathcal H_{\cS}$}
For $\mathcal H_{\cS}$ to carry a
stellar transform of its own, an inner product must first be defined:
the code space consists of generalized vectors---it sits inside the
rigged extension $\Phi'$---and inherits no inner product from $\HH$.
In particular the $Q$-function of a code state is not a priori
defined.  What \emph{is} available is the overlap density upstairs:
for code states $\psi,\phi$, form
\begin{equation}\label{fw:eq:overlap-density}
  \Lambda_{\psi,\phi}(z)
  \;:=\;\frac{\langle\psi|z\rangle\,\langle z|\phi\rangle}
             {\langle z|z\rangle}
  \;=\;\frac{f^{\psi}(z)\,f^{\phi}(z)^{*}}{\langle z|z\rangle}\,,
\end{equation}
the density of the stellar inner product on $C$---the overlap density
of a \emph{pair} of states, not merely the Husimi density of a single
one.  It is $\Gamma$-invariant:
\begin{equation}\label{fw:eq:density-invariance}
  \Lambda_{\psi,\phi}(\gamma z) =\frac{f^{\psi}(\gamma z)\,f^{\phi}(\gamma z)^{*}}
       {\langle\gamma z|\gamma z\rangle}
  \;=\;|j_{\gamma}(z)|^{-2}\,
  \frac{f^{\psi}(z)\,f^{\phi}(z)^{*}}
       {\langle z|\hat{\mathcal D}_{\sigma(\gamma)}^{\dagger}
        \hat{\mathcal D}_{\sigma(\gamma)}|z\rangle\,
        |j_{\gamma}(z)|^{-2}}
  \;=\;\Lambda_{\psi,\phi}(z),
\end{equation}
by the automorphy of the stellar functions of code states and
unitarity of the lifted stabilizer.  The upstairs inner product of two
code states therefore diverges for a trivial reason: it integrates
infinitely many copies of the same value, one per tile of the
$\Gamma$-action.  The fix is to choose a fundamental domain for the $\Gamma$ action, and integrate once over that domain.

\begin{definition}[Inner product on the code space]\label{fw:def:petersson}
Fix a fundamental domain $\mathcal F\subset C$ for the action of
$\Gamma$ (by \eqref{fw:eq:density-invariance} the
choice is immaterial).  For $\psi,\phi\in\mathcal H_{\cS}$,
\begin{equation}\label{fw:eq:petersson}
  \langle\psi|\phi\rangle_{\cS}
  \;:=\;\int_{\mathcal F}
  \frac{f^{\psi}(z)\,f^{\phi}(z)^{*}}{\langle z|z\rangle}\,d\mu(z),
\end{equation}
where $d\mu$ is the invariant volume form of $C$ (the pullback of the
ambient Fubini--Study form along $z\mapsto[\,|z\rangle\,]$),
normalized so that the ambient resolution of identity reads
$\id=\int_{C}P_z\,d\mu$ with
$P_z=|z\rangle\langle z|/\langle z|z\rangle$ (on the disk,
$d\mu=\tfrac{2k-1}{\pi}\,(1-|z|^2)^{-2}\,dA$; on the plane,
$d\mu=dA/\pi$).\footnote{On the disk we take $k>\tfrac12$.}
\end{definition}

\subsection{Coherent states in the code space}
We now define systems of coherent states associated with the code space. To do this precisely, we work in charts. Concretely, choose a \emph{good} cover $\{U_\alpha\}$ of
$\CS$---one whose nonempty finite intersections are all
contractible---and on each patch a local lift
$z_\alpha\colon U_\alpha\to C$ of the quotient map $q\colon C\to C/\Gamma$. Because $\Gamma$ acts holomorphically, $q$ is a local
biholomorphism, so each $z_\alpha$ is a biholomorphism onto a sheet
of the cover---a chart on $U_\alpha$ modeled on $C$ rather than on
$\mathbb C^{n}$.  For the complex \emph{curves} of this work $C\subseteq \mathbb{C}$ and the
distinction evaporates.\\

Projecting the corresponding coherent states into the code space produces local
holomorphic frames\footnote{This sum is formally divergent as an element of $\mathcal H$. When a rigorous meaning is needed, we interpret this sum via its stellar transform, as a sum in $\mathcal O_C$---pointwise convergent whenever the weight satisfies $2k > 2$, and understood as a formal sum otherwise.}
\begin{equation}\label{fw:eq:local-frame}
  |x\rangle_{\alpha} \;:=\; |z_\alpha(x)\rangle_{\cS}
  \;=\;\sum_{s\in\cS}\hat{\mathcal D}_s\,|z_\alpha(x)\rangle ,
  \qquad x\in U_\alpha .
\end{equation}
The collection $\{(U_\alpha,z_\alpha)\}$ is an atlas exhibiting $\CS$
as a complex manifold, and the group $\Gamma$ enters through its changes of
coordinates.  On an overlap $U_{\alpha\beta}=U_\alpha\cap U_\beta$ the
two charts are related by a holomorphic change of coordinates. Because $U_{\alpha\beta}$ is
connected, there exists a unique deck transformation $\gamma_{\alpha\beta}\in\Gamma$ which relates the two coordinates:
\begin{equation}\label{fw:eq:chart-transition}
  z_\beta \;=\; \gamma_{\alpha\beta}\circ z_\alpha
  \quad\text{on } U_{\alpha\beta},
  \qquad \gamma_{\alpha\beta}\in\Gamma \text{ unique}.
\end{equation}
Writing
$j_{\gamma}:=\cJ_{\sigma(\gamma)}$, the automorphy law \eqref{fw:eq:automorphy} gives
$|\gamma z\rangle
=j_{\gamma}(z)^{-1}\,\hat{\mathcal D}_{\sigma(\gamma)}|z\rangle$, and
the frames compare by a direct computation:
\begin{equation}\label{fw:eq:frame-transition}
\begin{split}
  |x\rangle_{\beta}
  \;=\;\sum_{s\in\cS}\hat{\mathcal D}_s\,
      \bigl|\gamma_{\alpha\beta}\,z_\alpha(x)\bigr\rangle
  \;&=\; j_{\gamma_{\alpha\beta}}\!\bigl(z_\alpha(x)\bigr)^{-1}
      \sum_{s\in\cS}\hat{\mathcal D}_{s}\,
      \hat{\mathcal D}_{\sigma(\gamma_{\alpha\beta})}\,
      |z_\alpha(x)\rangle \\
  \;&=\; j_{\gamma_{\alpha\beta}}\!\bigl(z_\alpha(x)\bigr)^{-1}\,
      |x\rangle_{\alpha} ,
\end{split}
\end{equation}
the last equality because $\sigma(\gamma_{\alpha\beta})\in\cS$, so
right multiplication by it merely reindexes the stabilizer sum: the
stabilizer absorbs its own elements.  This is precisely the
local data of a holomorphic line bundle $ \mathbb L_{\cS}\longrightarrow\CS$, with transition functions $j_{\gamma_{\alpha\beta}}^{-1}\circ z_\alpha$.\\

{\it Kodaira map}. The above computation demonstrates in particular that overlapping frames differ by scalars, so
$x\mapsto[\,|x\rangle_\alpha]$ is chart-independent and defines a
globally holomorphic map
\begin{equation}\label{fw:eq:kodaira-map}
  \Psi\colon \CS \longrightarrow \PP(\mathcal H_{\cS})
\end{equation}
into the projectivized code space, with
\begin{equation}
    \mathbb L_{\cS}=\Psi^{*}\mathcal O(-1)
\end{equation} the pullback
of the tautological bundle.\footnote{\label{fw:fn:cocycle-def}In some situations, $\Psi$ might not
exist: if the averaged frames \eqref{fw:eq:local-frame} vanish at
some point, the ray $[\,|x\rangle_\alpha]$ is undefined there.
This is avoided as long as $\Ls$ is \emph{basepoint-free}---its
global sections have no common zero.  More generally, we take as
the \emph{definition} of $\mathbb L_{\cS}$ the line bundle with
transition functions $j_{\gamma_{\alpha\beta}}^{-1}\circ z_\alpha$
[Eq.~\eqref{fw:eq:frame-transition}]: these are nowhere vanishing
and satisfy the cocycle identity.  This makes $\mathbb{L}_{\cS}$
well defined as long as the code $\cS$ is---even if $\Ls$ is not
basepoint-free, and even if the frames \eqref{fw:eq:local-frame} formally diverge as elements of $\mathcal O_C$. In situations where the frames exist and have no common zero, we keep the frame language, so that the abstract transition functions that define the sheaf $\mathbb{L}_{\mathcal S}$ are read off from actual physical states.}\\

\begin{theorem}[Compact stellar transform of a code
state]\label{fw:thm:reproducing}
Let $\psi\in \mathcal H_{\mathcal S}$ and $\{f^\psi_\alpha\}_\alpha$ denote its stellar transform. Then
\begin{equation}\label{fw:eq:reproducing}
  f^\psi_\alpha(x) :=\bigl\langle\psi\,\big|\,x\bigr\rangle_\alpha^{\mathcal S}
  \;=\;f^{\psi}\bigl(z_\alpha(x)\bigr),
\end{equation}
where $f^{\psi}$ is the stellar transform of $\psi$ with respect to
$C$, $|z_\alpha(x)\rangle_{\cS}=|x\rangle_\alpha$ is the compact
coherent state \eqref{fw:eq:local-frame}, and
$\langle\psi\,|\,x\rangle^{\cS}_{\alpha}$ denotes the inner product
of Definition~\ref{fw:def:petersson} of $\psi$ with the frame
$|x\rangle_\alpha$.
\end{theorem}

\begin{proof}
Let $\psi$ be a code state, write $z_0=z_\alpha(x)$, and for
$\gamma\in\Gamma$ define
\begin{equation}
  T_{\gamma}(z)\;:=\;
  \frac{\langle\psi|z\rangle\,
        \langle z|\hat{\mathcal D}_{\sigma(\gamma)}|z_0\rangle}
       {\langle z|z\rangle}.
\end{equation}
Then $T_{\gamma}(\gamma'z)=T_{\gamma'^{-1}\gamma}(z)$:
\begin{equation}\label{fw:eq:T-covariance}
  T_{\gamma}(\gamma'z)
  =\frac{|j_{\gamma'}(z)|^{-2}\,
    \langle\psi|\hat{\mathcal D}_{\sigma(\gamma')}|z\rangle\,
    \langle z|\hat{\mathcal D}_{\sigma(\gamma')}^{\dagger}
    \hat{\mathcal D}_{\sigma(\gamma)}|z_0\rangle}
   {|j_{\gamma'}(z)|^{-2}\,
    \langle z|\hat{\mathcal D}_{\sigma(\gamma')}^{\dagger}
    \hat{\mathcal D}_{\sigma(\gamma')}|z\rangle}
  =T_{\gamma'^{-1}\gamma}(z):
\end{equation}
the factors of automorphy cancel between numerator and denominator,
$\psi$ is an $\cS$-stabilizer state, and
$\hat{\mathcal D}_{\sigma(\gamma')}^{\dagger}
\hat{\mathcal D}_{\sigma(\gamma)}
=\hat{\mathcal D}_{\sigma(\gamma'^{-1}\gamma)}$ because
$\hat{\mathcal D}$ is an honest (linear) unitary representation of
$\cD$ and $\sigma$ is a homomorphism.  Now insert the stabilizer sum
\eqref{fw:eq:local-frame} into \eqref{fw:eq:petersson}: the integrand
is $\sum_{\gamma}T_{\gamma}$, and \eqref{fw:eq:T-covariance} with
$\gamma'=\gamma$, together with the invariance of $d\mu$, gives
$\int_{\mathcal F}T_{\gamma}\,d\mu
=\int_{\gamma^{-1}\mathcal F}T_{e}\,d\mu$.  The tiles
$\{\gamma^{-1}\mathcal F\}_{\gamma\in\Gamma}$ partition $C$, so
\begin{equation}
  \langle\psi|z_\alpha(x)\rangle_{\cS}
  =\sum_{\gamma}\int_{\mathcal F}T_{\gamma}\,d\mu
  =\int_{C}\frac{\langle\psi|z\rangle\,\langle z|z_0\rangle}
                {\langle z|z\rangle}\,d\mu(z)
  =\langle\psi|z_0\rangle
  =f^{\psi}(z_0),
\end{equation}
the last two equalities by the ambient resolution of identity.  The
stabilizer sum in the ket has exactly compensated the restriction of
the domain---one group element per tile.
\end{proof}

\begin{theorem}\label{fw:thm:linear-normality}
The compact stellar transform
\begin{equation}\label{fw:eq:compact-stellar}
  \mathcal H_{\cS}\;\longrightarrow\;\Gamma(\CS,\Ls),
  \qquad
  |\psi\rangle\;\longmapsto\;\{f^{\psi}_{\alpha}\}_{\alpha},
  \qquad
  f^{\psi}_{\alpha}(x):=\langle\psi|z_\alpha(x)\rangle,
\end{equation}
is an antilinear isomorphism onto the global sections of
$\Ls:=\mathbb L_{\cS}^{\vee}$.
\end{theorem}

\begin{proof}
The stellar transform $\psi\mapsto f^{\psi}$ on $C$ identifies the
code space $\mathcal H_{\cS}$ with the space of $\Gamma$-automorphic
forms
\begin{equation}\label{fw:eq:automorphic-space}
  \mathcal A(\Gamma,\cS)
  \;\equiv\;\bigl\{f\in\mathcal O(C):
  f(\gamma z)=j_{\gamma}(z)^{-1}f(z)\ \ \forall\gamma\in\Gamma\bigr\}.
\end{equation}
(The automorphy condition on $f^{\psi}$ is precisely the
$\cS$-invariance \eqref{fw:eq:code-space} of $\psi$, read through the
automorphy law \eqref{fw:eq:automorphy}; every automorphic form
arises because the stellar transform identifies $\Phi'$ with the
holomorphic functions of admissible growth, as in the main text.)
By Theorem~\ref{fw:thm:reproducing} the compact
stellar transform is $f^{\psi}_{\alpha}=f^{\psi}\circ z_\alpha$,
whose gluing law
$f_{\beta}=j_{\gamma_{\alpha\beta}}(z_\alpha)^{-1}f_{\alpha}$
[inherited from \eqref{fw:eq:frame-transition}] is that of a section
of $\Ls$.  It remains to invert the map
$f\mapsto\{f\circ z_\alpha\}$ from $\mathcal A(\Gamma,\cS)$ to
$\Gamma(\CS,\Ls)$.

Given $z\in C$ over a point $x\in U_\alpha$, freeness of the
$\Gamma$-action gives a unique $\gamma\in\Gamma$ with
$z=\gamma\,z_\alpha(x)$.  Define
\begin{equation}\label{fw:eq:inverse-formula}
  f(z) \;\equiv\; j_{\gamma}\bigl(z_\alpha(x)\bigr)^{-1} f_\alpha(x).
\end{equation}
\emph{Chart-independence.}  If also $x\in U_\beta$, write
$z=\gamma' z_\beta(x)$; from
$z_\beta=\gamma_{\alpha\beta}z_\alpha$ and uniqueness,
$\gamma=\gamma'\gamma_{\alpha\beta}$.  Using the gluing law
$f_\beta=j_{\gamma_{\alpha\beta}}(z_\alpha)^{-1}f_\alpha$ and the
cocycle identity,
\begin{equation}
  j_{\gamma'}(z_\beta)^{-1}f_\beta
  = j_{\gamma'}\bigl(\gamma_{\alpha\beta}z_\alpha\bigr)^{-1}
    j_{\gamma_{\alpha\beta}}(z_\alpha)^{-1} f_\alpha
  = j_{\gamma'\gamma_{\alpha\beta}}(z_\alpha)^{-1} f_\alpha 
  = j_{\gamma}(z_\alpha)^{-1} f_\alpha ,
\end{equation}
so \eqref{fw:eq:inverse-formula} is well defined.

\emph{Holomorphy.}  Near any $z_0$, the deck element $\gamma$ and the
chart $\alpha$ may be held fixed, so $f$ is locally the product of the
nowhere-vanishing holomorphic function $j_\gamma(z_\alpha\circ q)^{-1}$
with $f_\alpha\circ q$; hence $f\in\mathcal O(C)$.

\emph{Automorphy.}  For $\gamma_0\in\Gamma$ the unique deck element
carrying $z_\alpha(x)$ to $\gamma_0 z$ is $\gamma_0\gamma$, so by the
cocycle identity
\begin{equation}
  f(\gamma_0 z)
  = j_{\gamma_0\gamma}(z_\alpha)^{-1} f_\alpha
  = j_{\gamma_0}\bigl(\gamma z_\alpha\bigr)^{-1}
    j_{\gamma}(z_\alpha)^{-1} f_\alpha
  = j_{\gamma_0}(z)^{-1} f(z),
\end{equation}
i.e.\ $f\in\mathcal A(\Gamma,\cS)$.

\emph{The maps are mutually inverse.}  Setting $\gamma=e$ in
\eqref{fw:eq:inverse-formula} gives $f\circ z_\alpha=f_\alpha$, so the
composite $\{f_\alpha\}\mapsto f\mapsto\{f\circ z_\alpha\}$ is the
identity.  Conversely, if $f_\alpha=f^{\psi}\circ z_\alpha$, then for
\emph{every} $z=\gamma z_\alpha(x)\in C$,
\begin{equation}
  f(z) = j_\gamma(z_\alpha)^{-1} f^{\psi}(z_\alpha)
       = f^{\psi}\bigl(\gamma z_\alpha(x)\bigr) = f^{\psi}(z),
\end{equation}
by the automorphy of $f^{\psi}$---so the reconstruction returns
$f^{\psi}$ on all of $C$, not merely on the chart images.  The
compact stellar transform
$\mathcal H_{\cS}\to\Gamma(\CS,\Ls)$ is therefore an isomorphism.
\end{proof}

\begin{corollary}\label{fw:cor:sections}
$\mathcal H_{\cS}\simeq\Gamma(\CS,\Ls)$ with
$\Ls=\mathbb L_{\cS}^{\vee}$, the isomorphism being
$|\psi\rangle\mapsto\{f^{\psi}_{\alpha}\}_{\alpha}$ with
$f^{\psi}_{\alpha}(x)=\langle\psi|z_\alpha(x)\rangle$: the
\emph{compact} stellar transform of $\psi$ on the coordinate patch
$U_\alpha$ is the noncompact stellar transform of $\psi$ with respect
to $C$, viewed in the coordinate chart
$z_\alpha\colon U_\alpha\to C$.
\end{corollary}

Whenever $\Ls$ is basepoint-free, the Kodaira map~\eqref{fw:eq:kodaira-map}
lands in the projectivization of the finite-dimensional space
$H^{0}(\CS,\Ls)$.  The image is a projective algebraic variety: by Remmert's proper mapping theorem~\cite{Bedford84_fw},
the image of a compact complex manifold under a holomorphic map to
projective space is a compact analytic subvariety, and by Chow's
theorem~\cite{Chow49_fw}, every compact analytic subvariety of
projective space is an algebraic variety (see also
\cite{GriffithsHarris_fw}, p.~167).  

{\it Transport by displacement operators}.  The frame bundle responds
to unitaries by transport of structure, a mechanism we record once
and use repeatedly.

\begin{lemma}[Transport of sheaves by invertible
operators]\label{fw:lem:transport}
Let $\hat U$ be an invertible operator on (the rigged extension of)
$\HH$, and let $\mathbb L$ be an invertible sheaf on a complex
manifold $X$ whose local sections are vector-valued holomorphic maps
$x\mapsto|\psi_x\rangle$.  Define the image sheaf $\hat U(\mathbb L)$
by
\begin{equation}\label{fw:eq:image-sheaf-def}
  \Gamma\bigl(U,\hat U(\mathbb L)\bigr)
  \;:=\;\bigl\{\hat U|\psi_x\rangle \,:\,
  |\psi_x\rangle\in\Gamma(U,\mathbb L)\bigr\}.
\end{equation}
Then $\hat U\colon\mathbb L\to\hat U(\mathbb L)$ is an isomorphism of
sheaves of $\mathcal O_X$-modules; in particular $\hat U(\mathbb L)$
is again invertible.
\end{lemma}

\begin{proof}
$\hat U$ is $\mathcal O_X$-linear,
$\hat U\bigl(f(x)\,|\psi_x\rangle\bigr)=f(x)\,\hat U|\psi_x\rangle$,
and commutes with restriction:
$(\hat U|\psi_x\rangle)\big|_{V}=\hat U\bigl(|\psi_x\rangle\big|_{V}\bigr)$
for $V\subset U$; hence it is a morphism of sheaves of
$\mathcal O_X$-modules.  By the same two properties, $\hat U^{-1}$ is
a morphism $\hat U(\mathbb L)\to\mathbb L$, and the two composites are
the identity: $\hat U$ is an isomorphism, with explicit inverse
$\hat U^{-1}$ (for unitary $\hat U$, its adjoint $\hat U^{\dagger}$).
\end{proof}

\subsection{No global holomorphic lift on compact
  surfaces}\label{sec:no_lift}

We prove the claim in the main text that a holomorphic coherent-state
family parametrized by a compact Riemann surface $C$ admits no
global holomorphic section.

\begin{proof}
  Suppose $z \mapsto |z\rangle$ were a global holomorphic map from
  $C$ to $\HH$ (not just to $\PP(\HH)$).  For any
  $|\phi\rangle \in \HH$, the function
  $g(z) = \langle\phi|z\rangle$ would be holomorphic on the compact
  surface~$C$, hence constant by the maximum principle.  Since this
  holds for every $|\phi\rangle$, the map $z \mapsto |z\rangle$
  itself is constant---contradicting the assumption that it
  parametrizes a non-trivial coherent-state family.

  More precisely, the obstruction is topological.  The coherent-state
  map $\Psi\colon C \to \PP(\HH)$ defines the holomorphic line
  bundle $\mathbb{L} = \Psi^{*}\mathcal{O}(-1)$, which has negative
  degree for any non-constant holomorphic map to projective space.
  A global holomorphic lift would provide a nowhere-vanishing
  section of $\mathbb{L}$, impossible for a line bundle of nonzero
  degree.  One must therefore work with local frames
  $|z\rangle_\alpha$ on coordinate patches $U_\alpha$, related by
  transition functions
  $|z\rangle_\alpha = g_{\alpha\beta}(z)\,|z\rangle_\beta$.
\end{proof}

\section{Hyperbolic bosonic codes}\label{fw:hyperbolic}

Here $G=\PSU(1,1)$ and $\cD=\cD_k$ is the group of generalized
(hyperbolic) displacement operators of the main text, associated with
the holomorphic discrete series $D_k^{+}$, for which
\begin{equation}\label{fw:eq:hyp-multiplier}
  \cJ_{\tilde g}(z)
  \;=\;(\bar\beta z+\bar\alpha)^{-2k}
  \;=\;g'(z)^{\,k},
\end{equation}
where the lift $\tilde g$ over the classical M\"obius map $g$ is the
choice of phase entering the real power: at half-integral $k$,
$Z\cong\{\pm1\}$ is the choice of branch of the square root, and in
general $Z\cong\langle e^{2\pi ik}\rangle$\footnote{A priori the phase ambiguity in
\eqref{fw:eq:hyp-multiplier} could be function-valued; in fact it is constant, by a two-line argument.  Let
$\cJ$ and $\cJ'$ be two branches of $(\bar\beta z+\bar\alpha)^{-2k}$
over the same classical map $g$ (the matrix being itself defined up
to sign).  They share the modulus
$|\bar\beta z+\bar\alpha|^{-2k}$, so $\cJ'/\cJ$ is a holomorphic
function on $\DD$ of unit modulus, hence constant by the open mapping
theorem.}.

Recall from Sec.~\ref{fw:framework} that the compact system of
coherent states $\{|x\rangle_\alpha\}_\alpha$ of a phase-space
stabilizer code glues across overlapping charts by the holomorphic
transition functions
$j_{\gamma_{\alpha\beta}}(z_\alpha(x))^{-1}$
[Eq.~\eqref{fw:eq:frame-transition}]; by
\eqref{fw:eq:hyp-multiplier} this reads
\begin{equation}\label{fw:eq:hyp-transition}
  |x\rangle_{\beta}
  \;=\;\gamma_{\alpha\beta}'\bigl(z_\alpha(x)\bigr)^{-k}\,
  |x\rangle_{\alpha}.
\end{equation}
This is the transformation rule of holomorphic differentials on
$\CS$: the chain rule applied to the chart transition
\eqref{fw:eq:chart-transition} gives
$dz_\beta=\gamma_{\alpha\beta}'(z_\alpha)\,dz_\alpha$, hence
$(dz_\beta)^{\otimes(-k)}
=\gamma_{\alpha\beta}'(z_\alpha)^{-k}\,(dz_\alpha)^{\otimes(-k)}$.
Therefore the chart-by-chart assignment
\begin{equation}\label{fw:eq:cech-iso}
  \mathbb L_{\cS}\;\xrightarrow{\ \sim\ }\;
  \Omega_{\CS}^{\otimes(-k)},
  \qquad
  |x\rangle_{\alpha}\;\longmapsto\;(dz_\alpha)^{\otimes(-k)},
\end{equation}
matches transition functions on both sides and is an isomorphism of
holomorphic line bundles.\footnote{For integer $k$ the bundle
$\Omega_{\CS}^{\otimes(-k)}$ is well-defined.  For fractional $k$
the $k$-th power of $\Omega_{\CS}$ is defined \emph{through} the
choice of roots implicit in the lift $\sigma$: multiplying transition
functions corresponds to tensoring line bundles, so a consistent
choice of branches of $(\gamma_{\alpha\beta}')^{k}$---which is
exactly the data a lift $\cS$ of $\Gamma$ supplies---is the
transition data of a root of the canonical bundle.}

Dually, $\Ls=\mathbb L_{\cS}^{\vee}\simeq\Omega_{\CS}^{\otimes k}$,
and Theorem~\ref{fw:thm:linear-normality} specializes to
\begin{equation}\label{fw:eq:hyp-sections}
  \mathcal H_{\cS}
  \;\simeq\;
  \Gamma\bigl(\CS,\;\Omega_{\CS}^{\otimes k}\bigr),
\end{equation}
with the root of the canonical bundle on the right determined by the
choice of lift $\cS$ of the classical stabilizer group.  This
undergirds---and finally justifies---the use of the Riemann--Roch
theorem to count the code-space dimension in the main text:
\begin{equation}\label{fw:eq:RR}
  \dim\mathcal H_{\cS}
  \;=\;\deg\Ls\;-\;g(\CS)\;+\;1,
  \qquad
  \deg\Omega_{\CS}^{\otimes k}=k\bigl(2g(\CS)-2\bigr),
\end{equation}
valid whenever the stellar rank $\deg\Ls$ exceeds $2g(\CS)-2$; the
low-degree corrections are those quoted in the main text.

\section{The logical gate group}\label{fw:logical}

As for any stabilizer code, the
logical operators form the \emph{logical gate group}: the
normalizer of the stabilizer in the displacement group, modulo the
stabilizer,
\begin{equation}\label{fw:eq:logical-def}
  \cD_{\cS} \;=\; N_{\cD}(\cS)/\cS .
\end{equation}
This section identifies the logical gate group with the
polarization-preserving automorphisms of $(\CS,\Ls)$, and Mumford's theta group is recovered as the
$\mathbb C^\times$-pushout of the logical gate group
(Theorem~\ref{fw:thm:theta}), specializing at genus one---where the
quotient $\CS=\mathbb C/\Gamma$ \emph{is} an elliptic curve---to
Mumford's classical construction.

\begin{theorem}[Logical gate action]\label{fw:thm:main}
There is a short exact sequence
\begin{equation}\label{fw:eq:main-ses}
  1 \longrightarrow Z \longrightarrow \cD_{\cS}
    \longrightarrow G_{\cS} \longrightarrow 1 ,
\end{equation}
in which the \emph{classical logical gate group} $G_{\cS}$ is a
subgroup of $\Aut_{\Ls}(\CS)$, the polarization-preserving
automorphisms of $\CS$ (equivalently, the automorphisms preserving
the frame bundle $\mathbb L_{\cS}$).  Consequently $G_{\cS}$ acts
projectively on $\mathcal H_{\cS}$.
\end{theorem}

Throughout this section, we work on the good cover
$\mathfrak U=\{U_\alpha\}$ of Sec.~\ref{fw:framework}. To keep the letter $g$ for
group elements we write $t_{\alpha\beta}$ for transition cocycles, so
the code bundle $\mathbb L_{\cS}$ has cocycle
$t^{\cS}_{\alpha\beta}=j_{\gamma_{\alpha\beta}}^{-1}\circ z_\alpha$.\\

For $g\in N_{G}(\Gamma)$ write $\bar g$
for the induced automorphism of the quotient,
$\bar g\bigl(q(z)\bigr):=q(gz)$---well defined precisely because $g$
normalizes $\Gamma$---and let $\tilde g$ be any lift of $g$.  We now
compute the image sheaf of $\mathbb L_{\cS}$ under
$\hat{\mathcal D}_{\tilde g}$.
Applying $\hat{\mathcal D}_{\tilde g}$ to a frame
\eqref{fw:eq:local-frame} of $\mathbb L_{\cS}$ and commuting it
through the stabilizer sum,
\begin{equation}\label{fw:eq:transport-frames}
  \hat{\mathcal D}_{\tilde g}\,|x\rangle_{\alpha}
  \;=\!\!\sum_{s\in\tilde g\cS\tilde g^{-1}}\!\!
    \hat{\mathcal D}_{s}\;\hat{\mathcal D}_{\tilde g}\,
    |z_\alpha(x)\rangle
  \;=\;\cJ_{\tilde g}\bigl(z_\alpha(x)\bigr)
  \!\!\sum_{s\in\tilde g\cS\tilde g^{-1}}\!\!
    \hat{\mathcal D}_{s}\,|g\,z_\alpha(x)\rangle
  \;=\;\cJ_{\tilde g}\bigl(z_\alpha(x)\bigr)\,
    |\bar g x\rangle^{\tilde g\cS\tilde g^{-1}}_{\alpha},
\end{equation}
where the conjugated code $\tilde g\cS\tilde g^{-1}$---itself a
phase-space stabilizer code, since $g$ normalizes $\Gamma$---is
framed in the translated atlas
$\bigl\{\bigl(\bar g(U_\alpha),\;g\circ z_\alpha\circ\bar g^{-1}\bigr)\bigr\}$,
so that the last equality is exact.
Eq.~\eqref{fw:eq:transport-frames} says that
$\hat{\mathcal D}_{\tilde g}$ carries frames of $\mathbb L_{\cS}$ to
nowhere-vanishing holomorphic multiples of pulled-back frames, i.e.
\begin{equation}\label{fw:eq:image-sheaf}
  \hat{\mathcal D}_{\tilde g}\bigl(\mathbb L_{\cS}\bigr)
  \;=\;\bar g^{*}\,\mathbb L_{\tilde g\cS\tilde g^{-1}},
\end{equation}
where $\hat{\mathcal D}_{\tilde g}(\mathbb{L}_{\cS})$ is the image
sheaf of Lemma~\ref{fw:lem:transport}.  Immediately, we have the
following result:

\begin{proposition}\label{fw:prop:conj-pullback}
Let $g\in N_{G}(\Gamma)$ and let $\tilde g\in\cD$ be any lift of $g$.
Then
\begin{equation}\label{fw:eq:conj-pullback}
  \mathbb L_{\tilde g^{-1}\cS\tilde g}
  \;\simeq\; \bar g^{*}\,\mathbb L_{\cS} ,
\end{equation}
the isomorphism being implemented by the displacement operator itself.
\end{proposition}

\begin{proof}
Lemma~\ref{fw:lem:transport} and Eq.~\eqref{fw:eq:image-sheaf},
applied to the lift $\tilde g^{-1}$ of $g^{-1}$, give the isomorphism
$\hat{\mathcal D}_{\tilde g}^{-1}\colon\mathbb L_{\cS}
\xrightarrow{\ \sim\ }(\bar g^{-1})^{*}\,\mathbb L_{\tilde g^{-1}\cS\tilde g}$.
Apply the pullback $\bar g^{*}$ to both sides.
\end{proof}

\begin{corollary}\label{fw:cor:descend}
Define the \emph{$\cS$-polarized normalizer}
\begin{equation}\label{fw:eq:NGS}
  N^{\cS}_{G}(\Gamma)
  \;:=\;\bigl\{\,g\in N_{G}(\Gamma)\;:\;
  \bar g^{*}\,\mathbb L_{\cS}\simeq\mathbb L_{\cS}\,\bigr\}.
\end{equation}
If $\tilde g\in N_{\cD}(\cS)$, then
$g=\pi(\tilde g)\in N^{\cS}_{G}(\Gamma)$.
\end{corollary}

\begin{proof}
Applying the projection $\pi\colon\cD\to G$ to the normalizer
identity gives
$\Gamma=\pi(\cS)=\pi(\tilde g^{-1}\cS\tilde g)=g^{-1}\Gamma g$, so
$g\in N_{G}(\Gamma)$.  It remains to show
$\bar g^{*}\mathbb L_{\cS}\simeq\mathbb L_{\cS}$, which follows from
Proposition~\ref{fw:prop:conj-pullback} combined with the hypothesis
$\tilde g^{-1}\cS\tilde g=\cS$.
\end{proof}

For the converse we track the phase ambiguity of a lift explicitly.
Let $g\in N_{G}(\Gamma)$ and let $\tilde g$ be \emph{any} lift of
$g$.  For $\gamma\in\Gamma$,
$\pi(\tilde g^{-1}\sigma(\gamma)\tilde g)=g^{-1}\gamma g\in\Gamma$,
so
\begin{equation}\label{fw:eq:chi-def}
  \tilde g^{-1}\sigma(\gamma)\,\tilde g
  \;=\;\chi_{\tilde g}(g^{-1}\gamma g)\;\sigma(g^{-1}\gamma g),
  \qquad \chi_{\tilde g}(g^{-1}\gamma g)\in Z .
\end{equation}
From \eqref{fw:eq:chi-def} it follows immediately that
$\chi_{\tilde g}\colon\Gamma\to Z$ is a character:
\begin{equation}
  \tilde g^{-1}\sigma(\gamma\gamma')\,\tilde g
  =\bigl(\tilde g^{-1}\sigma(\gamma)\tilde g\bigr)
   \bigl(\tilde g^{-1}\sigma(\gamma')\tilde g\bigr)
  =\chi_{\tilde g}(g^{-1}\gamma g)\,
   \chi_{\tilde g}(g^{-1}\gamma' g)\,
   \sigma(g^{-1}\gamma\gamma' g),
\end{equation}
where we used that $\sigma$ is a homomorphism.  For a character
$\chi\colon\Gamma\to Z$, let $\mathbb L_{\chi}$ denote the flat line
bundle on $\CS$ defined by the constant transition functions
\begin{equation}\label{fw:eq:char-bundle}
  t^{\chi}_{\alpha\beta}\;:=\;\chi(\gamma_{\alpha\beta})^{-1}
\end{equation}
(the inverse mirrors the convention
$t^{\cS}_{\alpha\beta}=j_{\gamma_{\alpha\beta}}^{-1}\circ z_\alpha$
for $\mathbb L_{\cS}$).

\begin{proposition}[Flat twist]\label{fw:prop:flat-twist}
Let $g\in N_{G}(\Gamma)$ and let $\tilde g$ be any lift of $g$.  Then
\begin{equation}\label{fw:eq:flat-twist}
  \mathbb L_{\tilde g^{-1}\cS\tilde g}
  \;\simeq\;\mathbb L_{\cS}\otimes\mathbb L_{\chi_{\tilde g}},
\end{equation}
with $\chi_{\tilde g}$ the character of \eqref{fw:eq:chi-def}.
\end{proposition}

\begin{proof}
By \eqref{fw:eq:chi-def} the conjugated code is the lift
$\gamma\mapsto\chi_{\tilde g}(\gamma)\,\sigma(\gamma)$ of $\Gamma$,
so its frames are
$|x\rangle^{\tilde g^{-1}\cS\tilde g}_{\alpha}
=\sum_{\gamma\in\Gamma}\chi_{\tilde g}(\gamma)\,
\hat{\mathcal D}_{\sigma(\gamma)}|z_\alpha(x)\rangle$ (central
elements act as the corresponding phases).  We compute the cocycle
exactly as in Eq.~\eqref{fw:eq:frame-transition}:
\begin{equation}
\begin{split}
  |x\rangle^{\tilde g^{-1}\cS\tilde g}_{\beta}
  &=\sum_{\gamma\in\Gamma}\chi_{\tilde g}(\gamma)\,
    \hat{\mathcal D}_{\sigma(\gamma)}\,|z_\beta(x)\rangle
  = j_{\gamma_{\alpha\beta}}\bigl(z_\alpha(x)\bigr)^{-1}
    \sum_{\gamma\in\Gamma}\chi_{\tilde g}(\gamma)\,
    \hat{\mathcal D}_{\sigma(\gamma\gamma_{\alpha\beta})}\,
    |z_\alpha(x)\rangle \\
  &=\Bigl(j_{\gamma_{\alpha\beta}}\bigl(z_\alpha(x)\bigr)\,
    \chi_{\tilde g}(\gamma_{\alpha\beta})\Bigr)^{-1}\,
    |x\rangle^{\tilde g^{-1}\cS\tilde g}_{\alpha},
\end{split}
\end{equation}
reindexing $\gamma\mapsto\gamma\gamma_{\alpha\beta}^{-1}$ in the last
step and using multiplicativity of $\chi_{\tilde g}$.  The resulting
cocycle is
$t^{\cS}_{\alpha\beta}\,t^{\chi_{\tilde g}}_{\alpha\beta}$---the
cocycle of the tensor product
$\mathbb L_{\cS}\otimes\mathbb L_{\chi_{\tilde g}}$.
\end{proof}

\begin{corollary}\label{fw:cor:lift-normalizes}
Let $g\in N^{\cS}_{G}(\Gamma)$.  Then every lift $\tilde g$ of $g$
normalizes $\cS$.
\end{corollary}

\begin{proof}
Combining Propositions~\ref{fw:prop:conj-pullback}
and~\ref{fw:prop:flat-twist} with the hypothesis
$\bar g^{*}\mathbb L_{\cS}\simeq\mathbb L_{\cS}$ gives
\begin{equation}
  \mathbb L_{\cS}\otimes\mathbb L_{\chi_{\tilde g}}
  \;\simeq\;\mathbb L_{\tilde g^{-1}\cS\tilde g}
  \;\simeq\;\bar g^{*}\,\mathbb L_{\cS}
  \;\simeq\;\mathbb L_{\cS},
\end{equation}
so $\mathbb L_{\chi_{\tilde g}}$ is trivial.  A flat character bundle
is trivial iff the character is trivial---under the identification
$\Pic^{0}(\CS)\cong\Hom(\Gamma,\Us(1))$,
$\mathbb L_{\chi}\mapsto\chi$~\cite{Buss09_fw}---hence
$\chi_{\tilde g}\equiv1$, and \eqref{fw:eq:chi-def} reads
\begin{equation}\label{fw:eq:normalizes}
  \tilde g^{-1}\sigma(\gamma)\,\tilde g
  =\sigma(g^{-1}\gamma g)\in\cS
  \qquad\forall\,\gamma\in\Gamma,
\end{equation}
i.e.\ $\tilde g\in N_{\cD}(\cS)$.
\end{proof}

\begin{corollary}\label{fw:cor:normalizer-ses}
The sequence
\begin{equation}\label{fw:eq:normalizer-ses}
  1 \longrightarrow Z \longrightarrow N_{\cD}(\cS)
    \xrightarrow{\ \pi\ } N^{\cS}_{G}(\Gamma) \longrightarrow 1
\end{equation}
is exact.
\end{corollary}

\begin{proof}
Corollary~\ref{fw:cor:descend} says $\pi$ restricted to
$N_{\cD}(\cS)$ takes values in $N^{\cS}_{G}(\Gamma)$;
Corollary~\ref{fw:cor:lift-normalizes} says it is surjective onto it
(lifts exist because $\pi$ is surjective).  Its kernel is
$Z\cap N_{\cD}(\cS)$; since $Z$ is central,
$Z\subseteq N_{\cD}(\cS)$, so the kernel is all of $Z$.
\end{proof}

Quotienting \eqref{fw:eq:normalizer-ses} by the stabilizers yields
Theorem~\ref{fw:thm:main} through an elementary group-theoretic
device, Lemma~\ref{fw:lem:quotient} of Sec.~\ref{fw:ses}.
Applying that lemma to \eqref{fw:eq:normalizer-ses} with
$D=\cS$ (so $\pi(D)=\Gamma$ and $A\cap D=Z\cap\cS=1$) gives exactly
\eqref{fw:eq:main-ses}, with
$G_{\cS}=N^{\cS}_{G}(\Gamma)/\Gamma$.  Finally, $G_{\cS}$ sits
inside $\Aut_{\Ls}(\CS)$---the full group of automorphisms of $\CS$
preserving $\mathbb L_{\cS}$, equivalently its dual $\Ls$---as the
subgroup induced by elements of $G$.  This proves
Theorem~\ref{fw:thm:main}. \hfill$\square$

\emph{Logical gates as bundle isomorphisms.}---For
$\tilde p\in N_{\cD}(\cS)$ the conjugated code \emph{is} $\cS$, so
the image-sheaf computation \eqref{fw:eq:image-sheaf} specializes to
\begin{equation}\label{fw:eq:gate-iso}
  \hat{\mathcal D}_{\tilde p}\colon\;
  \mathbb L_{\cS}\;\xrightarrow{\ \sim\ }\;
  \bar p^{*}\,\mathbb L_{\cS},
  \qquad
  \hat{\mathcal D}_{\tilde p}\,|x\rangle_{\alpha}
  \;=\;\cJ_{\tilde p}\bigl(z_\alpha(x)\bigr)\,
       |\bar p x\rangle_{\alpha} .
\end{equation}
Therefore, for a phase-space stabilizer code, a logical gate in
$\cD_{\cS}$ can be viewed as an isomorphism of
the frame bundle onto its pullback.

\emph{Mumford's theta group.}---The logical gate group admits a
canonical enlargement by scalars.  Consider the set of operators
\begin{equation}\label{fw:eq:pushout-ops}
  \bigl\{\,\lambda\,\hat{\mathcal D}_{\tilde g}\,\cS \;:\;
  \lambda\in\mathbb C^{\times},\ \tilde g\in N_{\cD}(\cS)\,\bigr\}
  \;\simeq\;\cD_{\cS}\times_{Z}\mathbb C^{\times}
\end{equation}
(cosets of the stabilizer, as in \eqref{fw:eq:logical-def}; the group
on the right, the \emph{pushout}, is
$(\cD_{\cS}\times\mathbb C^{\times})/\!\sim$ with
$(z\tilde p,\lambda)\sim(\tilde p,z\lambda)$ for $z\in Z$---the
identification absorbs the central phases common to both factors).
The enlarged operators fit into the enlarged exact sequence
\begin{equation}\label{fw:eq:pushout-ses}
  1\longrightarrow\mathbb C^{\times}
  \longrightarrow\cD_{\cS}\times_{Z}\mathbb C^{\times}
  \longrightarrow G_{\cS}\longrightarrow 1 .
\end{equation}
The enlarged gate group \eqref{fw:eq:pushout-ses} coincides with a
celebrated construction of algebraic geometry, the \emph{theta
group}~\cite{Mumford66_fw} of the pair $(\CS,\mathbb L_{\cS})$,
defined as the set
\begin{equation}\label{fw:eq:theta-def}
  G(\mathbb L_{\cS})
  \;=\;\bigl\{(\bar g,\varphi)\;:\;\bar g\in G_{\cS},\ \
  \varphi\colon\mathbb L_{\cS}\xrightarrow{\ \simeq\ }
  \bar g^{*}\,\mathbb L_{\cS}\bigr\},
\end{equation}
equipped with the group law
$(\bar g,\varphi)\cdot(\bar h,\psi)
=(\bar g\bar h,\;\bar h^{*}\varphi\circ\psi)$.\footnote{Recall the
pullback of a morphism $\varphi\colon\mathcal F\to\mathcal G$ of
sheaves along an automorphism $\bar h$ of the base space, defined by
$\bar h^{*}(\varphi)\bigl(\bar h^{*}s\bigr)
=\bar h^{*}\bigl(\varphi(s)\bigr)$; the contravariance
$(\bar g\bar h)^{*}=\bar h^{*}\bar g^{*}$ makes the group law
associative.}  Transport \eqref{fw:eq:gate-iso} provides a map from
the pushout into the theta group,
\begin{equation}\label{fw:eq:theta-map}
  \Theta\colon\;\lambda\,\hat{\mathcal D}_{\tilde g}\,\cS
  \;\longmapsto\;
  \bigl(\bar g,\ \lambda\,\hat{\mathcal D}_{\tilde g}\colon
  \mathbb L_{\cS}\to\bar g^{*}\,\mathbb L_{\cS}\bigr),
\end{equation}
well defined because all frames of $\mathbb L_{\cS}$ consist of
stabilizer states: for any $\gamma\in\Gamma$,
$\lambda\hat{\mathcal D}_{\tilde g}\hat{\mathcal D}_{\sigma(\gamma)}$
and $\lambda\hat{\mathcal D}_{\tilde g}$ act identically on
$\mathbb L_{\cS}$.

\begin{theorem}[Logical gates and the theta group]\label{fw:thm:theta}
The map \eqref{fw:eq:theta-map} is an isomorphism of groups,
\begin{equation}\label{fw:eq:theta-iso}
  \cD_{\cS}\times_{Z}\mathbb C^{\times}
  \;\simeq\;G(\mathbb L_{\cS}).
\end{equation}
\end{theorem}

\begin{proof}
\emph{Homomorphism.}  The product
$(\lambda\hat{\mathcal D}_{\tilde g})(\chi\hat{\mathcal D}_{\tilde h})
=\lambda\chi\,\hat{\mathcal D}_{\tilde g\tilde h}$ is sent to
$\bigl(\overline{gh},\,
\lambda\chi\,\hat{\mathcal D}_{\tilde g\tilde h}\bigr)$, so we must
verify
$\lambda\chi\,\hat{\mathcal D}_{\tilde g\tilde h}
=\bar h^{*}\bigl(\lambda\hat{\mathcal D}_{\tilde g}\bigr)
 \circ\bigl(\chi\hat{\mathcal D}_{\tilde h}\bigr)$
as morphisms
$\mathbb L_{\cS}\to(\bar g\bar h)^{*}\mathbb L_{\cS}$.  This holds
because an operator acts on the \emph{value} of a section, not on its
base point, so the pulled-back morphism
$\bar h^{*}(\lambda\hat{\mathcal D}_{\tilde g})$ again acts as
$\lambda\hat{\mathcal D}_{\tilde g}$: for
$|\psi_x\rangle\in\Gamma(U,\mathbb L_{\cS})$,
\begin{equation}
  \bar h^{*}\bigl(\lambda\hat{\mathcal D}_{\tilde g}\bigr)
  \bigl(\chi\hat{\mathcal D}_{\tilde h}\,|\psi_x\rangle\bigr)
  \;=\;\lambda\hat{\mathcal D}_{\tilde g}\,
       \chi\hat{\mathcal D}_{\tilde h}\,|\psi_x\rangle .
\end{equation}

\emph{Injectivity.}  Suppose
$\Theta\bigl(\lambda\hat{\mathcal D}_{\tilde g}\,\cS\bigr)
=(1,\mathrm{id}_{\mathbb L_{\cS}})$.  Then $\bar g=1$, so
$g\in\Gamma$ and
$\hat{\mathcal D}_{\tilde g}=z\,\hat{\mathcal D}_{\sigma(g)}$ for
some $z\in Z$.  Hence $\lambda\hat{\mathcal D}_{\tilde g}$ acts on
$\mathbb L_{\cS}$ as the constant $\lambda z$, and
$\mathrm{id}_{\mathbb L_{\cS}}$ forces $\lambda z=1$; therefore
$\lambda\hat{\mathcal D}_{\tilde g}\,\cS=\cS$, the identity coset.

\emph{Surjectivity.}  Let $(\bar g,\varphi)\in G(\mathbb L_{\cS})$.
Choose a representative $g\in N^{\cS}_{G}(\Gamma)$ of $\bar g$ and a
lift $\tilde g$ of $g$; by
Corollary~\ref{fw:cor:lift-normalizes}, $\tilde g\in N_{\cD}(\cS)$.
Set
$\varphi_{\tilde g}:=\hat{\mathcal D}_{\tilde g}\big|_{\mathbb L_{\cS}}
\colon\mathbb L_{\cS}\to\bar g^{*}\mathbb L_{\cS}$
[Eq.~\eqref{fw:eq:gate-iso}].  Then
$\varphi_{\tilde g}^{-1}\circ\varphi$ is an automorphism of
$\mathbb L_{\cS}$, hence multiplication by a global unit
$u\in H^{0}(\CS,\mathcal O^{\times}_{\CS})$; because $\CS$ is compact
and connected, $u$ is a nonzero constant
$\lambda\in\mathbb C^{\times}$.  Therefore
$\varphi=\lambda\,\varphi_{\tilde g}
=\Theta\bigl(\lambda\hat{\mathcal D}_{\tilde g}\,\cS\bigr)$.
\end{proof}

\subsection{Genus one: recovering Mumford's theta group}
\label{fw:sub:genus1}

Take $\CS=C/\Gamma$ an elliptic curve, so that
$G=\mathbb C$, $Z=\Us(1)$, and $\cD=\Heis(2,\mathbb R)$ is the
single-mode Heisenberg--Weyl group generated by the Weyl operators
$\hat{\mathcal D}_{g}$, $g\in G$.  The stabilizer
$\cS=\sigma(\Gamma)$ is fixed by the images of two generators
$\gamma_1,\gamma_2$ of $\Gamma=\langle\gamma_1,\gamma_2\rangle$, which must commute.
Writing $\sigma(\gamma)=z_{\gamma}\hat{\mathcal D}_{\gamma}$ with $z_{\gamma}\in Z$
a phase, and using that phases are central, the group commutator is
\begin{equation}\label{fw:eq:commutator}
  [\sigma(\gamma_1),\sigma(\gamma_2)]
  = [\hat{\mathcal D}_{\gamma_1},\hat{\mathcal D}_{\gamma_2}]
  = e^{\,i\,\omega(\gamma_1,\gamma_2)},
  \qquad
  \omega(a,b) = 2\operatorname{Im}(\bar{a}\,b),
\end{equation}
Commutativity therefore requires
$\mathrm{Area}(\CS)=|\omega(\gamma_1,\gamma_2)|=2\pi d$ for an integer
$d\geq1$; fix the generators so $\omega(\gamma_1,\gamma_2)=2\pi d>0$.  Once the
images commute, $\sigma(n\gamma_1+m\gamma_2):=\sigma(\gamma_1)^n\sigma(\gamma_2)^m$
is a homomorphism, and $\cS$ is a valid code.

\emph{Logical gate group.}---Since $G=\mathbb C$ is abelian, every
$g\in G$ normalizes $\Gamma$ with trivial conjugation action, and
every lift of $g$ differs from the Weyl operator
$\hat{\mathcal D}_{g}$ by a central phase, which cancels in
\eqref{fw:eq:chi-def}: the character reduces to the group commutator
\begin{equation}\label{fw:eq:chi-commutator}
  \chi_{g}(\gamma)
  \;=\;\hat{\mathcal D}_{g}^{-1}\,\sigma(\gamma)\,
       \hat{\mathcal D}_{g}\,\sigma(\gamma)^{-1}
  \;=\;[\hat{\mathcal D}_{\gamma},\hat{\mathcal D}_{g}]
  \;=\;e^{\,i\,\omega(\gamma,g)} .
\end{equation}
By the exact sequence \eqref{fw:eq:normalizer-ses}, $g$ is covered by
a logical gate iff $\bar g^{*}\mathbb L_{\cS}\simeq\mathbb L_{\cS}$;
by Proposition~\ref{fw:prop:flat-twist} this holds iff the flat twist
$\mathbb L_{\chi_{g}}$ is trivial, i.e.\ iff $\chi_{g}$ is trivial:
$\omega(\gamma_i,g)\in2\pi\mathbb Z$ for
$i=1,2$.  Expanding $g=\alpha\gamma_1+\beta\gamma_2$ and using
$\omega(\gamma_1,\gamma_2)=2\pi d$ gives $d\alpha,d\beta\in\mathbb Z$, i.e.\
$dg\in\Gamma$.  Hence
\begin{equation}\label{fw:eq:torsion}
  G_{\cS} \;=\; \CS[d],
\end{equation}
the $d$-torsion points of the elliptic curve.  Equivalently
$G_{\cS}=N^{\cS}_{G}(\Gamma)/\Gamma=H(\mathbb L_{\cS})$ is Mumford's
$H$-group---the translations of the elliptic curve preserving the
frame bundle,
$H(\mathbb L_{\cS})=\{\bar g\in\CS:
\bar g^{*}\,\mathbb L_{\cS}\simeq\mathbb L_{\cS}\}$---and the logical
sequence \eqref{fw:eq:main-ses} specializes to
\begin{equation}\label{fw:eq:theta-unitary}
  1 \longrightarrow \Us(1) \longrightarrow \cD_{\cS}
    \longrightarrow \CS[d] \longrightarrow 1 .
\end{equation}
Extending scalars by Theorem~\ref{fw:thm:theta},
$\cD_{\cS}\times_{Z}\mathbb C^{\times}\simeq G(\mathbb L_{\cS})$,
this becomes Mumford's \emph{theta group} of the elliptic curve,
\begin{equation}\label{fw:eq:theta-classical}
  1 \longrightarrow \mathbb C^{\times}
    \longrightarrow G(\mathbb L_{\cS})
    \longrightarrow \CS[d] \longrightarrow 1 :
\end{equation}
the celebrated central extension of the $d$-torsion group \emph{is}
the logical Pauli group of the GKP code, with its phases extended
from $\Us(1)$ to $\mathbb C^{\times}$.

~\\

\section{Code words: elliptic-curve equations and the canonical ring}
\label{fw:codewords}

The Kodaira map \eqref{fw:eq:kodaira-map} realizes the compact phase
space as a projective algebraic variety inside the projectivized
code space, so the code words satisfy polynomial equations---and
every such equation is a polynomial identity among coherent-state
overlaps.  This section makes those equations explicit at genus one
and at genus two.  At genus one they are the
classical equations of elliptic curves, dictated by the
representation theory of the theta group of
Sec.~\ref{fw:sub:genus1}, following Mumford~\cite{Mumford66_fw}.  At
genus two the same algebraicity is governed by the canonical ring,
which we compute in closed form; its first output is the conic
$Y^{2}=XZ$ satisfied by the $k=2$ code words, and its last the
equation of the hyperelliptic curve itself, recovered as an
identity among code words.  A closing subsection supplies the
analytic construction of the code words on every phase space
uniformized by the hyperbolic disk: Poincar\'e theta series,
hyperbolic analogues of the Jacobi theta functions.

\subsection{Genus one: equations of elliptic curves}
\label{fw:sub:ellequations}

The identifications of this subsection---GKP code words as theta
functions for a system of multipliers, with the code words
realizing the projective embedding of the curve---were developed
by one of us in the notebook
\texttt{quantization-of-elliptic-curves} (dated March 21, 2024) of
Ref.~\cite{Roberts2024_SM}, and independently by Conrad, Burchards,
and Flammia~\cite{Conrad2024_SM}.  While completing the present
manuscript we became aware of the work of Mayrand and
Royer~\cite{MayrandRoyer2026_SM}, which develops the genus-one
theory systematically from the perspective of complex abelian
varieties.

\emph{Code words are theta functions.}---The stellar function of a
code state of the GKP code~\cite{GKP01_fw} obeys the
quasi-periodicity condition
\begin{equation}\label{fw:eq:quasiper}
  f^{\psi}(z+\gamma_j)
  \;=\;z_{\gamma_j}^{-1}\,
  e^{\,|\gamma_j|^{2}/2+\gamma_j^{*}z}\,f^{\psi}(z),
\end{equation}
with $z_{\gamma_j}$ the lift phase [the automorphy law
\eqref{fw:eq:automorphy} read for the Weyl operators]: the code
space is the space of theta functions of level $d$.  A natural basis---the logical
$Z$-eigenstates---is
\begin{equation}\label{fw:eq:zbasis}
  X_j(z)\;=\;e^{\gamma_1^{*}z^{2}/2\gamma_1}\,
  \vartheta\!\Bigl(\frac{z}{\gamma_1}-\frac{j}{d};\,
  \frac{\tau}{d}\Bigr),
  \qquad j=0,\dots,d-1,
\end{equation}
where $\vartheta(u;\tau)=\sum_{n\in\mathbb Z}
e^{i\pi n^{2}\tau+2\pi inu}$ is the Jacobi theta function and
$\tau=\gamma_2/\gamma_1$ the modular parameter: all basis elements
share one Gaussian prefactor and differ only by the shift of $j/d$
in the argument.

\emph{The Kodaira embedding.}---In this basis the Kodaira map
\eqref{fw:eq:kodaira-map} reads
\begin{equation}\label{fw:eq:theta-embedding}
  \Psi\colon\ \CS\longrightarrow\PP^{d-1},
  \qquad
  z\;\longmapsto\;\bigl[X_0(z):\cdots:X_{d-1}(z)\bigr],
\end{equation}
the coherent-state family seen from inside the code space.  The map
is an embedding precisely for $d\geq3$: a line bundle of degree at
least $2g+1$ on a genus-$g$ curve is \emph{very ample}, meaning its
sections separate points and tangent vectors, so that the Kodaira
map is an embedding.  At $d=2$ the map instead exhibits $\CS$ as a
$2:1$ cover of $\PP^{1}$---the genus-one shadow of the hyperelliptic
map that will organize the genus-two case below.

\emph{The theta group cuts out the curve.}---The generating logical
gates $\hat X=\hat{\mathcal D}_{\gamma_1/d}$ and
$\hat Z=\hat{\mathcal D}_{\gamma_2/d}$, lifting the $d$-torsion
translations \eqref{fw:eq:torsion}, act on the basis by
$\hat X\colon X_j\mapsto X_{j+1}$ and
$\hat Z\colon X_j\mapsto\omega^{j}X_j$ with $\omega=e^{2\pi i/d}$,
so the image of $\Psi$ is constrained by the representation theory
of the theta group \eqref{fw:eq:theta-classical}.  For $d=3$ the
constraint forces the image cubic in $\PP^{2}$ into the \emph{Hesse
normal form}
\begin{equation}\label{fw:eq:hesse}
  X_0^{3}+X_1^{3}+X_2^{3}\;=\;\mu\,X_0X_1X_2,
  \qquad
  \mu=\frac{q(0)^{3}+q(1)^{3}+q(2)^{3}}{q(0)\,q(1)\,q(2)},
\end{equation}
with theta constants $q(j)=\vartheta(-j/3,\tau/3)$: the Hesse
pencil is precisely the family of plane cubics invariant under
$\hat X$ and $\hat Z$.  For $d=4$ the image in $\PP^{3}$ is the
intersection of two Heisenberg-symmetric quadrics,
\begin{equation}\label{fw:eq:d4quadrics}
  X_1^{2}+X_3^{2}=2\lambda\,X_0X_2,
  \qquad
  X_0^{2}+X_2^{2}=2\lambda\,X_1X_3,
  \qquad
  \lambda=\frac{q(1)^{2}}{q(0)\,q(2)},
\end{equation}
and Mumford proves~\cite{Mumford66_fw} that whenever $4\,|\,d$ the
image is cut out by quadric hypersurfaces.

\subsection{Genus two: the canonical ring}\label{fw:sub:canonicalring}

At genus two the code spaces at all integer weights assemble into
the \emph{canonical ring}
\begin{equation}\label{fw:eq:canring-def}
  R(\CS)\;=\;\bigoplus_{k\geq0}
  \Gamma\bigl(\CS,\Omega_{\CS}^{\otimes k}\bigr),
\end{equation}
the graded ring in which code words of weights $k$ and $k'$
multiply to code words of weight $k+k'$
[Eq.~\eqref{fw:eq:hyp-sections}].  The main text generates the
entire tower of code spaces from a conic basis of the $k=2$ code
space; the two propositions below are exactly the facts used
there.  Both are proved by pushing the code line bundle
$\Ls\simeq\Omega_{\CS}^{\otimes k}$ of Sec.~\ref{fw:hyperbolic}
[Eq.~\eqref{fw:eq:hyp-sections}] down the hyperelliptic covering
and reading off the $\pm1$ eigenspaces of the hyperelliptic
involution.

\begin{proposition}[Square roots]\label{fw:prop:sqroots}
Let $X,Y,Z$ be a basis of $\Gamma(\CS,\Omega_{\CS}^{\otimes2})$
satisfying the conic relation $Y^{2}=XZ$.  Then $X$ and $Z$ admit
holomorphic square roots; any choice of square roots
$\omega_0,\omega_1$ is a basis of $\Gamma(\CS,\Omega_{\CS})$, and
$Y=\pm\,\omega_0\omega_1$.
\end{proposition}

\begin{proposition}[Canonical ring]\label{fw:prop:canring}
Let $\omega_0,\omega_1$ be any basis of $\Gamma(\CS,\Omega_{\CS})$
and let $W=\omega_0\omega_1'-\omega_1\omega_0'$ be their
Wronskian.  Then for every $k\geq1$ the $k$-differentials
\begin{equation}\label{fw:eq:canonical-basis}
  \omega_0^{\,k-j}\omega_1^{\,j}
  \quad (j=0,\dots,k),
  \qquad\qquad
  W\,\omega_0^{\,k-j-3}\omega_1^{\,j}
  \quad (j=0,\dots,k-3;\ k\geq3),
\end{equation}
form a basis of $\Gamma(\CS,\Omega_{\CS}^{\otimes k})$
[cf.~Ref.~\cite{KockTait_fw}, Thm.~5.1].
\end{proposition}

\emph{The pushforward computation.}---At genus two the canonical
bundle is basepoint-free [Ref.~\cite{Hartshorne_fw},
Prop.~IV.5.1] with $h^{0}(\Omega_{\CS})=2$, so its sections define
a morphism $\pi\colon\CS\to\PP^{1}$ with
\begin{equation}\label{fw:eq:KC-pullback}
  \Omega_{\CS}\;\simeq\;\pi^{*}\mathcal O_{\PP^{1}}(1)
\end{equation}
[Ref.~\cite{Hartshorne_fw}, Thm.~II.7.1], of degree
$\deg\Omega_{\CS}=2$: the hyperelliptic covering.
The covering involution $\iota$ is the hyperelliptic involution,
and the branch divisor $B\subset\PP^{1}$ has degree six
(Riemann--Hurwitz).  The one genuine input is the eigensheaf
decomposition of a double cover~\cite{Pardini91_fw},
\begin{equation}\label{fw:eq:eigensheaf}
  \pi_{*}\mathcal O_{\CS}
  \;=\;\mathcal O_{\PP^{1}}\oplus L^{-1},
  \qquad
  L^{\otimes2}\simeq\mathcal O(B),
\end{equation}
whose summands are the $\pm1$ eigensheaves of $\iota$; $\deg B=6$
gives $L\simeq\mathcal O(3)$.  By the projection formula,
\begin{equation}\label{fw:eq:pushforward-chain}
  \pi_{*}\bigl(\Omega_{\CS}^{\otimes k}\bigr)
  \;\simeq\;\pi_{*}\bigl(\pi^{*}\mathcal O_{\PP^{1}}(k)\bigr)
  \;\simeq\;\pi_{*}\mathcal O_{\CS}\otimes\mathcal O(k)
  \;\simeq\;\mathcal O(k)\oplus\mathcal O(k-3),
\end{equation}
and since $\pi$ is finite, taking sections upstairs is the same as
taking sections of the pushforward downstairs:
\begin{equation}\label{fw:eq:sections-split}
  \Gamma\bigl(\CS,\Omega_{\CS}^{\otimes k}\bigr)
  \;\simeq\;
  \Gamma\bigl(\PP^{1},\mathcal O(k)\bigr)
  \;\oplus\;
  \Gamma\bigl(\PP^{1},\mathcal O(k-3)\bigr).
\end{equation}
For $k\geq2$ the dimension count $(k+1)+(k-2)=(2k-1)(g-1)$, with
$h^{0}(\mathcal O(m))=0$ for $m<0$, reproduces Riemann--Roch
\eqref{fw:eq:RR}; at $k=1$ it gives
$\dim\Gamma(\CS,\Omega_{\CS})=2=g$, where the correction term in
Riemann--Roch is nonvanishing.

Two facts about the splitting \eqref{fw:eq:sections-split} drive
the proofs of Propositions~\ref{fw:prop:sqroots}
and~\ref{fw:prop:canring}.

(a)~The summand $\Gamma(\PP^{1},\mathcal O(k))$ consists of the
homogeneous degree-$k$ polynomials in $\omega_0,\omega_1$.
Indeed, the $k$-th tensor power of \eqref{fw:eq:KC-pullback}
identifies $\Omega_{\CS}^{\otimes k}\simeq\pi^{*}\mathcal O(k)$
for every $k$ at once, compatibly with multiplication of
sections---the product of pullbacks is the pullback of the
product---so if $\omega_0,\omega_1$ denote the abelian
differentials corresponding to the homogeneous coordinates $s,t$,
the section corresponding to a homogeneous degree-$k$ polynomial
$p(s,t)$ is $p(\omega_0,\omega_1)$.  A $\GL_2$ change of the
coordinates $s,t$ realizes every basis of
$\Gamma(\CS,\Omega_{\CS})$ this way, so the statement holds for an
arbitrary basis.

(b)~The two summands are $\iota$-eigenspaces of opposite parity.
The hyperelliptic involution acts on abelian differentials by
$\iota^{*}\omega=-\omega$,\footnote{\label{fw:fn:iota}Two
preliminaries.  At a fixed point of $\iota$ (a ramification
point), any centered local coordinate $v$ can be antisymmetrized:
$u:=\tfrac12(v-v\circ\iota)$ satisfies $u\circ\iota=-u$, and it is
again a coordinate because $d\iota=-1$ at a fixed point [$d\iota$
squares to $1$ since $\iota^{2}=\mathrm{id}$, and $d\iota=+1$
would make $\tfrac12(v+v\circ\iota)$ an $\iota$-invariant
coordinate, forcing $\iota=\mathrm{id}$ near the point and hence,
by the identity theorem, everywhere].  In this coordinate the
covering reads $w=u^{2}$, with $w$ a coordinate downstairs.  Now
let $\omega$ be an abelian differential and set
$\eta=\omega+\iota^{*}\omega$, which is $\iota$-invariant; we
claim $\eta=0$.  Away from the ramification points $\pi$ is an
unbranched two-sheeted covering, so $\eta$ is the pullback of a
holomorphic differential $\xi$ on the complement of $B$.  Near a
ramification point write $\eta=h(u)\,du$; invariance gives
$-h(-u)=h(u)$, so $h(u)=u\,g(u^{2})$ and
$\eta=\tfrac12\,g(w)\,dw$: $\xi$ extends holomorphically across
$B$.  Since $\PP^{1}$ carries no nonzero holomorphic
differential, $\xi=0$, hence $\eta=0$.  See also
Ref.~\cite{FarkasKra_fw}, Sec.~III.7.}  so the polynomial summand
$\Gamma(\mathcal O(k))$ carries $\iota$-eigenvalue $(-1)^{k}$,
and the complementary summand $\Gamma(\mathcal O(k-3))$---which
differs from it by the odd eigensheaf $L^{-1}$ of
\eqref{fw:eq:eigensheaf}---carries the opposite eigenvalue
$(-1)^{k+1}$.

\emph{Low weights.}---At $k=1$ the odd summand
$\Gamma(\mathcal O(-2))$ is empty.  At $k=2$ the odd summand
$\Gamma(\mathcal O(-1))$ is still empty, so by reading~(a) every
$k=2$ code word is a homogeneous quadratic polynomial in any basis
$\omega_0,\omega_1$:
\begin{equation}\label{fw:eq:bicanonical-pullback}
  \Gamma\bigl(\CS,\Omega_{\CS}^{\otimes2}\bigr)
  \;=\;\operatorname{span}\{\omega_0^{2},\;\omega_0\omega_1,\;
        \omega_1^{2}\},
\end{equation}
i.e.\ the multiplication map
$\operatorname{Sym}^{2}\Gamma(\Omega_{\CS})\to
\Gamma(\Omega_{\CS}^{\otimes2})$ is an isomorphism, and the
bicanonical Kodaira map factors through the degree-two Veronese
map $\nu_2\colon[s:t]\mapsto[s^{2}:st:t^{2}]$:
\begin{equation}\label{fw:eq:veronese-factor}
  \CS\;\xrightarrow{\ \pi\ }\;\PP^{1}
  \;\xrightarrow{\ \nu_2\ }\;\PP^{2}.
\end{equation}
In the basis $X=\omega_0^{2}$, $Y=\omega_0\omega_1$,
$Z=\omega_1^{2}$ the image is the smooth conic
\begin{equation}\label{fw:eq:veronese-conic}
  Y^{2}\;=\;XZ
\end{equation}
(the conic relation of the main text): the $k=2$ code words of a
genus-two code satisfy a single quadratic identity---the genus-two
manifestation of the same algebraicity that produces the Hesse
cubic~\eqref{fw:eq:hesse} and the
quadrics~\eqref{fw:eq:d4quadrics} at genus one.

\begin{proof}[Proof of Proposition~\ref{fw:prop:sqroots}]
By \eqref{fw:eq:bicanonical-pullback} and reading~(a), the given
basis corresponds to a basis $q_X,q_Y,q_Z$ of the homogeneous
quadratic polynomials in $s,t$; because the inclusion of sections
is injective and multiplicative, the conic relation descends to
$q_Y^{2}=q_X\,q_Z$.  Now factor into linear factors in the
polynomial ring $\mathbb C[s,t]$ (a unique factorization domain),
writing $q_Y=m_1m_2$ with $m_1,m_2$ homogeneous linear.  Then
$q_X$ is proportional to $m_1^{2}$, $m_1m_2$, or $m_2^{2}$; the
middle case makes $q_X$ proportional to $q_Y$, and
$m_1\propto m_2$ makes $q_X$ proportional to $q_Z$---both
excluded by linear independence.  Hence, after rescaling the
linear factors, $q_X=m_1^{2}$ and $q_Z=m_2^{2}$ with $m_1,m_2$
independent, and $q_Y=\pm m_1m_2$.  The abelian differentials
$\omega_0,\omega_1$ corresponding to $m_1,m_2$ therefore satisfy
$\omega_0^{2}=X$, $\omega_1^{2}=Z$, and
$\omega_0\omega_1=\pm Y$; being independent, they are a basis of
the two-dimensional space $\Gamma(\CS,\Omega_{\CS})$.  Finally,
square roots of a nonzero holomorphic section are unique up to
sign on a connected domain, so \emph{every} choice of square roots
is $(\pm\omega_0,\pm\omega_1)$; in particular the single-valued
square roots $\sqrt X,\sqrt Z$ extracted on $\DD$ in the main text
automatically transform as abelian differentials.
\end{proof}

\emph{The Wronskian.}---For any basis $\omega_0,\omega_1$, writing
$\omega_i=f_i\,dz$ in a local coordinate, the \emph{Wronskian}
\begin{equation}\label{fw:eq:wronskian}
  W\;=\;\omega_0\,\omega_1'-\omega_1\,\omega_0'
  \;:=\;\bigl(f_0f_1'-f_1f_0'\bigr)\,(dz)^{3}
\end{equation}
is well defined: under a coordinate change $z=z(w)$ each
coefficient becomes $g_i=f_i\,z'$, and the $z''$ cross terms
cancel in the antisymmetric combination,
$g_0g_1'-g_1g_0'=\bigl(f_0f_1'-f_1f_0'\bigr)(z')^{3}$---exactly
the transformation law of a cubic differential.  It is nonzero,
since $f_0f_1'=f_1f_0'$ would force $\omega_1/\omega_0$ to be
constant; and a change of basis multiplies it by the determinant,
so $W$ is canonical up to scale.  Its parity is read off at a
ramification point, in the antisymmetrized coordinate $u$ with
$\iota(u)=-u$ of footnote~\ref{fw:fn:iota}: there $du$ is odd and
each $\omega_i$ is odd, so each $f_i$ is even, each $f_i'$ is odd,
and the coefficient $h=f_0f_1'-f_1f_0'$ is odd.  Hence
$\iota^{*}W=-h(-u)\,(du)^{3}=h(u)\,(du)^{3}=W$ near the point, and
therefore on all of $\CS$ by the identity theorem: $W$ is
$\iota$-\emph{invariant}, the opposite parity to the cubic
monomials in $\omega_0,\omega_1$.  Moreover $h$ odd gives
$h(0)=0$, so $W$ vanishes, to odd order, at each of the six
ramification points; since $\deg\Omega_{\CS}^{\otimes3}=6$, each
zero is simple and there are no others.

\begin{proof}[Proof of Proposition~\ref{fw:prop:canring}]
The monomials $\omega_0^{k-j}\omega_1^{j}$ are the images of the
degree-$k$ monomials in $s,t$ [reading~(a)], hence $k+1$
independent elements of $\iota$-parity $(-1)^{k}$.
Multiplication by the nonzero section $W$ is injective, so the
$k-2$ elements $W\,\omega_0^{k-j-3}\omega_1^{j}$ are independent,
with parity $(+1)\cdot(-1)^{k-3}=(-1)^{k+1}$.  The two families
have opposite parities, hence are jointly independent, and their
total count---$(k+1)+(k-2)$ for $k\geq3$, $k+1$ for $k\leq2$---equals
$\dim\Gamma(\CS,\Omega_{\CS}^{\otimes k})$ by
\eqref{fw:eq:sections-split}; together they form a basis.
\end{proof}

\emph{The relation.}---The square $W^{2}$ is $\iota$-invariant of
weight six, so by Proposition~\ref{fw:prop:canring} it is a
homogeneous sextic polynomial in $(\omega_0,\omega_1)$:
\begin{equation}\label{fw:eq:hyperelliptic-relation}
  W^{2}\;=\;f_6(\omega_0,\omega_1)
\end{equation}
for a unique such sextic $f_6$.  Its roots are exactly the six
branch points: $\operatorname{div}W^{2}$ is twice the
ramification divisor, which is $\pi^{*}B$, so $f_6$ has simple
roots at $B$.  Equation~\eqref{fw:eq:hyperelliptic-relation} is
thus the classical hyperelliptic equation of the curve, recovered
as an identity among code words in exact parallel with the theta
relations of genus one.  Altogether,
\begin{equation}\label{fw:eq:canring}
  R(\CS)\;\simeq\;
  \mathbb C[\omega_0,\omega_1,W]\,\big/\,
  \bigl(W^{2}-f_6(\omega_0,\omega_1)\bigr),
\end{equation}
with $\omega_0,\omega_1$ in degree one and $W$ in degree three:
every code space at every integer weight, together with its
multiplicative structure, is generated by three code words subject
to the single relation \eqref{fw:eq:hyperelliptic-relation}.

\subsection{Hyperbolic theta series}\label{fw:sub:thetaseries}

The preceding subsections pin down equations that the stellar
functions of the code words \emph{satisfy}; we close with the
objects that \emph{produce} the code words.
At genus one, every code word is built from the Jacobi theta
function [Eq.~\eqref{fw:eq:zbasis}]---a sum over the abelian
stabilizer lattice.  The \emph{Poincar\'e theta series} furnish
an analogue of these sums on every phase space uniformized by
the hyperbolic disk, i.e.\ on every compact curve of genus
$g\geq2$, the genus-two case of this work being only the first
instance.  Here the stabilizer group is nonabelian, and the sums
are organized by its conjugacy classes---equivalently, by the
free-homotopy classes of closed geodesics of $\CS$: to the class
$[\gamma]\subset\Gamma$ of a hyperbolic element one attaches the
sum over the class,
\begin{equation}\label{eq:theta_def}
  \theta_\gamma(z)
  = \sum_{\nu \in [\gamma]}
    \bigl(\bar{b}\,z^2 + (\bar{a} - a)z - b\bigr)^{-k},
  \quad
  \nu = \begin{pmatrix} a & b \\ \bar{b} & \bar{a} \end{pmatrix}\!.
\end{equation}

\emph{Convergence.}---Each element $\nu$ of the class has a
\emph{height} $h(\nu) = |b(\nu)|$, and the summand decays as
$|b|^{-k}$ for large~$|b|$, so convergence is controlled by the
\emph{filtered} partial sums
\begin{equation}\label{eq:filtered}
  \theta_\gamma(z;\, B)
  = \sum_{\substack{\nu \in [\gamma] \\[2pt]
     h(\nu) < B}}
     \bigl(\bar{b}\,z^2 + (\bar{a} - a)z - b\bigr)^{-k},
\end{equation}
which truncate the class at height~$B$.  The number of class
elements below height~$B$ is governed by the prime geodesic
theorem: for a cocompact Fuchsian group and a \emph{primitive}
hyperbolic $\gamma$---one that is not a proper power in
$\Gamma$---Good's
formula~\cite{Good2015,ChatzakosPetridis2016} gives
\begin{equation}\label{eq:good}
  \#\bigl\{\nu\in[\gamma] : h(\nu)<B\bigr\}
  = \frac{2\ell(\gamma)}{\mathrm{Vol}(\CS) \cdot
    \sinh\bigl(\ell(\gamma)/2\bigr)}\,B
  + O\!\left(B^{2/3}\right),
\end{equation}
where $\ell(\gamma)$ is the length of the closed geodesic. The
class therefore has constant density in height, and the tail is
bounded, for $k\geq2$, by
\begin{equation}\label{eq:tail}
  \bigl|\theta_\gamma(z) - \theta_\gamma(z;\, B)\bigr|
  \;\leq\;
  C(z) \int_B^\infty h^{-k}\,dh
  \;=\;
  \frac{C(z)}{(k{-}1)\,B^{k-1}},
\end{equation}
with $C(z)$ bounded on compact subsets of~$\DD$: the series
converges absolutely for every $k\geq2$, uniformly on compacts,
and $\theta_\gamma$ is a weight-$k$ code word.
At $k=1$ the sum
diverges~\cite{Iwaniec2002_SM}, so the abelian differentials
cannot be produced this way directly---they are recovered from
the $k=2$ code space through
Proposition~\ref{fw:prop:sqroots}.

\section{Case study: the Bolza code}\label{sec:bolza}

This section collects the explicit data of the Bolza code: the
stabilizer generators, the pants theta-series basis and its
numerical convergence, the conic factorization that produces the
code words, and the Cartan (rotation--squeeze) form of the
Gaussian logical gates.
Throughout, $\CS$ is the Bolza surface, the most symmetric
genus-two curve, presented as the regular hyperbolic octagon with
opposite sides identified.
The pants decomposition below and the resulting $k=2$
theta-series code words were developed by one of us in the
notebook \texttt{genus-two-quantization} (dated June 8, 2024) of
Ref.~\cite{Roberts2024_SM}.

\subsection{Stabilizer generators}\label{sec:stabilizers}

The Fuchsian group $\Gamma$ of the Bolza surface is generated by
four hyperbolic translations $\gamma_0, \ldots, \gamma_3$ through
the origin at angles $j\pi/4$, all with translation
length $\varphi = 2\,\mathrm{arccosh}(1 + \sqrt{2})$.  Their
$\SU(1,1)$ matrices are
\begin{equation}\label{eq:gamma_matrices}
  \gamma_j =
  \begin{pmatrix}
    1 + \sqrt{2} & \mu\, e^{ij\pi/4} \\[2pt]
    \mu\, e^{-ij\pi/4} & 1 + \sqrt{2}
  \end{pmatrix}\!,
  \quad j = 0, 1, 2, 3,
\end{equation}
where
\begin{equation}
  \mu = (2 + \sqrt{2})\sqrt{\sqrt{2} - 1}.
\end{equation}
One verifies $\det \gamma_j = (1+\sqrt{2})^2 - (2+\sqrt{2})^2
(\sqrt{2}-1) = 1$ and
$\cosh(\varphi/2) = 1 + \sqrt{2}$, consistent with the stated
translation length.  Numerically, $\mu \approx 2.197$ and
$\varphi \approx 3.057$. As noted in the main text, these generators correspond to the $\mathrm{SU}(1,1)$ stabilizers $\hat{\mathcal{D}}_j = \exp(\xi_j\,\Kp - \xi_j^{*}\,\Km)$ with
$\xi_j = (\varphi/2)\,e^{ij\pi/4}$.

Throughout, an element
$g = \bigl(\begin{smallmatrix} a & b \\ \bar{b} & \bar{a}
\end{smallmatrix}\bigr) \in \SU(1,1)$ acts on $\DD$ by
$g \cdot z = (az + b)/(\bar{b}z + \bar{a})$.

\subsection{Pants basis and numerical convergence}
\label{sec:pants}

For the Bolza surface we choose the three systolic geodesics as
the pants curves of Sec.~\ref{fw:sub:thetaseries}; in terms of the
stabilizer generators~$\gamma_j$ [Eq.~\eqref{eq:gamma_matrices}],
their conjugacy classes are represented by
\begin{equation}\label{eq:pants_words}
  \gamma_G = \gamma_1,\qquad
  \gamma_B = \gamma_3^{-1}\gamma_2,\qquad
  \gamma_R = \gamma_1\gamma_2^{-1}\gamma_3.
\end{equation}
All three have the same translation length
$\ell_{\mathrm{sys}} = 2\,\mathrm{arccosh}(1 + \sqrt{2})$, as can
be verified from $|\mathrm{tr}\,\gamma_\ell| = 2(1+\sqrt{2})$; by
Good's formula~\eqref{eq:good}, the three theta series therefore
converge at the same rate.  Figure~\ref{fig:pants} shows the
corresponding geodesic arcs on the Bolza octagon.  At $k = 2$, the
three series $\theta_R, \theta_G, \theta_B$ are a basis of
$\Gamma(\CS,\Omega_{\CS}^{\otimes 2})$
(Sec.~\ref{fw:sub:thetaseries}).

\begin{figure}[t]
  \centering
  \includegraphics[width=0.5\columnwidth]{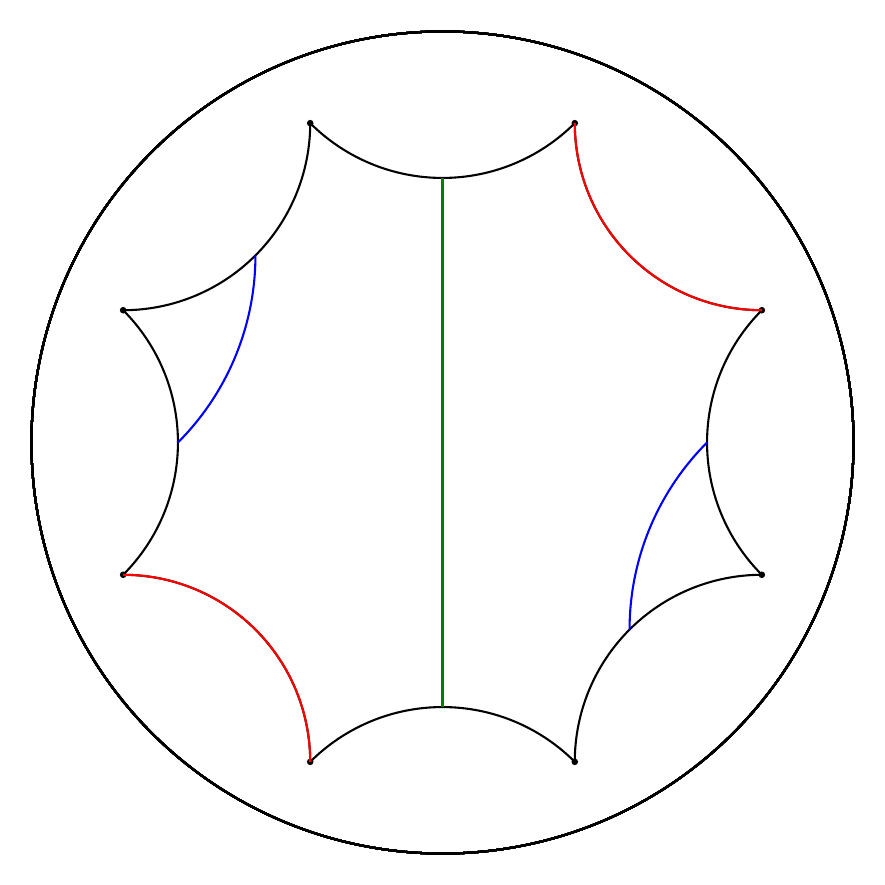}
  \caption{\label{fig:pants}%
    Pants decomposition of the Bolza octagon.  The three
    systolic geodesics $\gamma_R$ (red), $\gamma_G$ (green),
    $\gamma_B$ (blue) cut the surface into two pairs of pants.
    Their conjugacy classes in~$\Gamma$ define the Poincar\'e
    theta series~\eqref{eq:theta_def}.}
\end{figure}

\emph{Numerical convergence.}---Figure~\ref{fig:convergence}(a)
confirms the truncation rate~\eqref{eq:tail} at $k = 2$.  We
evaluate $\theta_R, \theta_G, \theta_B$ at 10~random sample points
in the interior of~$\DD$ ($|z| < 0.6$) for height cutoffs
$B \in \{50, 100, \ldots, 20000\}$, using $B_{\mathrm{ref}} = 50000$
as a proxy for the exact series.  At each point, the
\emph{relative truncation error}
\begin{align}
  \varepsilon(B;\, z)
  &= \max_{\ell \in \{R,G,B\}}\,
  \frac{\bigl|\theta_\ell(z;\, B)
        - \theta_\ell(z;\, B_{\mathrm{ref}})\bigr|}
       {\bigl|\theta_\ell(z;\, B_{\mathrm{ref}})\bigr|}
\end{align}
is computed by comparing the partial sum~\eqref{eq:filtered} at
cutoff~$B$ to the reference value.  On a log-log plot
(Fig.~\ref{fig:convergence}a), the data follow a clean power-law
decay $\varepsilon \propto B^{-\beta}$ over more than two decades,
with $\beta \approx 1$.
At $B = 10^4$
($\sim\!2000$ terms per conjugacy class), the median
error is $\sim\!10^{-5}$.

\begin{figure}[t]
  \centering
  \includegraphics[width=0.7\columnwidth]{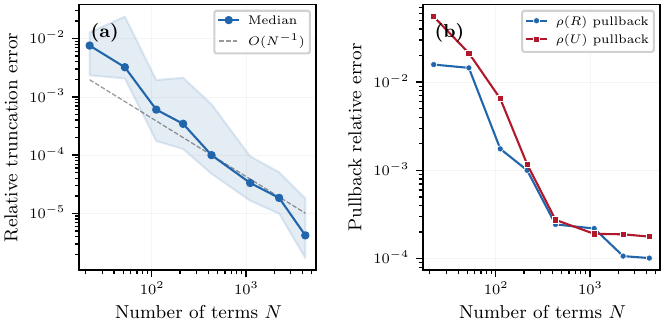}
  \caption{Convergence of the Poincar\'e series and pullback
    verification.
    \textbf{(a)}~Relative truncation error of the theta series
    vs.\ the number of terms~$N$ (blue: median over sample points;
    shaded band: pointwise range).
    \textbf{(b)}~Pullback verification error: maximum pointwise
    relative error in the weight-2 relation~\eqref{eq:pullback_check}
    for $\rho(R)$ (blue) and $\rho(U)$ (red).
    Panel~(a) shows $O(N^{-1})$ decay (dashed).
    The pullback errors in~(b) follow a similar initial decay but
    saturate at ${\sim}\,10^{-4}$: the Veronese basis
    change~\eqref{eq:XYZ_SM} depends on the conic
    parameter~$\lambda$, extracted by SVD at a fixed truncation,
    whose fitting residual sets a noise floor.}
  \label{fig:convergence}
\end{figure}

\subsection{Conic factorization}\label{sec:conicfactorization}

Section~\ref{fw:sub:canonicalring} builds the entire tower of code
spaces from a conic basis $X,Y,Z$ of the $k=2$ code space; here we
produce that basis for the Bolza code.  Write
$(X, Y, Z)^T = A\,(\theta_R, \theta_G, \theta_B)^T$ for a
$3 \times 3$ change-of-basis matrix~$A$.  The Veronese
condition~\eqref{fw:eq:veronese-conic} becomes a quadratic relation
\begin{equation}\label{eq:conic_theta}
  \bm{\theta}^T M\, \bm{\theta} = 0
\end{equation}
in the theta basis, where $M$ is a symmetric $3\times 3$ matrix
determined by the inner products of the theta series.

For the Bolza surface, the octahedral symmetry permutes the three
pants geodesics cyclically, forcing
$M = (1{-}\lambda)I + \lambda J$ where $J$ is the all-ones matrix.
The conic is then
$\theta_R^2 + \theta_G^2 + \theta_B^2
+ 2\lambda(\theta_R\theta_G + \theta_R\theta_B
+ \theta_G\theta_B) = 0$,
governed by a single parameter~$\lambda$.  Non-degeneracy
($\det M = (1{-}\lambda)^2(1{+}2\lambda) \neq 0$) ensures the conic
is smooth and hence rational, admitting the Veronese
parametrization~\eqref{fw:eq:veronese-factor}.

\emph{Extracting the change-of-basis matrix.}---The $C_3$ symmetry
of the conic decomposes the theta basis into irreducible
representations.  The symmetric combination
$Y \propto \theta_R + \theta_G + \theta_B$ is the trivial
representation, while $X$ and $Z$ lie in the two-dimensional
representation and differ only in the sign of one coefficient.
The explicit decomposition (quoted in the main text) is
\begin{align}
  Y &= \tfrac{\alpha}{3}\,(\theta_R + \theta_G + \theta_B),
      \notag \\[3pt]
  X &= \tfrac{\gamma}{6}\,(\theta_R {+} \theta_G {-}
       2\theta_B)
       + \tfrac{\beta}{2}\,(\theta_R {-} \theta_G),
       \label{eq:XYZ_SM} \\[3pt]
  Z &= \tfrac{\gamma}{6}\,(\theta_R {+} \theta_G {-}
       2\theta_B)
       - \tfrac{\beta}{2}\,(\theta_R {-} \theta_G),
       \notag
\end{align}
with $\alpha = \sqrt{6\lambda{+}3}$,
$\beta = \sqrt{2{-}2\lambda}$,
$\gamma = \sqrt{6\lambda{-}6}$.
Numerically, the SVD of the $6$-monomial matrix over $400$ sample
points has rank $5$ (singular-value ratio
$S_5/S_4 \sim 10^{-4}$), and its null vector gives
\begin{equation}
  \lambda = -1.0142(1) - 0.2385(1)\,i,
\end{equation}
where the quoted uncertainty is the observed drift of $\lambda$
under variation of the height cutoff over
$B = 2\times10^4$--$2\times10^5$ ($< 7\times10^{-5}$) and of the
sample grid (${\sim}5\times10^{-5}$).

\emph{Square roots and the Wronskian.}---By
Proposition~\ref{fw:prop:sqroots}, $X$ and $Z$ admit single-valued
holomorphic square roots on the simply connected cover~$\DD$,
unique up to an overall sign, and any choice
$\omega_0 = \sqrt{X}$, $\omega_1 = \sqrt{Z}$ is a basis of the
$k=1$ code space; each has exactly $\deg \Omega_{\CS} = 2$ zeros.
The third generator of the canonical ring is the Wronskian
\begin{equation}
  W = \omega_0\,\omega_1' - \omega_1\,\omega_0'
  = \frac{X Z' - Z X'}{2Y},
\end{equation}
computable rationally from the conic basis, and
Proposition~\ref{fw:prop:canring} then yields explicit code words
at every integer weight.

Figure~\ref{fig:conic} provides a direct test of the conic
relation~\eqref{fw:eq:veronese-conic}.  We evaluate $\theta_R,
\theta_G, \theta_B$ at a grid of ${\sim}\,700$ sample points in the
interior of $\DD$ ($|z| < 0.6$), apply the change of
basis~\eqref{eq:XYZ_SM} to obtain $X, Y, Z$, and form the affine
coordinates $x = X/Z$ and $y = Y/Z$.  The conic $Y^2 = XZ$ in
affine coordinates reads
$\mathrm{Re}(x) = \mathrm{Re}(y)^2 - \mathrm{Im}(y)^2$.
Panel~(a) of Fig.~\ref{fig:conic} plots the measured
$\mathrm{Re}(X/Z)$ against this prediction; the data lie on the
diagonal with fractional residuals $\lesssim 10^{-3}$, limited by
truncation of the Poincar\'e series at finite height
cutoff~$B \sim 3000$.  Panels~(b) and~(c) display the spatial
distribution of the fractional and absolute residuals on the
Poincar\'e disk, confirming uniform accuracy across the fundamental
domain.

\begin{figure}[t]
  \centering
  \includegraphics[width=\columnwidth]{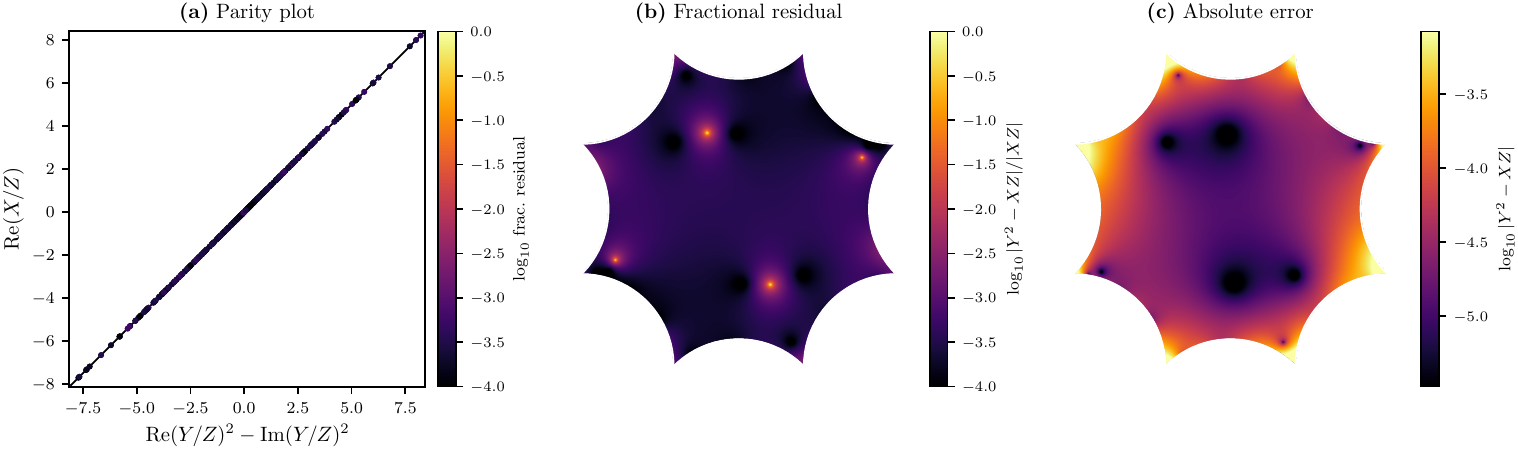}
  \caption{\label{fig:conic}%
    Numerical verification of the Veronese conic $Y^2 = XZ$.
    \textbf{(a)}~Parity plot: measured $\mathrm{Re}(X/Z)$ vs.\ the
    conic prediction
    $\mathrm{Re}(Y/Z)^2 - \mathrm{Im}(Y/Z)^2$; color encodes
    the $\log_{10}$ fractional residual.
    \textbf{(b)}~Spatial map of the fractional residual
    $|Y^2 - XZ|/|XZ|$ on the Poincar\'e disk (white lines:
    Bolza octagon boundary).
    \textbf{(c)}~Spatial map of the absolute error
    $|Y^2 - XZ|$.  The fractional residual in~(b) spikes near
    the zeros of $X$ and~$Z$, as expected from the vanishing
    denominator; panel~(c) confirms that the absolute error
    remains small at these points.  All interior points achieve
    fractional accuracy $\lesssim 10^{-3}$ away from the zeros.}
\end{figure}

\subsection{Gaussian gates}\label{sec:matrices}

\subsubsection{Cartan form of the generators}

The automorphism group $\Aut(\CS) \cong \GL(2, \mathbb{F}_3)$ is
generated by two elements: an order-8 rotation $R$ and an order-3
M\"obius transformation~$U$.

\emph{The rotation $R$}: $z \mapsto e^{i\pi/4}z$, corresponding to
\begin{equation}\label{eq:R_matrix}
  R =
  \begin{pmatrix}
    e^{i\pi/8} & 0 \\[2pt]
    0 & e^{-i\pi/8}
  \end{pmatrix}\!.
\end{equation}
This is simply a rotation of the Poincar\'e disk by $\pi/4$, implemented
physically as
$\hat{\mathcal{D}}_R = \exp(i(\pi/4)\,\Kz)
= \exp\bigl(i(\pi/8)(\hat{n}_a + \hat{n}_b + 1)\bigr)$.

\emph{The order-3 element $U$}: a rotation by $2\pi/3$ about the
point $p = |p|\,e^{i\pi/8}$, where
$|p|^2 = (q - 1)/(q + 1)$ with
$q = \cot(\pi/3)\cot(\pi/8)
= (1 + \sqrt{2})/\sqrt{3}$.  Its $\SU(1,1)$ matrix is
\begin{equation}\label{eq:U_matrix}
  U =
  \begin{pmatrix}
    e^{i3\pi/8}\sqrt{1 + 2^{-1/2}} &
    2^{-1/4}\,e^{-i3\pi/8} \\[2pt]
    2^{-1/4}\,e^{i3\pi/8} &
    e^{-i3\pi/8}\sqrt{1 + 2^{-1/2}}
  \end{pmatrix}\!,
\end{equation}
obtained from the Blaschke decomposition
$U = A(p)^{-1}\,\mathrm{diag}(e^{i\pi/3},
e^{-i\pi/3})\,A(p)$, where $A(p)$ is the Blaschke
automorphism sending $p \mapsto 0$.  One verifies
$|a|^2 - |b|^2 = (1 + 2^{-1/2}) - 2^{-1/2} = 1$.

\emph{Cartan form: the logical gate as a two-mode
squeezer.}---The Cartan ($KP$) decomposition of $\SU(1,1)$
renders any element as a phase rotation times a two-mode squeeze,
and its chart is nothing but the modulus--phase coordinates of
the first row.  Writing $k_z, k_\pm$ for the defining
$2 \times 2$ representation of $\Kz, \hat{K}_\pm$ and multiplying
out the two factors,
\begin{equation}\label{eq:KP_chart}
  e^{i\varphi k_z}\,e^{\xi k_+ - \xi^* k_-}
  =
  \begin{pmatrix}
    e^{i\varphi/2}\cosh r &
    e^{i(\varphi/2 + \theta)}\sinh r \\[2pt]
    e^{-i(\varphi/2 + \theta)}\sinh r &
    e^{-i\varphi/2}\cosh r
  \end{pmatrix}\!,
  \qquad \xi = r\,e^{i\theta},
\end{equation}
so for any
$g = \bigl(\begin{smallmatrix} \alpha & \beta \\
\beta^* & \alpha^* \end{smallmatrix}\bigr) \in \SU(1,1)$ the
inversion is immediate and unique ($|\alpha| = \cosh r \geq 1$ is
automatic from $|\alpha|^2 - |\beta|^2 = 1$):
\begin{equation}\label{eq:KP_inversion}
  \varphi = 2\arg\alpha,
  \qquad
  r = \mathrm{arccosh}\,|\alpha|,
  \qquad
  \theta = \arg\beta - \arg\alpha.
\end{equation}
Applied to Eq.~\eqref{eq:U_matrix}: $\arg\alpha = 3\pi/8$,
$|\alpha| = \sqrt{1 + 2^{-1/2}}$, $\arg\beta = -3\pi/8$, so
\begin{equation}\label{eq:U_cartan}
  \hat{\mathcal D}_U
  = e^{i(3\pi/4)\Kz}\; e^{\xi\Kp - \xi^*\Km},
  \qquad
  \xi = r\,e^{-3i\pi/4},
  \qquad
  r = \mathrm{arcsinh}\bigl(2^{-1/4}\bigr)
    = \tfrac{1}{2}\,\mathrm{arccosh}(1 + \sqrt{2})
    \approx 0.7643
\end{equation}
(equivalently $\cosh 2r = 1 + \sqrt{2} = \cot(\pi/8)$); we have
verified the product against Eq.~\eqref{eq:U_matrix} numerically
to machine precision.  The order-3 logical gate is therefore a
$3\pi/4$ two-mode phase rotation followed by a two-mode squeezer
with $r \approx 0.764$, about $6.6~\mathrm{dB}$---half the
$13.3~\mathrm{dB}$ squeeze of a stabilizer.  Indeed, since
$e^{i(3\pi/4)k_z} = R^3$ and conjugating the squeeze through the
rotation carries its axis from $-3\pi/4$ to $0$, the
decomposition reads equivalently
$U = \gamma_0^{1/2}\,R^3$: half a stabilizer squeeze composed
with three ticks of the octagon rotation.

\subsubsection{Action on the $k=1$ code space}\label{sec:numerics}

We verify the weight-2 representation matrices $\rho(R)$ and
$\rho(U)$ below by two independent methods (``weight 2'' is the
traditional automorphic weight $2k$; the space acted on is the
$k = 1$ code space).  Because the stellar
transform $\psi \mapsto f^\psi$ is antilinear in~$\psi$, code
words transform by the complex-conjugate matrices: the logical
gates quoted in the main text are
$\hat{V}_R = \overline{\rho(R)}$ and
$\hat{V}_U = \overline{\rho(U)}$.  The two methods are:
symbolic verification in exact arithmetic, and numerical
verification against the automorphic forms $\omega_0, \omega_1$.
Full details are in the companion notebook
\texttt{SV-representation-check.ipynb}.

\emph{Claimed matrices.}---The weight-2 representation of
$\Aut(\CS) \cong \GL(2, \mathbb{F}_3)$ in the basis
$(\omega_0, \omega_1)$ is generated by
\begin{equation}\label{eq:rho_extracted}
\begin{aligned}
  \rho(R) &= \frac{1}{\sqrt{6}}
  \begin{pmatrix}1{+}i\sqrt{3} & -i\sqrt{2} \\[-2pt]
    i\sqrt{2} & -(1{-}i\sqrt{3})\end{pmatrix}\!,
  \\[4pt]
  \rho(U) &= \frac{1}{2\sqrt{3}}
  \begin{pmatrix}-(\sqrt{3}{+}i) & 2\sqrt{2} \\[-2pt]
    -2\sqrt{2} & -(\sqrt{3}{-}i)\end{pmatrix}\!.
\end{aligned}
\end{equation}

\emph{Symbolic verification.}---The following are verified
in exact arithmetic (SageMath symbolic ring):
\begin{itemize}
  \item $\rho(R)^4 = -I$ (the hyperelliptic involution acts as $-1$
        on abelian differentials) and $\rho(U)^3 = I$;
  \item $\rho(R)^\dagger\rho(R) = \rho(U)^\dagger\rho(U) = I$
        (unitarity);
  \item the group $\langle\rho(R),\rho(U)\rangle$ has order~48,
        confirming $\GL(2, \mathbb{F}_3)$.
\end{itemize}

\emph{Numerical verification.}---We construct the $k = 1$
automorphic forms $\omega_0(z)$, $\omega_1(z)$ from the $k = 2$
Poincar\'e theta series via the conic factorization:
\begin{enumerate}
  \item Evaluate $\theta_R, \theta_G, \theta_B$ on a $400\times 400$
        grid covering the Poincar\'e disk, retaining ${\sim}\,2000$
        terms per conjugacy class (height cutoff $B = 10^4$), and
        form the Veronese basis $(X, Y, Z)$ via
        Eq.~\eqref{eq:XYZ_SM}.
  \item Extract $\omega_0 = \sqrt{X}$ using a branch-consistent
        square root (priority-queue BFS propagating the sign from
        large-$|X|$ seed points through 8-connected neighbors), and
        set $\omega_1 = Y / \omega_0$.
  \item For each automorphism $\varphi \in \{R, U\}$ and test points
        $\alpha \in \DD$, verify the weight-2 pullback relation
        \begin{equation}\label{eq:pullback_check}
          (\varphi'(\alpha))^k\,\omega_i\bigl(\varphi(\alpha)\bigr)
          = \sum_j \rho(\varphi)_{ij}\,\omega_j(\alpha)
        \end{equation}
        by interpolating $\omega_0, \omega_1$ from the
        branch-consistent grid.
\end{enumerate}
At ${\sim}\,2000$ terms per conjugacy class, the pointwise relative
errors in~\eqref{eq:pullback_check} are $\lesssim 10^{-4}$ for
both generators (Fig.~\ref{fig:convergence}b).

\section{Experimental platform for higher-genus GKP codes}
\label{sec:protocol}

We derive the effective two-mode Hamiltonian implemented by the
cQED architecture of Fig.~\ref{fig:circuit}, then describe the
dissipative kick channel that uses it.  Two resonators
$\hat a,\hat b$ are inductively coupled to a transmon $\hat c$,
and a flux-pumped symmetric SQUID closes the $a$--$b$ loop;
pumping that SQUID at the dressed sum frequency drives a
two-mode squeezing process that is conditional on the
transmon being in $|0\rangle$.  The derivation proceeds in the lab frame, following standard black-box
circuit-quantization logic~\cite{Nigg2012_SM}: solve the linear
circuit, then add the Josephson nonlinearity perturbatively.  A
fully annotated symbolic verification of every step is provided
as the companion notebook
\texttt{SVI-conditional-kick-protocol.ipynb}.

\subsection{Static circuit}
With nodal phases
$\hat\phi_\mu = \varphi_\mu(\hat\mu+\hat\mu^\dagger)$ for
$\mu\in\{a,b,c\}$ and zero-point amplitudes $\varphi_\mu$, and expanding the cosine potential of the symmetric SQUID, we arrive at the bare Hamiltonian
\begin{equation}\label{eq:Hbare_SM}
\begin{aligned}
\hat H &= \sum_\mu\omega_\mu\hat n_\mu
\;+\;\tfrac{L_a}{2}(\hat\phi_a-\hat\phi_c)^2
\;+\;\tfrac{L_b}{2}(\hat\phi_b-\hat\phi_c)^2\\
&\quad\;+\;\tfrac{E_{ab}}{2}\,\epsilon_p(t)\,(\hat\phi_b-\hat\phi_a)^2
\;+\;\tfrac{E_c}{24}\hat\phi_c^4 + \cdots,
\end{aligned}
\end{equation}
with $\epsilon_p(t)=\epsilon\cos(\omega_p t+\theta)$ the SQUID
flux pump.  Each inductive coupler splits as
$\tfrac{L_j}{2}(\hat\phi_j-\hat\phi_c)^2=\tfrac{L_j}{2}\hat\phi_j^2
+\tfrac{L_j}{2}\hat\phi_c^2-L_j\hat\phi_j\hat\phi_c$; absorbing the
self-shifts into renormalized oscillator frequencies and dropping
counter-rotating exchange (valid for
$|g_j|\ll\omega_j+\omega_c$), then normal-ordering the transmon
quartic and applying the RWA, gives the static (pump-off)
Hamiltonian
\begin{equation}\label{eq:Hst_SM}
\hat H_\mathrm{st}
\;=\;\underbrace{\textstyle\sum_\mu\omega_\mu\hat n_\mu+\tfrac{\alpha}{2}\hat n_c(\hat n_c-1)}_{\hat H_0}
\;+\;\underbrace{\textstyle\sum_{j=a,b}g_j(\hat j^\dagger\hat c+\hat c^\dagger\hat j)}_{\hat V},
\end{equation}
with linear exchange $g_j=-L_j\varphi_j\varphi_c$ and transmon
anharmonicity $\alpha=\tfrac{1}{2}E_c\varphi_c^4$ (the sign
convention is $\alpha<0$ for a usual transmon).

\subsection{Dispersive Schrieffer--Wolff}\label{subsec:dispSW}
For an anti-Hermitian
generator $\hat S$, the unitary frame change
$\hat H_\mathrm{eff} = \mathrm{Ad}_{\exp\hat S}\,\hat H_\mathrm{st}$
expands via the identity
$\mathrm{Ad}_{\exp\hat S} = \exp\mathrm{ad}_{\hat S}$ as
\begin{equation}\label{eq:SW_Adexp_SM}
\hat H_\mathrm{eff} = \hat H_\mathrm{st}+\mathrm{ad}_{\hat S}\hat H_\mathrm{st}+\tfrac{1}{2}\mathrm{ad}_{\hat S}^2\hat H_\mathrm{st}+\cdots,
\end{equation}
which, on splitting $\hat H_\mathrm{st}=\hat H_0+\hat V$ and taking
$\hat S$ to be of order $\hat V$, reorganizes into orders of
$\hat V$ as
\begin{equation*}
\hat H_\mathrm{eff} = \hat H_0+\bigl(\mathrm{ad}_{\hat S}\hat H_0+\hat V\bigr)+\bigl(\mathrm{ad}_{\hat S}\hat V+\tfrac{1}{2}\mathrm{ad}_{\hat S}^2\hat H_0\bigr)+O(\hat V^3).
\end{equation*}
The Schrieffer--Wolff choice fixes $\hat S$ by killing the
$O(\hat V)$ correction,
\begin{equation}\label{eq:SW_condition_SM}
\mathrm{ad}_{\hat S}\hat H_0 + \hat V \;=\; 0
\quad\Longleftrightarrow\quad
[\hat S,\hat H_0] \;=\; -\hat V,
\end{equation}
leaving $\hat H_\mathrm{eff} = \hat H_0 + \tfrac{1}{2}[\hat S,\hat V] + O(\hat V^3)$.

\begin{figure}[t]
  \centering
  \includegraphics[width=0.7\columnwidth]{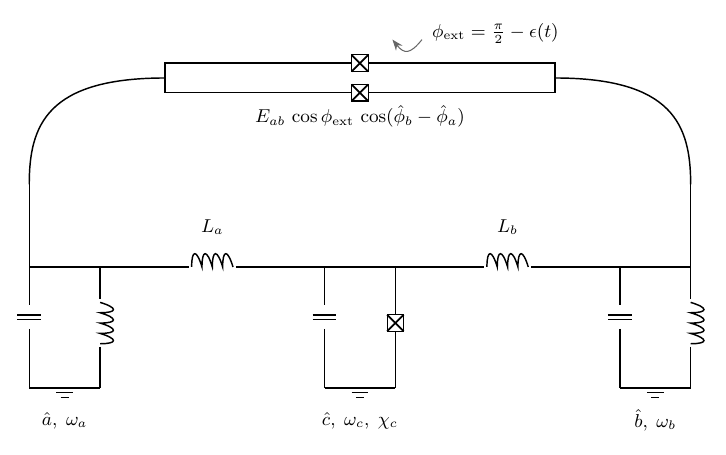}
  \caption{\label{fig:circuit}%
    Two-mode cQED architecture for conditional $\SU(1,1)$ control.
    Two LC resonators ($\hat a, \hat b$) and a transmon ($\hat c$,
    capacitor in parallel with a Josephson junction) are
    inductively coupled in pairs: linear inductors $L_a, L_b$
    realize the static $a$--$c$ and $b$--$c$ exchange (giving rise
    to the dispersive shifts $\chi_a, \chi_b$ in
    Sec.~\ref{subsec:dispSW}), while a
    symmetric SQUID (two parallel JJs) closes the $a$--$b$ loop
    and is flux-pumped at $\omega_p\simeq\widetilde\omega_a+\widetilde\omega_b$
    [precise form in Eq.~\eqref{eq:pump_freq_SM}] via
    $\phi_\mathrm{ext}=\pi/2-\epsilon(t)$.  The pumped SQUID drives
    the parametric two-mode squeezing
    $\Lambda(e^{i\phi}\hat a^\dagger\hat b^\dagger+\mathrm{h.c.})$
    of \eqref{eq:Lambda_SM}.}
\end{figure}

In our setting, define $\Delta_j=\omega_c-\omega_j$.  In the
dispersive regime $|g_j|,|\alpha|\ll|\Delta_j|$, the standard SW
choice is the generator $\hat S^\mathrm{exact}$ satisfying
\eqref{eq:SW_condition_SM} exactly.  Decomposing
$\hat V$ by transmon manifold this is solved by
\begin{equation}\label{eq:Sd_exact_SM}
\hat S^\mathrm{exact}
= \sum_{j=a,b}\Bigl(\hat c^\dagger\hat j\,\frac{g_j}{\Delta_j+\alpha\hat n_c}
\;-\;\mathrm{h.c.}\Bigr),
\end{equation}
manifestly anti-Hermitian, giving $e^{\hat S^\mathrm{exact}}\hat H_\mathrm{st}\,e^{-\hat S^\mathrm{exact}}
=\hat H_0+\tfrac{1}{2}[\hat S^\mathrm{exact},\hat V]+O(\hat V^3)$.
For symbolic tractability the companion notebook works instead with
the \emph{approximate} generator
\begin{equation}\label{eq:Sd_SM}
\hat S = \sum_{j=a,b}\frac{g_j}{\Delta_j}\Bigl[\hat c^\dagger\hat j\,\Bigl(\id-\tfrac{\alpha}{\Delta_j}\hat n_c\Bigr)-\mathrm{h.c.}\Bigr],
\end{equation}
obtained by Taylor-expanding $1/(\Delta_j+\alpha\hat n_c)$ to first
order in $\alpha\hat n_c/\Delta_j$ (a truncation of the
dimensionless generator at relative order $(\alpha/\Delta_j)^2$).
This satisfies the SW condition only approximately: the residual,
computed exactly in the companion notebook, is
\begin{equation}\label{eq:SW_residual_SM}
[\hat S,\hat H_0]+\hat V
= \sum_{j=a,b}\frac{g_j\alpha^2}{\Delta_j^2}
\bigl(\hat n_c^{2}\,\hat j^\dagger\hat c+\mathrm{h.c.}\bigr)
= O\!\bigl(g\,\alpha^{2}/\Delta^{2}\bigr),
\end{equation}
which we neglect below alongside the $O(\hat V^3)$ terms.

Splitting $\hat S=\hat S^{(a)}+\hat S^{(b)}$ and
$\hat V=\hat V_a+\hat V_b$ with
$\hat V_j = g_j(\hat j^\dagger\hat c+\hat c^\dagger\hat j)$ and
$\hat S^{(j)}$ the $j$th summand in \eqref{eq:Sd_SM}, the
second-order correction decomposes into a piece diagonal in the
resonator index,
$\sum_j \tfrac{1}{2}[\hat S^{(j)},\hat V_j]$, and an off-diagonal
piece,
$\tfrac{1}{2}([\hat S^{(a)},\hat V_b]+[\hat S^{(b)},\hat V_a])$.
The two brackets together produce seven operator structures ---
(i)--(v) from the diagonal commutators, and (vi)--(vii) from the off-diagonal commutators:
\begin{itemize}
\item[(i)] Resonator Lamb shift, $-(g_j^2/\Delta_j)\,\hat n_j$;
      absorbed into the dressed frequency
      $\widetilde\omega_j = \omega_j-g_j^2/\Delta_j$.
\item[(ii)] Transmon Lamb shift, $+(g_j^2/\Delta_j)\,\hat n_c$;
      the $j$-sum is absorbed into the dressed frequency
      $\widetilde\omega_c = \omega_c+\sum_j g_j^2/\Delta_j$.
\item[(iii)] Dispersive cross-Kerr, $\chi_j\,\hat n_j\hat n_c$, with
      \begin{equation}\label{eq:chi_jj_SM}
        \chi_j \;=\; \frac{2g_j^2\,\alpha}{\Delta_j^{2}},\qquad j=a,b;
      \end{equation}
      the gate-relevant term.
\item[(iv)] Renormalization of the Kerr coefficient,
      $-(g_j^2\alpha/\Delta_j^2)\,\hat n_c(\hat n_c-1)$; absorbed into a dressed anharmonicity
      $\widetilde\alpha = \alpha-2\sum_j g_j^2\alpha/\Delta_j^2$.
\item[(v)] Two-photon hopping,
      $-(g_j^2\alpha/2\Delta_j^2)(\hat c^{\dagger 2}\hat j^2+\mathrm{h.c.})$.
\item[(vi)] Transmon-mediated $a$--$b$ beam-splitter,
      $\hat J_{ab}(\hat n_c)\,(\hat a^\dagger\hat b+\hat b^\dagger\hat a)$,
      with $\hat n_c$-dependent (operator-valued) coefficient
      \begin{equation}\label{eq:Jab_SM}
      \hat J_{ab}(\hat n_c) = -\frac{g_a g_b}{2}\Bigl(\frac{1}{\Delta_a}+\frac{1}{\Delta_b}\Bigr) + O(\alpha/\Delta).
      \end{equation}
\item[(vii)] Two-mode pair hopping, at the same order as item (v),
      \begin{equation}\label{eq:pair_SM}
      -\frac{g_a g_b\,\alpha}{2}\Bigl(\frac{1}{\Delta_a^{2}}+\frac{1}{\Delta_b^{2}}\Bigr)\bigl(\hat a^\dagger\hat b^\dagger\hat c^{2}+\mathrm{h.c.}\bigr),
      \end{equation}
      converting two transmon excitations into an $a$--$b$ photon
      pair and back.
\end{itemize}
All seven appear explicitly in the symbolic output of the
companion notebook.  The displayed coefficients are expressed to leading-order in $\alpha/\Delta_j$.

Collecting: items (i)--(iv) are diagonal in the occupation numbers
and are absorbed into the dressed parameters
$\widetilde\omega_j,\widetilde\omega_c,\widetilde\alpha,\chi_j$;
items (v)--(vii) are not, and constitute a remnant
$\hat H_\mathrm{off}$:
\begin{equation}\label{eq:Hdisp_SM}
\hat H_\mathrm{eff}^\mathrm{st} = \hat H_\mathrm{disp}+\hat H_\mathrm{off},
\qquad
\hat H_\mathrm{disp} = \sum_{j=a,b}\bigl(\widetilde\omega_j\hat n_j+\chi_j\hat n_j\hat n_c\bigr)+\widetilde\omega_c\hat n_c+\tfrac{\widetilde\alpha}{2}\hat n_c(\hat n_c-1).
\end{equation}
No entry of $\hat H_\mathrm{off}$ may be dropped at this stage. However, later on we will see how they become off-resonant once the SQUID is pumped to produce two-mode squeezing interactions.

\subsection{Pumped SQUID in the dressed frame}\label{subsec:pump}
With the pump on, the SW frame change acts also on the pumped
SQUID term
$\tfrac{E_{ab}}{2}\,\epsilon_p(t)\,(\hat\phi_b-\hat\phi_a)^2$ of
\eqref{eq:Hbare_SM}, adding to the static piece \eqref{eq:Hdisp_SM}
a pumped-branch piece:
\begin{equation}\label{eq:Heff_split_SM}
\hat H_\mathrm{eff} = \hat H_\mathrm{eff}^\mathrm{st}+\hat H_\mathrm{eff}^\mathrm{b}.
\end{equation}
Conjugation
distributes over products,
$e^{\hat S}(\hat\phi_b-\hat\phi_a)^2 e^{-\hat S}=
[e^{\hat S}(\hat\phi_b-\hat\phi_a)e^{-\hat S}]^2$, and the
$\mathrm{Ad}_{\exp\hat S}=\exp\mathrm{ad}_{\hat S}$ expansion of
the inner factor gives
\begin{equation}\label{eq:phidiff_SM}
e^{\hat S}(\hat\phi_b-\hat\phi_a)e^{-\hat S}
=-\varphi_a(\hat a+\hat a^\dagger)+\varphi_b(\hat b+\hat b^\dagger)+O(\varphi\, g/\Delta).
\end{equation}
The $O(\varphi\,g/\Delta)$ remainder of \eqref{eq:phidiff_SM} is
the renormalization of the branch flux. We drop this remainder
altogether; this is a core modeling step.  The dropped contributions
carry an extra factor of $g/\Delta$ relative to the leading piece,
so any operator they generate after squaring and pumping has
coefficient of order $\Lambda\cdot(g/\Delta)$ --- smaller than
$\Lambda$ itself by the dispersive parameter.  Combined with the
condition $|\Lambda|\ll|\chi_a+\chi_b|$, to be imposed in
Sec.~\ref{subsec:rotframe} so that the drive addresses only the
transmon ground state, we obtain the nested chain
\begin{equation}\label{eq:Jchain_SM}
\boxed{\;\;\Lambda\cdot\frac{g}{\Delta}\;\ll\;\Lambda\;\ll\;\chi_a,\,\chi_b,\;\;}
\end{equation}
so every coefficient generated by the dropped remainder lies
parametrically below both the pump amplitude $\Lambda$ and the
dispersive shifts $\chi_j$.  The truncation is consistent at the
level of the bare Hamiltonian without invoking any resonance
argument.

\paragraph*{Pump coefficients.}
Squaring \eqref{eq:phidiff_SM} produces all bilinear pairings of
$\hat a, \hat a^\dagger, \hat b, \hat b^\dagger$, each carrying the
common pump prefactor $\cos(\omega_p t+\theta)$: the two-mode
pairing $\hat a^\dagger\hat b^\dagger+\mathrm{h.c.}$ with c-number
amplitude
\begin{equation}\label{eq:Lambda_SM}
\Lambda = \tfrac{1}{2}E_{ab}\epsilon\,\varphi_a\varphi_b\,e^{-i\theta},
\end{equation}
single-mode pairings $\hat a^{\dagger2},\hat b^{\dagger2}$ with
amplitudes $\lambda_{jj}=\tfrac{1}{4}E_{ab}\epsilon\,\varphi_j^2\,e^{-i\theta}$
of the same order as $\Lambda$, and pump-modulated beam-splitter
and number terms --- together constituting
$\hat H_\mathrm{eff}^\mathrm{b}$.  As with $\hat H_\mathrm{off}$,
none of these may be dropped yet.

\subsection{Rotating frame and the conditional Hamiltonian}\label{subsec:rotframe}
We now set the pump frequency to the dressed sum frequency,
\begin{equation}\label{eq:pump_freq_SM}
\omega_p = \widetilde\omega_a+\widetilde\omega_b+\delta_p,
\end{equation}
with $\delta_p$ a small pump detuning.
This will make the
two-mode squeeze on resonance when the transmon is in
$|0\rangle$ and dispersively detuned by $\chi_a+\chi_b$ in
each higher manifold.  Transform to the rotating
frame
\begin{align}
	R(t) = \exp\bigl[-it\bigl(\nu_a\hat n_a+\nu_b\hat n_b+
\widetilde\omega_c\hat n_c+\tfrac{\widetilde\alpha}{2}\hat n_c(\hat n_c-1)\bigr)\bigr]
\end{align}
with $\nu_j=\widetilde\omega_j+\delta_p/2$ (so
$\nu_a+\nu_b=\omega_p$ exactly).  In this frame the diagonal part
of \eqref{eq:Hdisp_SM} collapses identically: the transmon
self-energy
$\widetilde\omega_c\hat n_c+\tfrac{\widetilde\alpha}{2}\hat n_c(\hat n_c-1)$
is the frame generator and cancels, the cross-Kerr terms are
invariant, and
$\widetilde\omega_j\hat n_j\to-\tfrac{\delta_p}{2}\hat n_j$; the
resonant pump pairing
$\Lambda\,\hat a^\dagger\hat b^\dagger e^{-i\omega_p t}+\mathrm{h.c.}$
becomes static.  Every other term catalogued in $\hat H_\mathrm{off}$
and $\hat H_\mathrm{eff}^\mathrm{b}$ acquires a residual phase
$e^{i\Omega t}$.  We quote each detuning at $\delta_p=0$ (where
$\nu_j=\widetilde\omega_j$) and at leading order: for transmon
occupations $m=O(1)$ in the dispersive regime, the anharmonic,
Lamb-shift, and cross-Kerr corrections to each $\Omega$ ---
$O(\widetilde\alpha\,m)$, $O(g^2/\Delta)$, and $O(\chi_j m)$ ---
are parametrically small, and the detunings reduce to bare
frequency combinations:
\begin{itemize}
\item two-photon hopping (v),
      $\sim(g_j^{2}\alpha/\Delta_j^{2})\,\hat c^{\dagger2}\hat j^{2}$:
      \quad$\Omega=2\Delta_j$;
\item mediated beam-splitter (vi),
      $\sim(g_ag_b/\Delta)\,\hat a^{\dagger}\hat b$:
      \quad$\Omega=\Delta_b-\Delta_a$;
\item pair hopping (vii),
      $\sim(g_ag_b\alpha/\Delta^{2})\,\hat a^{\dagger}\hat b^{\dagger}\hat c^{2}$:
      \quad$\Omega=-(\Delta_a+\Delta_b)$;
\item single-mode squeezes, $\sim\Lambda\,\hat a^{\dagger2}$ and
      $\sim\Lambda\,\hat b^{\dagger2}$:
      \quad$\Omega=\pm(\Delta_b-\Delta_a)$;
\item pump-modulated beam-splitter,
      $\sim\Lambda\,\hat a^{\dagger}\hat b\,e^{\mp i\omega_p t}$:
      \quad$\Omega=-2\omega_b$ or $+2\omega_a$.
\end{itemize}
The RWA drops each term whose coefficient is small against its
detuning, at the price of second-order shifts
$\sim|\mathrm{coeff}|^2/|\Omega|$.  Reading the conditions
$|\mathrm{coeff}|\ll|\Omega|$ off the list: for (v) and (vii) they
are $g_j^{2}|\alpha|/\Delta_j^{2}\ll2|\Delta_j|$ and
$g_ag_b|\alpha|/\Delta^{2}\ll|\Delta_a+\Delta_b|$,
automatic in the dispersive regime; for the pump-modulated
beam-splitter they are
$|\Lambda|\ll2\omega_{a},2\omega_b$, implied a
fortiori by \eqref{eq:Jchain_SM} since
$\chi_j\ll\omega_j$ in this
dispersive architecture.  Finally, (vi) and the single-mode
squeezes require the mode--mode detuning condition
\begin{equation}\label{eq:BS_offres_SM}
\frac{g_ag_b}{2}\Bigl|\frac{1}{\Delta_a}+\frac{1}{\Delta_b}\Bigr|,
\;|\lambda_{aa}|,\;|\lambda_{bb}|
\;\ll\;|\Delta_a-\Delta_b|.
\end{equation}
After the RWA the boxed effective Hamiltonian is
\begin{equation}\label{eq:HRWA_SM}
\hat H_\mathrm{RWA} = \bigl[\chi_a\hat n_c-\tfrac{\delta_p}{2}\bigr]\hat n_a +\bigl[\chi_b\hat n_c-\tfrac{\delta_p}{2}\bigr]\hat n_b+\Lambda\,\hat a^\dagger\hat b^\dagger+\Lambda^*\,\hat a\hat b+\hat H_\mathrm{small},
\end{equation}
with $\hat H_\mathrm{small}$ the collection of RWA-dropped terms
enumerated above.
The two-mode squeezing interaction is on resonance only when $\hat n_c=0$.

\subsection{Summary of regime of validity}
What the detailed analysis above tells us is that, in addition to
the standard dispersive-regime constraints
$|g_j|,|\alpha|\ll|\Delta_j|$ (which by themselves suppress the
two-photon hopping (v), the pair hopping (vii), and the branch-flux
sidebands) and the chain \eqref{eq:Jchain_SM}, the Hamiltonian
\eqref{eq:HRWA_SM} requires the mode--mode detuning condition
\eqref{eq:BS_offres_SM}: the bare resonant frequencies of the
resonator modes $\hat a,\hat b$ must be detuned from each other by
much more than $g_ag_b/\Delta$ and $|\lambda_{jj}|$.  In the
dispersive limit $g_j/\Delta_j \ll 1$ this is not a stringent
constraint.

\subsection{Conditional gate primitives}\label{sec:kick}
The conditional Hamiltonian of the main text requires only matched
dispersive shifts, $\chi_a=\chi_b\equiv\chi$.  With $\chi_a=\chi_b=\chi$ and $\delta_p=0$, the
boxed Hamiltonian \eqref{eq:HRWA_SM} reduces to
\begin{equation}\label{eq:Hcond_SM}
\hat H_\mathrm{eff}\;\simeq\;\chi\,\hat n_c\,(\hat n_a+\hat n_b)+2|\Lambda|\,\Kphi,
\end{equation}
with $\phi=\arg\Lambda=-\theta$ and
$\Kphi = \tfrac{1}{2}(e^{i\phi}\hat a^\dagger\hat b^\dagger+e^{-i\phi}\hat a\hat b)$,
so that
$2|\Lambda|\Kphi=\Lambda\,\hat a^\dagger\hat b^\dagger+\Lambda^*\hat a\hat b$.
Like \eqref{eq:HRWA_SM}, this is written in the full rotating
frame: the transmon self-energy is part of the frame generator and,
if reintroduced, contributes only a passive transmon phase.
Two primitives follow:
\begin{itemize}
  \item \emph{Conditional squeeze} ($|\chi|\gg|\Lambda|$): with
        the pump on, the resonant manifold $\hat n_c=0$
        undergoes $\exp(-i\theta\,P_0\Kphi)$, where
        $P_0 = |0\rangle\langle 0|_c$ projects onto the transmon
        ground state.  On the ancilla qubit subspace
        $P_0 = \id - \hat n_c$, so this is equivalent
        (up to the unconditional squeeze $e^{-i\theta\Kphi}$) to
        $\mathrm{C}\Kphi(\theta) = e^{i\hat n_c\theta\,\Kphi}$.
  \item \emph{Conditional} $\Kz$ ($\Lambda=0$): with
        $\Kz = \tfrac{1}{2}(\hat n_a+\hat n_b+1)$, the cross-Kerr
        $\chi\hat n_c(\hat n_a+\hat n_b)
        = 2\chi\hat n_c\Kz - \chi\hat n_c$ realizes
        $\mathrm{C}\Kz(\theta) = e^{i\hat n_c\theta\,\Kz}$ up to a
        transmon phase.
\end{itemize}
Generic conditional $\SU(1,1)$ operations are synthesized by
Trotterization of these primitives; preparing the transmon in
$|1\rangle$ converts conditional to unconditional.

\subsection{Conditional-kick protocol}
For a single stabilizer generator
$\hat{\mathcal{D}}_\gamma$, the conditional-kick channel applies
the Kraus operators
\begin{equation}\label{eq:kraus_SM}
  \hat{\mathcal{K}}_0 = \tfrac{1}{2}(\id + \hat{\mathcal{D}}_\gamma), \qquad
  \hat{\mathcal{K}}_1 = \tfrac{1}{2}\Veps(\id - \hat{\mathcal{D}}_\gamma),
\end{equation}
where $\Veps = e^{i\varepsilon\Kz}$ is a small
$\SU(1,1)$ ``kick.''  These satisfy
$\hat{\mathcal{K}}_0^\dagger \hat{\mathcal{K}}_0
+ \hat{\mathcal{K}}_1^\dagger \hat{\mathcal{K}}_1 = \id$
identically: $\hat{\mathcal{D}}_\gamma$ and $\Veps$ are unitary, so
the cross terms cancel exactly.

The circuit implementation is:
\begin{enumerate}
  \item Prepare ancilla in $|+\rangle$.
  \item Apply $\mathrm{C}\hat{\mathcal{D}}_\gamma$ (conditional
        stabilizer displacement).
  \item Measure ancilla in the $X$ basis.
  \item On outcome $|{-}\rangle$, apply correction
        $\Veps$.
\end{enumerate}

\subsection{Stabilizer-defect dynamics}
Because the generator $\hat{K}_{\alpha_\gamma}$ is hyperbolic, its
spectrum is purely continuous and exact eigenstates
$\Dg|\psi\rangle = |\psi\rangle$ are non-normalizable; the
normalizable-state formulation of the stabilizer constraints is
treated at full rigor in the spectral-gap analysis of
Sec.~\ref{sec:gap}.  For the protocol it suffices to track the
stabilizer defect
\begin{equation}\label{eq:defect_SM}
  \mathcal{W}(\rho) \;=\; 1 - \Re\Tr[\rho\,\Dg]\;\in\;[0,2],
\end{equation}
which is exactly twice the probability of the error outcome in one
round:
\begin{equation}\label{eq:pminus_SM}
  P_-(\rho)
  \;=\;\Tr[\hat{\mathcal{K}}_1^\dagger\hat{\mathcal{K}}_1\,\rho]
  \;=\;\tfrac{1}{4}\Tr[(\id-\Dg^\dagger)(\id-\Dg)\,\rho]
  \;=\;\tfrac{1}{2}\,\mathcal{W}(\rho).
\end{equation}
Small-defect states are therefore asymptotically dark: the error
branch fires with probability $\mathcal{W}/2$, vanishing as the
state approaches the code constraint. Convergence of the resulting dynamics is
established numerically: for a single generator the protocol
drives $\mathcal{W}$ rapidly to a small residual floor from
arbitrary initial states, as is seen in the main text.

\subsection{Multi-generator protocol}
In practice, one cycles through
the stabilizer generators
$\hat{\mathcal{D}}_0, \ldots, \hat{\mathcal{D}}_3$, applying
one round of the conditional kick channel for each.  As shown in
the main text's protocol figure, the single-generator protocol
drives its stabilizer defect to a small $\varepsilon$-limited
floor, while the multi-generator protocol plateaus at
$\langle \Hstab \rangle = \mathcal{O}(1)$ rather
than converging toward zero, providing independent numerical
evidence for the spectral gap
(Theorem~\ref{thm:gap_SM}).

\section{Proof of the spectral gap}\label{sec:gap}

We give the full proof of the spectral-gap theorem of the main
text.  The argument has three steps, following the proof sketch
given there:
(1)~the stabilizer group $\Gamma \cong \pi_1(\CS)$ is
non-amenable for $g \geq 2$, so the Markov operator on the regular
representation satisfies $\|M_\lambda\| < 1$;
(2)~the discrete series restricts to a representation weakly
contained in the regular representation,
$D_k^+\big|_\Gamma \preceq \lambda_\Gamma$; and
(3)~weak containment implies
$\|M\| \leq \|M_\lambda\|$, giving the gap $\delta > 0$.
We make each step precise.  Combined with the existence of
approximate GKP code words of arbitrary precision at genus one,
the result establishes the amenability dichotomy of the main
text.

\subsection{Step 1: Non-amenability and the Kesten bound}

Recall that $\Gamma$ is \emph{amenable} if and only if the trivial
representation is weakly contained in the regular representation:
$\mathbbm{1}_\Gamma \preceq \lambda_\Gamma$
(Hulanicki--Reiter~\cite{Hulanicki1966_SM,BHV2008_SM}).
Equivalently, $\Gamma$ is amenable if and only if the Markov
operator $M_\lambda$ of any symmetric random walk on the Cayley
graph $\mathrm{Cay}(\Gamma, S)$ has
$\|M_\lambda\| = 1$~\cite[Cor.~G.4.6]{BHV2008_SM}.

\begin{lemma}\label{lem:nonamenable}
For $g \geq 2$, the surface group
$\Gamma \cong \pi_1(\CS)$ is non-amenable, and consequently
$\|M_\lambda\| < 1$.
\end{lemma}

\begin{proof}
The standard presentation
$\Gamma = \langle a_1, b_1, \ldots, a_g, b_g \mid
\prod_{i=1}^g [a_i, b_i] = 1\rangle$ admits a surjection
$\Gamma \twoheadrightarrow F_2$ onto the free group on two
generators (send $a_1 \mapsto a$, $a_2 \mapsto b$, and every $b_i$ and
remaining $a_i \mapsto e$; then every commutator in the defining relation
maps to~$e$)~\cite[Ex.~G.2.4(iii)]{BHV2008_SM}.
Since $F_2$ is non-amenable~\cite[Ex.~G.2.4(ii)]{BHV2008_SM}
and quotients of amenable groups are
amenable~\cite[Prop.~G.2.2]{BHV2008_SM}, $\Gamma$ cannot be
amenable.
\end{proof}

\subsection{Step 2: Weak containment of the restricted
  discrete series}

We write $\pi \preceq \sigma$ when a unitary representation
$\pi$ is \emph{weakly contained}
in~$\sigma$~\cite[Def.~F.1.1]{BHV2008_SM}: every matrix
coefficient of~$\pi$ can be uniformly approximated on compact
sets by convex combinations of matrix coefficients of~$\sigma$.

\begin{lemma}\label{lem:restriction}
Let $\Gamma \subset G = \PSU(1,1)$ be a discrete Fuchsian
subgroup and $D_k^+$ the holomorphic discrete series at integer
Bargmann index~$k$.  Then
\begin{equation}\label{eq:restriction}
  D_k^+\big|_\Gamma \;\preceq\; \lambda_\Gamma.
\end{equation}
\end{lemma}

\begin{proof}
By Lang's Plancherel theorem for
$\SU(1,1)$~\cite[Ch.~IX, Thm.~1]{Lang1985_SM}, $D_k^+$ realizes
as the closed, left-translation-invariant, irreducible subspace
of $L^2(\SU(1,1))$ whose lowest-weight vector is
$\varphi_m(g) = \alpha(g)^{-m}$ and which is spanned by the
eigenfunctions
$\varphi_{m+2r}(g) = \alpha(g)^{-m-r}\,\bar\beta(g)^{r}$ of weight
$m+2r$, $r = 0,1,2,\ldots$, where $m = 2k$ and
$g = \bigl(\begin{smallmatrix}\alpha&\beta\\\bar\beta&\bar\alpha
\end{smallmatrix}\bigr) \in \SU(1,1)$.  At integer~$k$ the weight
$m = 2k$ is even, so $\varphi_{m+2r}(-g) = (-1)^{m}\varphi_{m+2r}(g)
= \varphi_{m+2r}(g)$ for every~$r$, and hence the entire $D_k^+$
summand is invariant under $g \mapsto -g$.  The natural isometry
\[
  L^2(\PSU(1,1)) \;\xrightarrow{\;\sim\;}\;
  L^2(\SU(1,1))^{\mathbb{Z}_2}
  \;=\; \{ f \in L^2(\SU(1,1)) : f(-g) = f(g)\}
\]
therefore identifies $D_k^+$ with a subrepresentation of
$\lambda_G = \lambda_{\PSU(1,1)}$, so $D_k^+ \leq \lambda_G$ and
in particular $D_k^+ \preceq \lambda_G$.

Since $\Gamma$ is discrete and hence a closed subgroup of the
Hausdorff group~$G$, two standard results give the chain
\[
  D_k^+\big|_\Gamma \;\preceq\;
  \lambda_G\big|_\Gamma \;\preceq\; \lambda_\Gamma\,:
\]
the first step is \cite[Prop.~F.3.4]{BHV2008_SM} (restriction
preserves weak containment for closed subgroups); the second is
\cite[Prop.~F.1.10]{BHV2008_SM}
($\lambda_G|_H \preceq \lambda_H$ for closed $H \leq G$).
\end{proof}

\subsection{Step 3: From weak containment to the spectral gap}

The stabilizer Hamiltonian and the Cayley graph Laplacian are
evaluations of the \emph{same} group algebra element
$a_S \in \mathbb{C}[\Gamma]$ in different representations.
Define the \emph{Markov (averaging) operator}
\begin{equation}\label{eq:markov}
  T_S \;=\; \frac{1}{|\widetilde{S}|}
  \sum_{s \in \widetilde{S}} s
  \;\;\in\; \mathbb{C}[\Gamma],
  \qquad \widetilde{S} = S \cup S^{-1},
\end{equation}
so that $\Hstab = \id - \pi(T_S)$ for any unitary
representation~$\pi$ of~$\Gamma$ [the stabilizer Hamiltonian of
the main text].  Since $\widetilde{S}$ is symmetric,
$T_S = T_S^*$ and $\pi(T_S)$ is
self-adjoint, so the spectral gap is
\begin{equation}\label{eq:gap_from_norm}
  \delta_\pi
  \;=\; \inf\,\sigma(\Hstab)
  \;=\; 1 - \sup\sigma\bigl(\pi(T_S)\bigr)
  \;\geq\; 1 - \|\pi(T_S)\|.
\end{equation}

The key tool is the operator norm inequality for weak
containment~\cite[Thm.~F.4.4]{BHV2008_SM}: if
$\pi \preceq \sigma$, then
\begin{equation}\label{eq:norm_ineq}
  \|\pi(a)\| \;\leq\; \|\sigma(a)\|
  \qquad \text{for all } a \in \mathbb{C}[\Gamma].
\end{equation}
In the regular representation $\lambda_\Gamma$, the norm
$\|\lambda_\Gamma(T_S)\|$ is the \emph{Kesten spectral radius}
$\rho(S)$, which equals~$1$ if and only if $\Gamma$ is
amenable~\cite[Cor.~G.4.6]{BHV2008_SM}.

\begin{theorem}[Spectral gap]\label{thm:gap_SM}
Let $\Gamma$ be a non-amenable discrete group with finite
symmetric generating set~$S$, acting unitarily on a Hilbert
space~$\HH$ via a representation~$\pi \preceq \lambda_\Gamma$.
Then the stabilizer
Hamiltonian $\Hstab = \id - \pi(T_S)$ has a spectral gap
\begin{equation}\label{eq:gap_bound}
  \delta_\pi
  \;=\; 1 - \sup\sigma\bigl(\pi(T_S)\bigr)
  \;\geq\; 1 - \|\pi(T_S)\|
  \;\geq\; 1 - \rho(S)
  \;>\; 0.
\end{equation}
\end{theorem}

\begin{proof}
Apply~\eqref{eq:norm_ineq} to $a = T_S$ and
$\sigma = \lambda_\Gamma$:
\[
  \|\pi(T_S)\| \;\leq\; \|\lambda_\Gamma(T_S)\|
  \;=\; \rho(S).
\]
By~\eqref{eq:gap_from_norm},
$\delta_\pi \geq 1 - \rho(S)$.  Since $\Gamma$ is
non-amenable, $\rho(S) < 1$
(Lemma~\ref{lem:nonamenable}), so $\delta_\pi > 0$.
\end{proof}

\noindent The spectral-gap theorem of the main text follows by
taking $\pi = D_k^+$ and applying Lemma~\ref{lem:nonamenable}
(non-amenability for $g \geq 2$) and
Lemma~\ref{lem:restriction}
($D_k^+\big|_\Gamma \preceq \lambda_\Gamma$).

For the genus-two surface group with standard generators,
Nagnibeda's upper bound~\cite{Nagnibeda1997_SM}
$\rho(S) \leq 0.662816$ gives
$\delta \geq 1 - \rho(S) \geq 0.337$; Bartholdi's lower
bound~\cite{Bartholdi2002_SM} $\rho(S) \geq 0.662420$ shows that
this estimate of $1 - \rho(S)$ is sharp to four digits.  We verify
the gap numerically by diagonalizing $\Hstab$ for the Bolza
surface at $k = 2$ in a truncated Fock space of dimension~$M$
(Table~\ref{tab:E0}).  The lowest eigenvalue $E_0(M)$
converges from above to $E_0 = 0.472$ at $M = 2000$; Richardson
extrapolation in $1/M$ on the last two rows,
$2E_0(2000) - E_0(1000) = 0.454$, gives $\delta \approx 0.45$, of
which the bound $\delta \geq 0.337$ captures ${\sim}75\%$.

\begin{table}[h]
  \caption{\label{tab:E0}Lowest eigenvalue $E_0$ of $\Hstab$
    for the Bolza surface at $k = 2$, computed by diagonalizing
    the $M \times M$ upper-left block of $\Hstab$ constructed at
    a larger truncation $5M$ (to suppress boundary artifacts).
    $E_0(M)$ converges from above as $M$ increases.}
\begin{ruledtabular}
\begin{tabular}{rccc}
  $M$ & $E_0$ &
  $\langle\mathrm{Re}\,D\rangle$ &
  $|\Delta E_0|/(\Delta M/100)$ \\
  \hline
  100  & 0.579 & 0.421 & --- \\
  200  & 0.562 & 0.438 & 0.017 \\
  500  & 0.510 & 0.490 & 0.017 \\
  1000 & 0.490 & 0.510 & 0.004 \\
  2000 & 0.472 & 0.528 & 0.002 \\
\end{tabular}
\end{ruledtabular}
\end{table}

\subsection{Extension to tessellation codes}
\label{sec:tessellation_extension}

The spectral gap obstruction is not specific to bosonic codes: the same machinery applies to any
unitary representation~$\pi$ of~$\Gamma$ that is weakly
contained in~$\lambda_\Gamma$.  As an illustration, the
tessellation codes of Wang, Xu, and Liu~\cite{WXL2025}---which
encode qudits into~$L^2(\DD)$ via delta-function
superpositions on a hyperbolic tessellation---fall under
exactly the same framework.  Their stabilizer Hamiltonian is
$\Hstab =\id- \pi(T_S)$ for $\pi$ the quasi-regular representation
of $\Gamma \subset \PSU(1,1)$ on~$L^2(\DD)$, and
Theorem~\ref{thm:gap_SM} applies once the analogue of
Lemma~\ref{lem:restriction} is established:

\begin{lemma}\label{lem:restriction_tessellation}
Let $\Gamma \subset G = \PSU(1,1)$ be a discrete Fuchsian
subgroup and $\pi$ the quasi-regular representation of~$G$
on~$L^2(\DD)$, $[\pi(g)f](x) = f(g^{-1}\cdot x)$.
Then
\[
  \pi\big|_\Gamma \;\preceq\; \lambda_\Gamma.
\]
\end{lemma}

\begin{proof}
Writing $\DD = \SU(1,1)/\mathrm{U}(1)$ as the quotient by the
maximal compact subgroup, the natural map
$L^2(\DD) \xrightarrow{\sim} L^2(\SU(1,1))^{\mathrm{U}(1)}$
(right $\mathrm{U}(1)$-invariant functions) realizes $\pi$ as a
subrepresentation of $\lambda_{\SU(1,1)}$.  Since
$\{\pm I\} \subset \mathrm{U}(1)$, every right-$\mathrm{U}(1)$-invariant
function is a fortiori $\{\pm I\}$-invariant, so the image lies
in $L^2(\SU(1,1))^{\mathbb{Z}_2} \cong L^2(\PSU(1,1)) = \lambda_G$
via the central isometry used in the proof of
Lemma~\ref{lem:restriction}.  Hence $\pi \leq \lambda_G$, and in
particular $\pi \preceq \lambda_G$.  Restriction to the closed
subgroup $\Gamma$ then gives the chain
\[
  \pi\big|_\Gamma \;\preceq\;
  \lambda_G\big|_\Gamma \;\preceq\; \lambda_\Gamma\,,
\]
by~\cite[Prop.~F.3.4]{BHV2008_SM}
and~\cite[Prop.~F.1.10]{BHV2008_SM} respectively.
\end{proof}

\section{Beyond integer weight: spin structures and a
non-split gate group}
\label{fw:doublecover}\label{fw:bolzahalf}

At half-integral Bargmann index, the logical extension need not split. Two distinct obstructions occur on the Bolza surface. For the odd theta characteristics, the obstruction is purely a double-cover obstruction: no choice of signs realizes the classical group inside \(SU(1,1)\), although arbitrary scalar rephasing linearizes the resulting projective action. For certain even theta characteristics, the obstruction is stronger. One is able to find a commuting pair of classical logic gates which lift to an anti-commuting pair of quantum gates, an invariant unchanged by arbitrary \(U(1)\)-rephasing. In this case the action of the classical logical gate group on the code space is irreducibly projective.\\

 Throughout this section the phase space is hyperbolic and the weight $k$
is half-integral: $G=\PSU(1,1)$,
$k\in\tfrac12\mathbb Z\setminus\mathbb Z$, and the generalized
displacement group is the double cover $\cD_k\cong\SU(1,1)$
(Sec.~\ref{fw:hyperbolic}), so the ambient extension
\eqref{fw:eq:pauli-ext} is the concrete exact sequence
\begin{equation}\label{fw:eq:su-psu}
  1 \longrightarrow \{\pm I\} \longrightarrow \SU(1,1)
    \xrightarrow{\ \pi\ } \PSU(1,1) \longrightarrow 1 ,
\end{equation}
with $Z_k=\{\pm I\}\cong\mathbb Z_2$ acting on states by
$\hat{\mathcal D}_{-I}=e^{2\pi ik}\,\id=-\id$.  Two consequences of
\eqref{fw:eq:su-psu} are immediate and used repeatedly below: every
classical map $g$ has exactly two lifts $\pm\tilde g$, so any two
lifts of a group $\Gamma\subset\PSU(1,1)$ differ by a sign character
$\Gamma\to\{\pm1\}$; and the representation distinguishes the two
lifts, since $-I$ acts as $-\id\neq\id$. \\

\subsection{Stabilizer lifts are theta characteristics}\label{fw:sub:lifts}
At half-integral $k$ the stabilizer selects a square root of the
canonical bundle: $\cS$ corresponds to a line bundle $\mathcal E\to\CS$
with $\mathcal E^{\otimes2}\simeq\Omega_{\CS}$, a \emph{theta
characteristic} (spin structure), and, writing $k=m+\tfrac12$,
$\Ls\simeq\Omega_{\CS}^{\otimes m}\otimes\mathcal E$
(Sec.~\ref{fw:hyperbolic}).  Every automorphism preserves
$\Omega_{\CS}^{\otimes m}$, so preserving the frame bundle is
equivalent to preserving the spin structure, and the $\cS$-polarized
normalizer of Sec.~\ref{fw:logical} becomes
\begin{equation}\label{fw:eq:half-logical}
  G_{\cS}
  \;=\; N^{\cS}_{G}(\Gamma)/\Gamma
  \;=\; \bigl\{\bar g\in\Aut(\CS):
      \bar g^{*}\mathcal E\simeq\mathcal E\bigr\},
\end{equation}
independent of the level---the automorphisms fixing the chosen spin
structure.  (In the hyperbolic setting every automorphism of $\CS$ is
induced by an element of $G=\PSU(1,1)$ normalizing $\Gamma$, which is
why the quotient $N^{\cS}_{G}(\Gamma)/\Gamma$ is the full
stabilizer of $\mathcal E$ in $\Aut(\CS)$.)\\

For the Bolza curve
$v^2=t^5-t$ the Weierstrass points---the fixed points of the
hyperelliptic involution $\iota$, which for the Bolza octagon is the
rotation $\iota(z)=-z$ about the center---sit at the vertices of a
regular octahedron on the hyperelliptic $\PP^1$, and
$\Aut(\CS)/\langle\iota\rangle\cong
S_4$ acts on them as the rotation group of that octahedron.

\subsection{Odd theta characteristics are Weierstrass
points}\label{fw:sub:oddchars}
On a surface of
genus $g$ there are $2^{2g}$ theta characteristics
[Ref.~\cite{Atiyah71_fw}, Thm.~2]; the \emph{parity}
$\varphi(\mathcal E):=\dim\Gamma(C,\mathcal E)\bmod 2$ is a
deformation invariant [\emph{ibid.}]; and exactly
$2^{g-1}(2^{g}+1)$ of them are even [Ref.~\cite{Atiyah71_fw},
Thm.~3].  At genus two: sixteen characteristics, ten even and six
odd.  The six odd ones admit a completely explicit description.

\begin{lemma}\label{fw:lem:degree-one}
On a compact Riemann surface $C$ of positive genus,
$\dim\Gamma(C,\mathcal O(p))=1$ for every point $p\in C$, and
$\mathcal O(p)\simeq\mathcal O(q)$ iff $p=q$.
\end{lemma}

\begin{proof}
The constant function $1$ satisfies $\mathrm{div}(1)+p=p\geq0$, so
$1\in\Gamma(C,\mathcal O(p))$ and the dimension is at least one.
Conversely, a nonconstant $f\in\Gamma(C,\mathcal O(p))$ would have a
single simple pole, hence define a holomorphic map
$f\colon C\to\PP^{1}$ of degree
$\deg f=\sum_{x\in f^{-1}(\infty)}m_{f}(x)=1$; a degree-one
holomorphic map of compact Riemann surfaces is an isomorphism,
forcing $C\simeq\PP^{1}$---genus zero.  The same argument proves the
second claim: $\mathcal O(p)\simeq\mathcal O(q)$ with $p\neq q$
supplies a meromorphic $f$ with $\mathrm{div}f=q-p$, again a
degree-one map.
\end{proof}

\begin{theorem}\label{fw:thm:odd-chars}
On a curve of genus two, $\mathcal O(w)$ is an odd theta
characteristic for every Weierstrass point $w$, and the six bundles
$\mathcal O(w_{1}),\dots,\mathcal O(w_{6})$ exhaust the odd theta
characteristics.
\end{theorem}

\begin{proof}
We first show $\mathcal O(2w)\simeq K_{C}$.  At genus two the
canonical map
$\varphi_{C}=[\omega_{0}:\omega_{1}]\colon C\to\PP^{1}$, with
$\omega_{0},\omega_{1}$ a basis of holomorphic differentials,
\emph{is} the hyperelliptic double cover, and the Weierstrass points
are its ramification points.  Let $\lambda=\varphi_{C}(w)$.  The holomorphic
section $\omega_{1}-\lambda\,\omega_{0}$ of $K_{C}$ vanishes exactly
on the fiber $\varphi_{C}^{-1}(\lambda)$; since $w$ is a ramification
point, that fiber is the single point $w$, so
$\mathrm{div}(\omega_{1}-\lambda\omega_{0})$ is supported at $w$, and
$\deg K_{C}=2$ fixes the multiplicity:
\begin{equation}\label{fw:eq:canonical-2w}
  \mathrm{div}\bigl(\omega_{1}-\lambda\,\omega_{0}\bigr)\;=\;2w ,
\end{equation}
a canonical divisor.  Hence $K_{C}\simeq\mathcal O(2w)$: each
$\mathcal O(w)$ is a theta characteristic, and its parity is odd by
Lemma~\ref{fw:lem:degree-one}, $\dim\Gamma(C,\mathcal O(w))=1$.  The
six characteristics are pairwise distinct, again by
Lemma~\ref{fw:lem:degree-one}, and a genus-two surface has exactly
$2^{g-1}(2^{g}-1)=6$ odd theta characteristics---so the list is
complete.
\end{proof}

\subsection{Odd characteristics carry a non-split $\mathbb Z_{16}$
gate group}\label{fw:sub:z16}

Take the odd characteristic at the origin.  The center of the octagon,
$z=0$, is fixed by the hyperelliptic involution $\iota(z)=-z$, hence its
image $w_0\in\CS$ is a Weierstrass point; set
$\mathcal E=\mathcal O(w_0)$ and let $\cS$ be the corresponding lift.\\

\emph{Classical logical group.}---By \eqref{fw:eq:half-logical}, the
classical gate group is the stabilizer of $\mathcal E$ in
$\Aut(\CS)$.  For $\mathcal E=\mathcal O(w_0)$, pullback acts on the
underlying divisor by
$\bar g^{*}\mathcal O(w_0)\simeq\mathcal O(\bar g^{-1}w_0)$, and
$\mathcal O(\bar g^{-1}w_0)\simeq\mathcal O(w_0)$ iff
$\bar g^{-1}w_0=w_0$ (Lemma~\ref{fw:lem:degree-one}); so $\bar g$
preserves $\mathcal E$ iff it fixes the point $w_0$.  The stabilizer is computed by the orbit--stabilizer theorem: the
octahedral action of Sec.~\ref{fw:sub:lifts} is transitive on the six
Weierstrass points, so the stabilizer of $w_0$ in the order-48 group
$\Aut(\CS)$ has order $48/6=8$; and it visibly contains the order-8
rotation $R(z)=e^{i\pi/4}z$ [Eq.~\eqref{eq:R_matrix}], which fixes
$z=0$.
Hence
\begin{equation}\label{fw:eq:z8}
  G_{\cS} = \langle R\rangle \cong \mathbb Z_8 .
\end{equation}

\emph{Quantum logical group.}---The rotation $R$ lifts to
$\SU(1,1)$ as $\widetilde R=\mathrm{diag}(e^{i\pi/8},e^{-i\pi/8})$
(the other lift being $-\widetilde R$), and
$\widetilde R^{\,8}=-I$, so $\widetilde R$ has order sixteen.  At
half-integral $k$ the representation
$\tilde g\mapsto\hat{\mathcal D}_{\tilde g}$ is faithful
($\hat{\mathcal D}_{-I}=-\id\neq\id$, Sec.~\ref{fw:hyperbolic}), so
$\hat{\mathcal D}_{\widetilde R}$ has order sixteen as well;
and since $-\id\notin\cS$,
no power of $\hat{\mathcal D}_{\widetilde R}$ below the sixteenth is
trivial modulo
the stabilizer.  The extension~\eqref{fw:eq:main-ses} over
$G_{\cS}\cong\mathbb Z_8$ is therefore
\begin{equation}\label{fw:eq:z16}
  1\longrightarrow \mathbb Z_2 \longrightarrow
  \cD_{\cS}=\bigl\langle \hat{\mathcal D}_{\widetilde R}\bigr\rangle
  \cong \mathbb Z_{16}
  \longrightarrow \mathbb Z_8 \longrightarrow 1 .
\end{equation}

\begin{theorem}[Double-cover obstruction]\label{fw:thm:nonsplit}
At half-integral Bargmann index the logical
extension~\eqref{fw:eq:z16} of the Bolza code
$\mathcal E=\mathcal O(w_0)$ does not split: the logical
gate group
$\cD_{\cS}\cong\mathbb Z_{16}$ is cyclic, while
$\mathbb Z_2\times\mathbb Z_8$ has no element of order sixteen.  No
choice of signs therefore makes the lifted gates close into a copy
of the classical gate group $G_{\cS}\cong\mathbb Z_8$ inside
$\cD_{\cS}$.
\end{theorem}

\noindent The five remaining odd characteristics are the
$\Aut$-translates of $\mathcal E=\mathcal O(w_0)$
(Sec.~\ref{fw:sub:lifts}), with conjugate quantum gate groups
$\cong\mathbb Z_{16}$: the computation above covers all six odd
spin structures.\\

\begin{remark} Although the odd-characteristic extension above does not split as an
extension by $Z=\{\pm1\}$, it splits after adjoining arbitrary scalar phases.
Let $d_R\in\mathcal D_{\mathcal S}$ denote the logical gate represented by
$\hat{\mathcal D}_{\widetilde R}$, so that $d_R^8=-1$.  In the
$U(1)$-pushout choose the rephased lift $v_R:=e^{i\pi/8}d_R$.  Then
$v_R^8=e^{i\pi}d_R^8=1$, while $v_R$ still projects to $R$ and therefore has
order exactly eight.  Since
$G_{\mathcal S}=\langle R\rangle\simeq\mathbb Z_8$, the assignment
\[
    s:G_{\mathcal S}\longrightarrow
    \mathcal D_{\mathcal S}\times_Z U(1),
    \qquad
    s(R^j):=v_R^j,
\]
is a homomorphism.  Thus the odd-characteristic obstruction is
confined to the sign extension: it therefore disappears in the $U(1)$-pushout.  The even characteristic below
instead yields a commutator obstruction, which no central rephasing can
alter.
\end{remark}

\subsection{An even-spin obstruction that survives rephasing}
\label{sec:even-spin-commutator}

The odd characteristic considered above gives a non-split extension by the
sign $Z=\{\pm1\}$, but its order obstruction disappears once arbitrary scalar
phases are allowed.  We now exhibit a stronger phenomenon.  For a suitable
even theta characteristic, two commuting classical logical gates have lifts
with commutator $-1$.  Since the group commutator is unchanged
by central rephasing, no enlargement of the allowed phase group linearizes
the action.  The natural place to detect this obstruction is the theta group,
which is the maximal scalar enlargement of the logical gate group furnished
by Theorem~\ref{fw:thm:theta}.

\begin{lemma}[Maximal scalar test]
\label{lem:maximal-scalar-test}
Let $A$ be any group of scalars with
\begin{equation}
    Z\subseteq A\subseteq\mathbb C^\times,
\end{equation}
and consider the pushout extension
\begin{equation}
    1\longrightarrow A
    \longrightarrow \mathcal D_{\mathcal S}\times_Z A
    \longrightarrow G_{\mathcal S}
    \longrightarrow 1.
    \label{eq:A-pushout-extension}
\end{equation}
If \eqref{eq:A-pushout-extension} splits, then the theta extension
\begin{equation}
    1\longrightarrow\mathbb C^\times
    \longrightarrow G(\mathbb L_{\mathcal S})
    \longrightarrow G_{\mathcal S}
    \longrightarrow1
    \label{eq:theta-extension-even-spin}
\end{equation}
splits.  Consequently, if the theta extension does not split, then neither
the original logical extension nor any intermediate scalar pushout
\eqref{eq:A-pushout-extension} splits.
\end{lemma}

\begin{proof}
The inclusion $A\hookrightarrow\mathbb C^\times$ induces a homomorphism
\begin{equation}
    j_A:\mathcal D_{\mathcal S}\times_Z A
    \longrightarrow
    \mathcal D_{\mathcal S}\times_Z\mathbb C^\times,
    \qquad
    [d,a]_A\longmapsto[d,a]_{\mathbb C^\times}.
    \label{eq:phase-pushout-inclusion}
\end{equation}
It is well defined because the same relation
$[zd,a]=[d,za]$, $z\in Z$, defines both pushouts, and it commutes with the
projections to $G_{\mathcal S}$.  If
$s_A:G_{\mathcal S}\to\mathcal D_{\mathcal S}\times_Z A$ is a splitting, then
\begin{equation}
    \Theta\circ j_A\circ s_A:
    G_{\mathcal S}\longrightarrow G(\mathbb L_{\mathcal S})
\end{equation}
is a splitting of \eqref{eq:theta-extension-even-spin}, where
$\Theta$ is the isomorphism of Theorem~\ref{fw:thm:theta}.
\end{proof}

It therefore suffices to find one half-integral-weight code whose theta
extension does not split.\\

\emph{The even characteristic.---}
We now specialize to a theta characteristic for which the stronger
obstruction occurs.  Write the Bolza curve in hyperelliptic form
\[
    C:\qquad v^2=t^5-t,
\]
and denote by $w_0,w_1,w_i$ the Weierstrass points lying over
$t=0,1,i$, respectively.  Set
\[
    \delta:=w_0+w_1-w_i,
    \qquad
    \mathcal E:=\mathcal O_C(\delta).
\]
This is one of the even theta characteristics of $C$.  To see directly
that it is a theta characteristic, consider the meromorphic differential
\[
    \eta_\delta
    :=
    \frac{t(t-1)}{t-i}\frac{dt}{v}.
\]
Using
\begin{equation}
    \operatorname{div}(t-a)=2w_a-2w_\infty,
    \qquad
    \operatorname{div}\!\left(\frac{dt}{v}\right)=2w_\infty,
    \label{eq:even-divisor-t}
\end{equation}
and
\begin{equation}
    \operatorname{div}v
    =
    w_0+w_1+w_{-1}+w_i+w_{-i}-5w_\infty,
    \label{eq:even-divisor-v}
\end{equation}
which follows from $v^{2}=t^{5}-t$ and \eqref{eq:even-divisor-t},
one finds
\[
    \operatorname{div}\eta_\delta
    =
    2w_0+2w_1-2w_i
    =
    2\delta.
\]
Since $\eta_\delta$ is a nonzero meromorphic section of $\Omega_C$, its
divisor determines the canonical bundle,
\[
    \Omega_C
    \simeq
    \mathcal O_C(\operatorname{div}\eta_\delta)
    =
    \mathcal O_C(2\delta)
    =
    \mathcal E^{\otimes 2}.
\]
Thus $\mathcal E$ is indeed a theta characteristic.  It is
\emph{even}, i.e.\ $h^{0}(\mathcal E)=0$: were $w_0+w_1-w_i$
linearly equivalent to an effective divisor---necessarily a single
point $p$---then $w_0+w_1\sim w_i+p$ would give
$h^{0}\bigl(\mathcal O_C(w_0+w_1)\bigr)\geq2$, placing $w_0+w_1$ in
the hyperelliptic pencil; but the divisors of the pencil are the
fibers of $t$, and the fiber through a Weierstrass point $w_a$ is
$2w_a$, not $w_0+w_1$.

For $g\in G_{\mathcal S}$, choose a meromorphic function
$\rho_g$ satisfying
\begin{equation}
    \rho_g^2=\frac{g^*\eta_\delta}{\eta_\delta},
    \qquad
    \operatorname{div}\rho_g=g^*\delta-\delta.
    \label{eq:normalized-even-spin-lift}
\end{equation}
Such a function exists and is unique up to sign.  Indeed,
$g^*\mathcal E\simeq\mathcal E$ means that $g^*\delta-\delta$ is principal,
whereas the ratio in \eqref{eq:normalized-even-spin-lift} has divisor
$2(g^*\delta-\delta)$; after dividing by the square of a function with this
divisor, only a nonzero constant remains.  Multiplication by $\rho_g$ then
defines an isomorphism
\begin{equation}
    u_g:g^*\mathcal E\overset{\sim}{\longrightarrow}\mathcal E.
    \label{eq:normalized-bundle-lift}
\end{equation}
Indeed, if $h$ is a local meromorphic section of
$g^*\mathcal E\simeq\mathcal O_C(g^*\delta)$, then
$\operatorname{div}(\rho_g h)+\delta
=\operatorname{div}h+g^*\delta\geq0$. The first identity in \eqref{eq:normalized-even-spin-lift} says precisely
that $u_g^{\otimes2}$ agrees, under
$\mathcal E^{\otimes2}\simeq\Omega_C$, with the canonical pullback
isomorphism $g^*\Omega_C\simeq\Omega_C$.

\emph{A commuting pair with anticommuting spin lifts.---}
Besides $\iota$, consider
\begin{equation}
    \tau(t,v)
    =\left(-\frac{i}{t},\,e^{i\pi/4}\frac{v}{t^3}\right).
    \label{eq:even-tau}
\end{equation}
Direct substitution gives
\begin{equation}
    \tau^2=1,
    \qquad
    \tau\iota=\iota\tau.
    \label{eq:tau-iota-commute}
\end{equation}
Since $\iota$ fixes each Weierstrass point,
$\iota^*\delta=\delta$ and hence $\iota\in G_{\mathcal S}$.  It
also fixes $t$ and $dt$ and sends $v\mapsto-v$, so
\begin{equation}
    \iota^*\eta_\delta=-\eta_\delta,
    \qquad
    \rho_\iota=\pm i.
    \label{eq:rho-iota}
\end{equation}
For $\tau$, direct pullback gives
\begin{equation}
    \frac{\tau^*\eta_\delta}{\eta_\delta}
    =e^{-3i\pi/4}\frac{t^2+1}{t^3-t}.
    \label{eq:tau-pullback-ratio}
\end{equation}
Because
\begin{equation}
    \left(\frac{v}{t^3-t}\right)^2
    =\frac{t^2+1}{t^3-t},
\end{equation}
we may take
\begin{equation}
    \rho_\tau=c\frac{v}{t^3-t},
    \qquad
    c^2=e^{-3i\pi/4}.
    \label{eq:rho-tau}
\end{equation}
Using \eqref{eq:even-divisor-t}--\eqref{eq:even-divisor-v},
\begin{equation}
    \operatorname{div}\rho_\tau
    =w_i+w_{-i}+w_\infty-w_0-w_1-w_{-1}
    =\tau^*\delta-\delta.
    \label{eq:rho-tau-divisor}
\end{equation}
Thus $\tau^*\mathcal E\simeq\mathcal E$, so
$\tau\in G_{\mathcal S}$.  Moreover, \eqref{eq:rho-tau} is $v$ times a
rational function of $t$, and hence
\begin{equation}
    \iota^*\rho_\tau=-\rho_\tau.
    \label{eq:rho-tau-odd}
\end{equation}
Combining \eqref{eq:tau-iota-commute}, \eqref{eq:rho-iota}, and
\eqref{eq:rho-tau-odd} gives
\begin{equation}
    u_\iota\circ\iota^*u_\tau
    =-u_\tau\circ\tau^*u_\iota.
    \label{eq:spin-bundle-anticommutator}
\end{equation}
The sign is unchanged by replacing either normalized lift by its negative.

\emph{Passage to the theta group.---}
Let $k=m+\tfrac12$ with $m\geq1$; the restriction $m\geq1$ only
ensures the code space is nonzero---at $k=\tfrac12$ the code space
$H^{0}(\mathcal E)$ of the even characteristic vanishes, so
statements about the action on code states would be vacuous, though
the bundle-level obstruction below persists.  As established above,
\begin{equation}
    \mathcal L_{\mathcal S}
    \simeq\Omega_C^{\otimes m}\otimes\mathcal E,
    \qquad
    \mathbb L_{\mathcal S}=\mathcal L_{\mathcal S}^{\vee}.
    \label{eq:half-integral-code-and-frame-bundles}
\end{equation}
Let
\begin{equation}
    \kappa_g^{(m)}:
    g^*\Omega_C^{\otimes m}
    \overset{\sim}{\longrightarrow}
    \Omega_C^{\otimes m}
\end{equation}
be the canonical pullback isomorphism for $m$-differentials.  Pullback
commutes with tensor products, so
\begin{equation}
\begin{split}
    a_g:=\kappa_g^{(m)}\otimes u_g:
    g^*\mathcal L_{\mathcal S}
    &\simeq g^*\Omega_C^{\otimes m}\otimes g^*\mathcal E\\
    &\overset{\sim}{\longrightarrow}
      \Omega_C^{\otimes m}\otimes\mathcal E
      \simeq\mathcal L_{\mathcal S}
\end{split}
    \label{eq:code-bundle-lift}
\end{equation}
is well defined.  Functoriality of pullback---equivalently, the chain
rule---gives
\begin{equation}
    \kappa_\iota^{(m)}\circ\iota^*\kappa_\tau^{(m)}
    =\kappa_\tau^{(m)}\circ\tau^*\kappa_\iota^{(m)},
    \label{eq:canonical-factor-commutes}
\end{equation}
so the integer-weight factor contributes trivial commutator to the commuting
pair $\tau,\iota$.  Tensoring
\eqref{eq:spin-bundle-anticommutator} with
\eqref{eq:canonical-factor-commutes} gives
\begin{equation}
    a_\iota\circ\iota^*a_\tau
    =-a_\tau\circ\tau^*a_\iota.
    \label{eq:code-bundle-anticommutator}
\end{equation}
Dualizing produces the morphisms in the convention of the theta group,
\begin{equation}
    \Phi_g:=a_g^\vee:
    \mathbb L_{\mathcal S}
    \overset{\sim}{\longrightarrow}
    g^*\mathbb L_{\mathcal S},
    \qquad
    x_g:=(g,\Phi_g)\in G(\mathbb L_{\mathcal S}).
    \label{eq:frame-bundle-lift}
\end{equation}
The dual of \eqref{eq:code-bundle-anticommutator} is
\begin{equation}
    \iota^*\Phi_\tau\circ\Phi_\iota
    =-\tau^*\Phi_\iota\circ\Phi_\tau.
    \label{eq:frame-bundle-anticommutator}
\end{equation}
Using the theta-group law
\begin{equation}
    (g,\Phi)(h,\Psi)
    =(gh,h^*\Phi\circ\Psi)
\end{equation}
and the commutativity of $\tau$ and $\iota$, this is exactly
\begin{equation}
    x_\tau x_\iota=-x_\iota x_\tau,
    \qquad
    [x_\tau,x_\iota]=-1\in\mathbb C^\times.
    \label{eq:theta-commutator-minus-one}
\end{equation}
Here $-1$ is the central scalar automorphism of
$\mathbb L_{\mathcal S}$.

\begin{theorem}[Even-spin commutator obstruction]
\label{thm:even-spin-commutator}
For the even characteristic
$\mathcal E=\mathcal O_C(w_0+w_1-w_i)$ and every
$k=m+\tfrac12$ with $m\geq1$, the theta extension
\eqref{eq:theta-extension-even-spin} does not split.  More strongly, for
every subgroup $A$ with
$Z\subseteq A\subseteq\mathbb C^\times$, the scalar pushout
\eqref{eq:A-pushout-extension} does not split.  In particular, the logical
extension
\begin{equation}
    1\longrightarrow Z
    \longrightarrow\mathcal D_{\mathcal S}
    \longrightarrow G_{\mathcal S}
    \longrightarrow1
\end{equation}
is non-split, and no scalar rephasing linearizes the projective representation of
$G_{\mathcal S}$ on the code space.
\end{theorem}

\begin{proof}
Suppose that the theta extension admitted a splitting
$s:G_{\mathcal S}\to G(\mathbb L_{\mathcal S})$.  Since $s(g)$ and $x_g$
lie over the same classical element, the fiber of the theta extension gives
scalars $\lambda_\tau,\lambda_\iota\in\mathbb C^\times$ such that
\begin{equation}
    s(\tau)=\lambda_\tau x_\tau,
    \qquad
    s(\iota)=\lambda_\iota x_\iota.
\end{equation}
Central rephasing does not change a commutator, whereas a homomorphism carries
a commuting pair to a commuting pair.  Therefore
\begin{equation}
    1=s([\tau,\iota])
     =[s(\tau),s(\iota)]
     =[x_\tau,x_\iota]
     =-1,
\end{equation}
a contradiction.  Thus the theta extension does not split.  The assertion
for every intermediate phase group $A$, including $A=\mathrm{U}(1)$, now follows from
Lemma~\ref{lem:maximal-scalar-test}.
\end{proof}

The distinction from the odd characteristics is therefore intrinsic.  There
the obstruction is carried by the order of a single lift and can be removed
by scalar rephasing; here it is carried by the commutator of a commuting pair,
which survives any rephasing by scalars.

\section{A lemma on short exact sequences}\label{fw:ses}

We record here the elementary group-theoretic lemma used in the
proof of Theorem~\ref{fw:thm:main}: it passes a short exact
sequence to quotients.

\begin{lemma}[Quotient of a short exact sequence]\label{fw:lem:quotient}
Let $1\to A\to B\xrightarrow{\pi}C\to1$ be exact and let $D\lhd B$.  Then
\begin{equation}
  1 \longrightarrow A/(A\cap D) \longrightarrow B/D
    \xrightarrow{\ \bar\pi\ } C/\pi(D) \longrightarrow 1
\end{equation}
is exact.
\end{lemma}

\begin{proof}
$\bar\pi$ is a well-defined surjective homomorphism: $c\,\pi(D)=\bar\pi(\tilde
c\,D)$ for any lift $\tilde c$ of $c$.  For the kernel, $\bar\pi(bD)=\pi(D)$
iff $\pi(b)\in\pi(D)$, i.e.\ $b\in\pi^{-1}(\pi(D))=AD$ (using $\ker\pi=A$).
Thus $\ker\bar\pi=AD/D\cong A/(A\cap D)$ by the second isomorphism theorem.
\end{proof}

\bibliography{prl-sm-refs}